\documentclass[11pt,a4]{article}
\usepackage{amsfonts}
\usepackage{graphics,amssymb,amsthm,amsmath,epsf}
\usepackage{graphicx,rotating}
\usepackage{subfigure}
\usepackage[round]{natbib}

\newtheorem{theorem}{Theorem}[section]
\newtheorem{proposition}{Proposition}

\numberwithin{equation}{section}
\newtheorem{remark}{Remark}[section]

\begin{document}

\date{}
\title{Gompertz - Power Series Distributions}
\author{A. A. Jafari$^{1,}$\thanks{E-mail: aajafari@yazd.ac.ir},  S. Tahmasebi$^2$ \\
%EndAName
{\small $^{1}$Department of Statistics, Yazd University, Yazd, Iran}\\
{\small $^{2}$Department of Statistics, Persian Gulf University, Bushehr, Iran}\\
}
\date{}
\maketitle
\begin{abstract}
In this paper, we introduce the Gompertz power series (GPS) class of distributions which is obtained by compounding Gompertz and power series distributions. This distribution contains several lifetime models such as Gompertz-geometric (GG), Gompertz-Poisson (GP), Gompertz-binomial (GB), and
 Gompertz-logarithmic (GL) distributions as special cases. Sub-models of the GPS distribution are studied in details. The hazard rate function of the GPS distribution can be increasing, decreasing, and bathtub-shaped. We obtain several properties of the GPS distribution such as its probability density function, and failure rate function, Shannon entropy, mean residual life function, quantiles and moments. The maximum likelihood estimation procedure via a EM-algorithm is presented, and simulation studies are performed for evaluation of this estimation for complete data, and the MLE of parameters for censored data. At the end, a real example is given.

\end{abstract}


\section{Introduction}

The  exponential distribution is commonly used in many applied problems, particularly in lifetime data analysis. A generalization of this distribution is the Gompertz distribution. It is a lifetime distribution and is often applied to describe the distribution of adult life spans by actuaries and demographers. In some sciences such as  biology, gerontology, computer, and marketing science,  the Gompertz distribution is considered for the analysis of survival.

A random variable $X$ is said to have a Gompertz distribution, denoted by $X\sim G(\beta,\gamma)$, if its
cumulative distribution function (cdf) is
\begin{eqnarray}\label{1e}
G(x)=1-e^{-\frac{\beta}{\gamma}(e^{\gamma x}-1)}, \;\;x\geq0, \;\; \beta>0, \;\;\gamma>0,
\end{eqnarray}
and the probability density function (pdf) is
\begin{eqnarray}\label{2e}
g(x)= \beta e^{\gamma x}e^{-\frac{\beta}{\gamma}(e^{\gamma x}-1)}.
\end{eqnarray}

The Gompertz distribution is a flexible distribution that can be skewed to the right and to the left. The hazard rate function of Gompertz distribution is $h_g(x)=\beta e^{\gamma x}$ which is a increasing function. The  exponential distribution can be derived from the  Gompertz distribution  when $\gamma\rightarrow0^+$.

Also, a discrete random variable, $N$  is a member of power series distributions (truncated at zero) if its probability mass function is given by
\begin{eqnarray}
P(N=n)=\frac{a_{n}\theta^{n}}{C(\theta)},\;\; n=1,2,...,
\end{eqnarray}
where $a_{n}\geq0$, $C(\theta)=\sum\limits_{n=1}^{\infty}a_{n}\theta^{n}$, and $\theta\in(0,s)$ is chosen such that $C(\theta)$ is finite and its first, second and third derivatives are defined and shown by $C'(.)$, $C''(.)$ and $C'''(.)$. The term "power series distribution" is generally credited to
\cite{noack-50}.
% Noack (1950).
 This family of distributions includes many of the most common distributions, including the binomial, Poisson, geometric, negative binomial, logarithmic distributions. For more details of power series distributions, see \cite{jo-ke-ko-05},
%Johnson et al. (2005),
page 75.

In this paper, we compound the Gompertz  and power series distributions and  introduce a new class of distribution. This procedure follows similar way that was previously carried out by some authors: The exponential-power series distribution is introduced by \cite{ch-ga-09},
%Chahkandi and Ganjali (2009)
which is included the exponential-geometric \citep{ad-lo-98,ad-di-05},
 %(Adamidis and Loukas, 1998, Adamidis et al., 2005),
exponential-Poisson \citep{kus-07},
%(Kus, 2007),
and exponential-logarithmic \citep{ta-re-08}
 %(Tahmasbi and Rezaei, 2008)
 distributions; the Weibull-power series distributions is introduced by \cite{mo-ba-11}
 %Morais and Barreto-Souza (2011)
   which is a generalization of the exponential-power series distribution;
the generalized exponential-power series distribution is introduced by \cite{ma-ja-12}
%Mahmoudi and Jafari (2012)
which is included the Poisson-exponential \citep{ca-lo-fr-ba-11},
 %(Cancho et al., 2011),
 complementary exponential-geometric \citep{lo-ro-ca-11},
 %(Louzada-Neto et al., 2011),
 and the complementary exponential-power series \citep{fl-bo-ca-11}
 %(Flores et al., 2011)
  distributions.

 The remainder of our paper is organized as follows: in Section \ref{sec.GPS}, we give the density and failure rate functions of the GPS distribution. Some properties
such as quantiles, moments, order statistics, Shannon entropy and mean residual life  are given in Section \ref{sec.pro}. Special cases of GPS distribution are given in Section \ref{sec.spe}.
We discuss estimation by maximum likelihood and provide an expression for Fisher's information matrix in Section \ref{sec.est}. In this Section, we present the estimation based on EM-algorithm, and
Section \ref{sec.sim} contains Monte Carlo simulation results on the finite sample behavior of these estimators. In this Section, we also investigate the properties of MLE of parameters when the data are censored. An application of GPS distribution is given in the Section \ref{sec.ex}.
%Finally, Section \ref{sec.con} concludes the paper.

\section{The Gompertz-power series model}

\label{sec.GPS}
The GPS model is derived as follows. Let $N$ be a random variable denoting the number of failure causes which it is a member of power series distributions (truncated at zero). For given $N$, let $X_{1},X_{2},...,X_{N}$ be independent identically distributed random variables from Gompertz distribution. If we consider
$X_{(1)}=\min(X_1,...,X_N)$, then $X_{(1)}\mid N=n$  has Gompertz distribution with parameters  $n\beta$ and $\gamma$. Therefore,
the GPS class of distributions, denoted by $GPS(\beta,\gamma,\theta)$, is defined by
 \begin{eqnarray}\label{FGP}
F(x)=1-\frac{C(\theta-\theta G(x))}{C(\theta)}=1-\frac{C(\theta e^{-\frac{\beta}{\gamma}(e^{\gamma x}-1)})}{C(\theta)}, \ \ \ \ \ x>0.
\end{eqnarray}
The pdf of $GPS(\beta,\gamma,\theta)$ is given by
\begin{eqnarray}\label{fGP}
f(x)=\theta g(x) \frac{{C'}(\theta-\theta G(x))}{C(\theta)}=\theta \beta e^{\gamma x}e^{-\frac{\beta}{\gamma}(e^{\gamma x}-1)} \frac{C'(\theta e^{-\frac{\beta}{\gamma}(e^{\gamma x}-1)})}{C(\theta)}.
\end{eqnarray}

\bigskip
\begin{proposition}
If $C(\theta)=\theta$, then the Gompertz distribution function  concludes from the GPS distribution function in \eqref{FGP}. Therefore, the Gompertz distribution is a special case of GPS distribution.
\end{proposition}

\begin{proposition}
The limiting distribution of $GPS(\beta,\gamma,\theta)$ when $\theta\rightarrow 0^{+}$ is
\[ {\mathop{\lim }_{\theta \rightarrow 0^{+}}  F(x)}=1-e^{\frac{-c\beta}{\gamma}(e^{\gamma x}-1)},\]
which is a $G(c\beta,\gamma)$, where $c=\min\{n\in N: a_{n}>0\}$.
\end{proposition}

\begin{proposition}
 The limiting distribution of $GPS(\beta,\gamma,\theta)$ when $\gamma\longrightarrow0^+$ is
\[ {\mathop{\lim }_{\gamma \rightarrow 0^{+}}  F(x)}=1-\frac{C(\theta e^{-\beta x})}{C(\theta)}. \]
In fact, it  is the cdf of the exponential-power series (EPS) distribution and is introduced by
\cite{ch-ga-09}.
%Chahkandi and Ganjali (2009).
This distribution contains several distributions; geometric-exponential distribution
\citep{ad-lo-98,ad-di-05},
%(Adamidis and Loukas, 1998),
 Poisson-exponential distribution
 \citep{kus-07},
 %(Kus, 2007),
  and logarithmic-exponential distribution
  \citep{ta-re-08}.
  %(Tahmasbi and Rezaei, 2008).
  Therefore, the GPS distribution is a generalization of EPS distribution.
   Note that EPS distribution is a distribution family with decreasing failure rate (hazard rate).
\end{proposition}

\begin{proposition}
 The densities of GPS class can be expressed as infinite linear combination of density of order distribution, i.e. it can be written as
\begin{eqnarray}
f(x)=\sum\limits_{n=1}^{\infty} P(N=n) \ g_{(1)}(x;n),
\end{eqnarray}
where $g_{(1)}(x;n)$ is the pdf of $Y_{(1)}=\min(Y_{1},Y_{2},...,Y_{n})$, given by
 $$g_{(1)}(x;n)=n g(x)[1-G(x)]^{n-1}=n \beta e^{\gamma x}e^{\frac{-n\beta}{\gamma}(e^{\gamma x}-1)},$$
 i.e. Gompertz distribution with parameters $n\beta$ and $\gamma$.
 \end{proposition}

%\begin{proposition}
% The Laplace transform of the GPS class can be expressed as
%\begin{eqnarray}
%L(s)=E(e^{-sX})=\sum\limits_{n=1}^{\infty} P(N=n)L_{1}(s),
%\end{eqnarray}
%where $L_{1}(s)$ is the Laplace transform of $g_{(1)}(x;n)$ which is obtained as
% $$ L_{1}(s)=\frac{n\beta}{\gamma}e^{\frac{n\beta}{\gamma}}W_{\frac{s}{\gamma}}(\frac{n\beta}{\gamma}), $$
%where $W_{f}(z)=\int_{1}^{\infty}\frac{e^{-zu}}{u^{f}}du$.  \citep[see][]{lenart-12}
%%(Adam Lenart 2012)
%%(Abramowitz and Stegun 1965:5.1.4).
% \end{proposition}

\begin{proposition}
 The survival function and the hazard rate function of the GPS class of distributions, are given respectively by
\begin{eqnarray}\label{hGP}
S(x)=\frac{C(\theta e^{-\frac{\beta}{\gamma}(e^{\gamma x}-1)})}{C(\theta)}, \ \  \ \ \ \  \ \
 h(x)=\theta \beta e^{\gamma x}e^{-\frac{\beta}{\gamma}(e^{\gamma x}-1)} \frac{C'(\theta e^{-\frac{\beta}{\gamma}(e^{\gamma x}-1)})}{C(\theta e^{-\frac{\beta}{\gamma}(e^{\gamma x}-1)})}.
 \end{eqnarray}
 \end{proposition}

\begin{proposition}
 For the pdf in \eqref{fGP} we have
 \begin{eqnarray*}
\lim_{x\rightarrow 0^+}f(x)= \dfrac{\beta\theta C'(\theta)}{ C(\theta)}=\beta E(N), \ \ \ \  \lim_{x\rightarrow +\infty}f(x)= 0.
\end{eqnarray*}
\end{proposition}

\begin{proposition}
 For the hazard rate function, $h(x)$,  in \eqref{hGP} we have
 \begin{eqnarray*}
\lim_{x\rightarrow 0^+}h(x)=\lim_{x\rightarrow 0^+}f(x)=\dfrac{\beta\theta C'(\theta)}{ C(\theta)} , \ \ \ \  \lim_{x\rightarrow +\infty}h(x)=+\infty.
\end{eqnarray*}
\end{proposition}

\bigskip

Consider $C\left(\theta \right)=\theta +{\theta }^{20}$. Therefore, the pdf of GPS distribution is given as
 $$
 f\left(x\right)= \beta e^{\gamma x}e^{-\frac{\beta}{\gamma}(e^{\gamma x}-1)}(1+20\theta^{19} e^{-\frac{19\beta}{\gamma}(e^{\gamma x}-1)}) (1 +{\theta }^{19})^{-1}.
$$
The plots of this density and its hazard rate function, for some parameters are given in Figure \ref{dhbim}. For $ \beta=0.1, \gamma=3, \theta=1.0$, this density is bimodal, and the values  of modes are 0.1582 and 1.1505.

\begin{figure}[]
\centering
\includegraphics[scale=0.5]{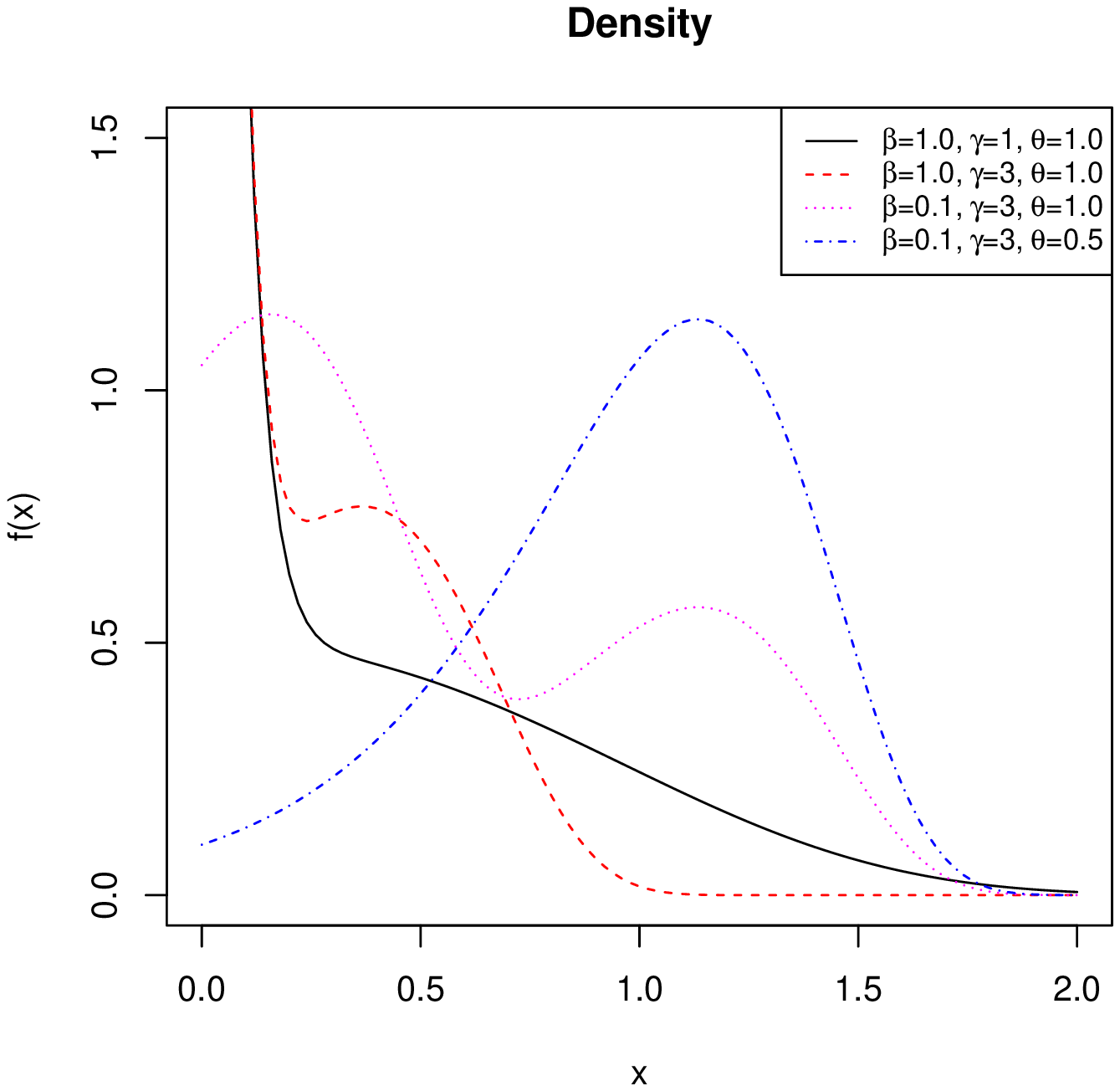}
\includegraphics[scale=0.5]{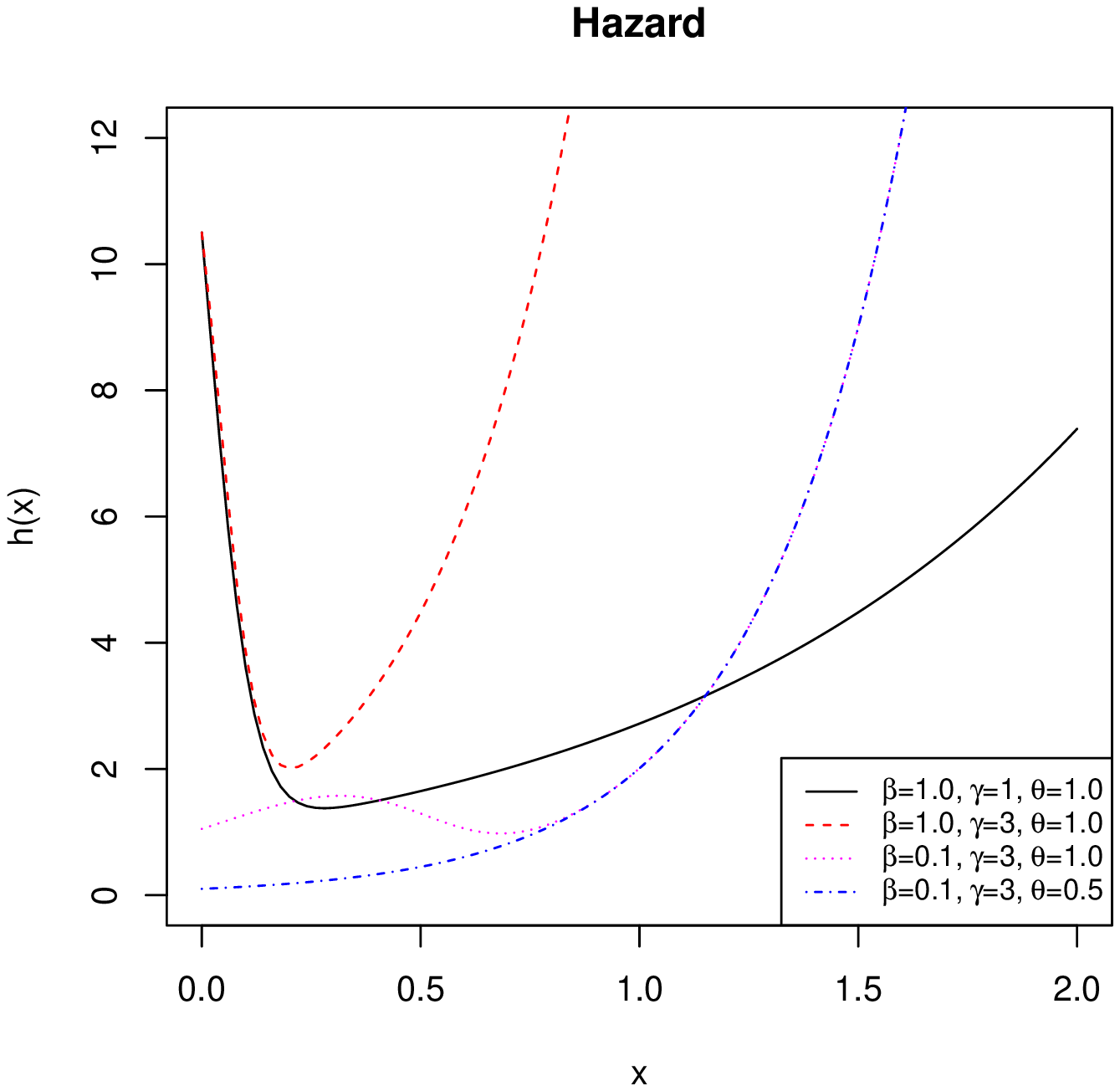}
\vspace{-0.8cm}
\caption[]{Plots of pdf and hazard rate functions of GPS with $C\left(\theta \right)=\theta +{\theta }^{20}$. } \label{dhbim}
\end{figure}

\section{Statistical properties}
\label{sec.pro}
In this section, some properties of the GPS distribution,
such as quantiles, moments, order statistics, Shannon entropy and mean residual life are obtained.
\subsection{Quantiles and Moments}
The quantile $q$ of GPS distribution is given by
$$
x_{q}=G^{-1}\left(1-\frac{1}{\theta}C^{-1}\left((1-q)C(\theta)\right)\right),\;\;\;\;\;0<q<1,
$$
where $G^{-1}(y)=\frac{1}{\gamma}\log\left(1-\frac{\gamma}{\beta} \log(1-y)\right)$ and $C^{-1}(.)$ is the inverse function of $C(.)$.
This  result helps in simulating data from the GPS distribution with generating  uniform distribution data.

For checking the consistency of the simulating data set form GPS distribution, the histogram for a generated data set with size 100
and the exact GPS density with $C\left(\theta \right)=\theta +{\theta }^{20}$, and  parameters $\beta=0.1$, $\gamma=3$, $\theta=1.0$,  are displayed in Figure \ref{Fig.gd}
(left). Also, the empirical distribution function and the exact distribution function are given in Figure \ref{Fig.gd} (right).

\begin{figure}[ht]
\centering
\includegraphics[scale=0.5]{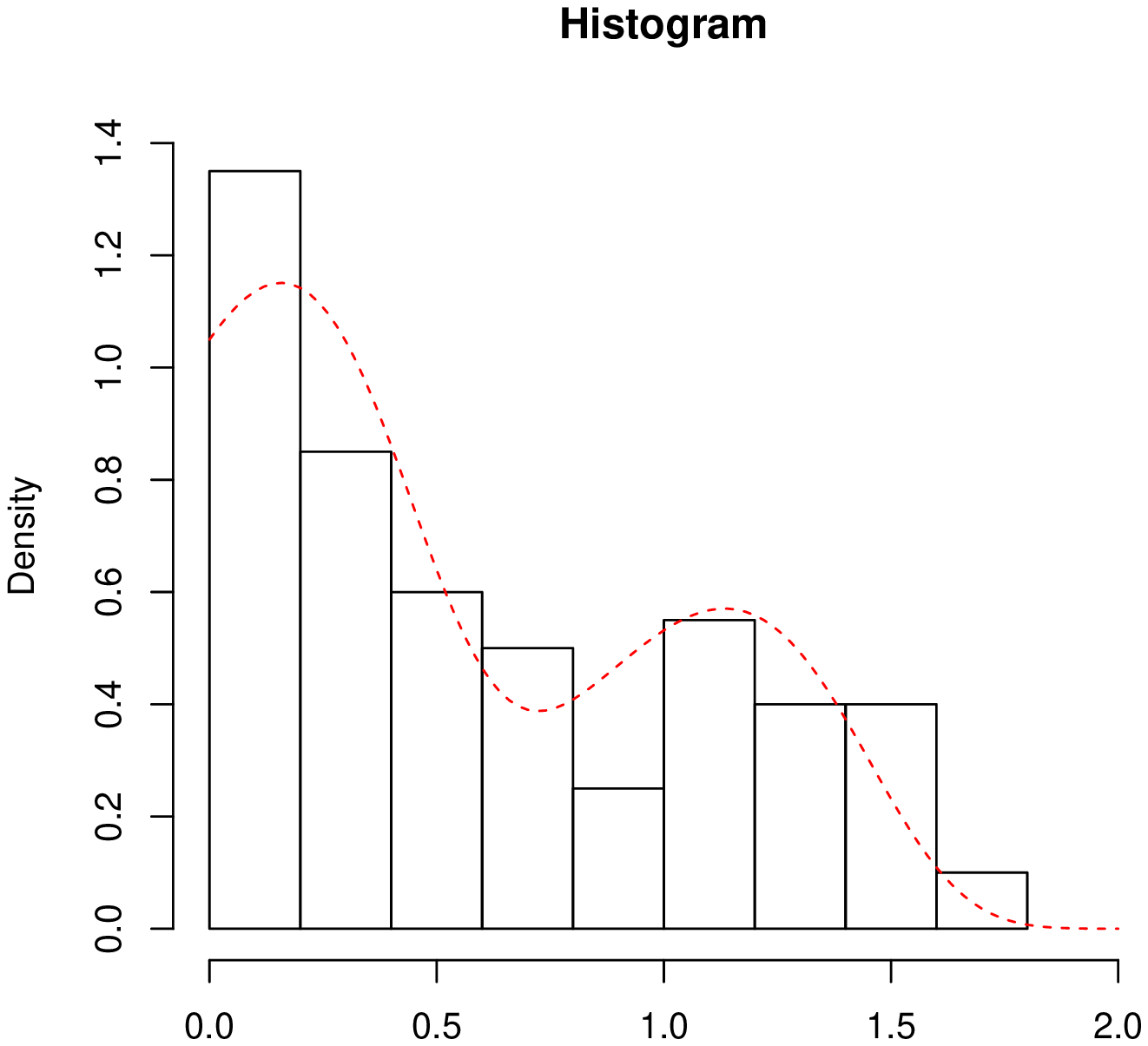}
\includegraphics[scale=0.5]{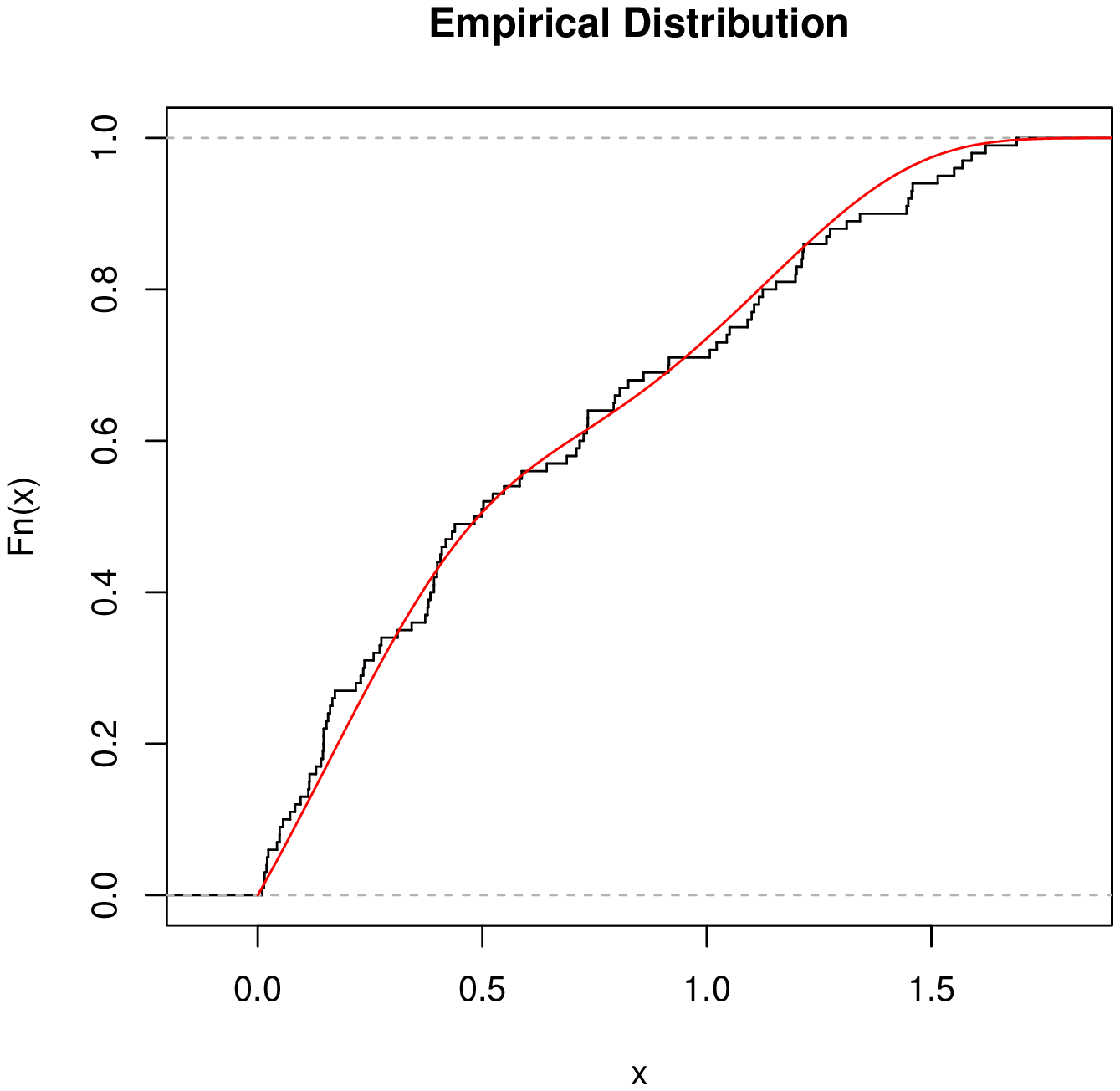}
\vspace{-0.8cm}
\caption[]{The histogram of a generated data set with size 100 and the exact GPS  density (left) and the empirical distribution
function and exact distribution function (right). } \label{Fig.gd}
\end{figure}

Now, we obtain the moment generating function of the GPS distribution by its Laplace transform. Consider $X\sim GPS(\beta,\gamma,\theta)$. Then,
the Laplace transform of the GPS class can be expressed as
\begin{eqnarray}
L(s)=E(e^{-sX})=\sum\limits_{n=1}^{\infty} P(N=n)L_{1}(s),
\end{eqnarray}
where
$$ L_{1}(s)=\frac{n\beta}{\gamma}e^{\frac{n\beta}{\gamma}}W_{\frac{s}{\gamma}}(\frac{n\beta}{\gamma}), $$
is the Laplace transform of Gompertz distribution with parameters $n\beta$ and $\gamma$, and
 $W_{f}(z)=\int_{1}^{\infty}\frac{e^{-zu}}{u^{f}}du$.  \citep[see][]{lenart-12}.
%(Adam Lenart 2012)
%(Abramowitz and Stegun 1965:5.1.4).
Therefore, the moment generating function of the GPS distribution is
\begin{eqnarray}
M_{X}(t)%&=&\sum\limits_{n=1}^{\infty}P(N=n)M_{Y_{(1)}}(t)\nonumber\\
=\sum\limits_{n=1}^{\infty}P(N=n)L_{1}(-t)%\nonumber\\
=\frac{\beta }{\gamma}\sum\limits_{n=1}^{\infty} \frac{ a_{n}\theta^{n}}{C(\theta)}ne^{\frac{n\beta}{\gamma}}W_{\frac{-t}{\gamma}}(\frac{n\beta}{\gamma})%\nonumber\\
=\frac{\beta}{\gamma} E[Ne^{\frac{N\beta}{\gamma}}W_{\frac{-t}{\gamma}}(\frac{N\beta}{\gamma})].
\end{eqnarray}

We can use $M_{X}(t)$ to obtain the central moment functions, $\mu_{r}=E[X^{r} ]$. But from the direct calculation, we have
\begin{eqnarray}
\mu_{r}=\int_{0}^{+\infty} x^{r}f(x) dx=\sum\limits_{n=1}^{\infty}P(N=n)E[Y_{(1)}^{r}],
\end{eqnarray}
where $ E[Y_{(1)}^{r}]$ is the $r$th moment of $Y_{(1)}$, the Gompertz distribution with parameters $n\beta$ and $\gamma$, given by
\cite{lenart-12}
%Adam Lenart(2012)
 as
 \begin{eqnarray}
E[Y_{(1)}^{r}]=\frac{r!}{\gamma^{r}}e^{\frac{n\beta}{\gamma}}W_{1}^{r-1}(\frac{n\beta}{\gamma}),
\end{eqnarray}
where $W_{1}^{r-1}(z)=\frac{1}{(r-1)!}\int_{1}^{\infty}(\ln x)^{r-1}\frac{e^{-zx}}{x}dx$ is the generalised integro-exponential function. % (Milgram 1985).
  See \cite{lenart-12}, for some expressions and approximations about the expected value and variance of Gompertz distribution. For example,
 when $\beta$ is close to $0$, an approximate result for $E[Y_{(1)}]$ %and $Var[Y_{(1)}]$
 is %also obtained by
 %\cite{lenart-12}
% Adam Lenart(2012)
 %as
\begin{eqnarray}
E[Y_{(1)}]\approx \frac{1}{\gamma}e^{\frac{n\beta}{\gamma}}(\frac{n\beta}{\gamma}-\ln(\frac{n\beta}{\gamma})- 0.57722).% \ \ \ \  \ \ \ \
%\end{eqnarray}
%and
%\begin{eqnarray}
%Var[Y_{(1)}]\approx \frac{1}{\gamma^{2}}\frac{\pi^{2}}{6}-\frac{2n\beta}{\gamma^{2}}.
\end{eqnarray}
%where $\Gamma'(1)=$ denotes the Euler constant.
\subsection{Order statistic}

Let $X_{1},X_{2},...,X_{n}$  be a random sample of size $n$ from $GPS(\beta,\gamma,\theta)$, then the pdf of the $i$th order statistic, say $X_{i:n}$, is given by
$$
f_{i:n}(x)=\frac{n!}{(i-1)!(n-i)!}f(x)[1-\frac{C(\theta e^{-\frac{\beta}{\gamma}(e^{\gamma x}-1)})}{C(\theta)}]^{i-1}[\frac{C(\theta e^{-\frac{\beta}{\gamma}(e^{\gamma x}-1)})}{C(\theta)}]^{n-i},
$$
where $f(.)$ is the pdf given by \eqref{fGP}. Also, the cdf of $X_{i:n}$ is given by
$$
F_{i:n}(x)=\frac{n!}{(i-1)!(n-i)!}\sum\limits_{k=0}^{n-i}\frac{(-1)^{k} \dbinom{n-i}{k}}{k+1}[1-\frac{C(\theta e^{-\frac{\beta}{\gamma}(e^{\gamma x}-1)})}{C(\theta)}]^{k+i},
$$
An analytical expression for  $r$th moment of  order statistics $X_{i:n}$ is obtained as
\begin{eqnarray}
 E[X_{i:n}^{r}]&=& \sum\limits_{k=n-i+1}^{n}r(-1)^{k-n+i-1}\dbinom{k-1}{n-i}\dbinom{n}{k}\int_{0}^{+\infty}x^{r-1}S(x)^{k}dx \nonumber\\
 &=& \sum\limits_{k=n-i+1}^{n}\frac{r(-1)^{k-n+i-1}}{[C(\theta)]^{k}}\dbinom{k-1}{n-i}\dbinom{n}{k}\int_{0}^{+\infty}x^{r-1}[C(\theta e^{-\frac{\beta}{\gamma}(e^{\gamma x}-1)})]^{k}dx.
  \end{eqnarray}
\subsection{Shannon entropy and mean residual life }
 If $X$ is a none-negative continuous random variable with pdf $f(x)$, then Shannon's entropy of $X$ is defined  by
\cite{shan-48}
 %Shannon (1948)
as
\begin{equation*}
H(f)=E[-\log f(X)]=-\int_{0}^{+\infty} f(x)\ln (f(x))dx,
\end{equation*}
and this  is usually referred to as the continuous entropy (or differential entropy). An explicit expression of
Shannon entropy for GPS  distribution is obtained as
\begin{eqnarray}
H(f)=-\log(\theta\beta)-\gamma\mu_{1}-\frac{\beta}{\gamma}+\frac{\beta}{\gamma}M_{X}(\gamma)+\log(C(\theta))-E_{N}[A(N,\theta)],
\end{eqnarray}
where $A(N,\theta)=\int_{0}^{1}Nu^{N-1}\log(C'(\theta u))du$. Also,  the mean residual life function of $X$ is given by
$$m(t)=E[X-t|X>t]=\frac{\int_{t}^{+\infty}(x-t)f(x)dx}{S(t)}=\frac{C(\theta)E_N[B(t,N,\beta,\gamma)]}{C(\theta e^{\frac{-\beta}{\gamma}(e^{\gamma x}-1)})}-t,$$
where $B(t,N,\beta,\gamma)=\int_{t}^{+\infty}N\beta xe^{\gamma x}e^{-\frac{N\beta}{\gamma}(e^{\gamma x}-1)}dx$.

\section{Special cases of the GPS distributions}
\label{sec.spe}
In this Section, we consider four special cases of the GPS distribution.

\subsection{ Gompertz - geometric distribution}

\begin{figure}[]
\centering
\includegraphics[scale=0.35]{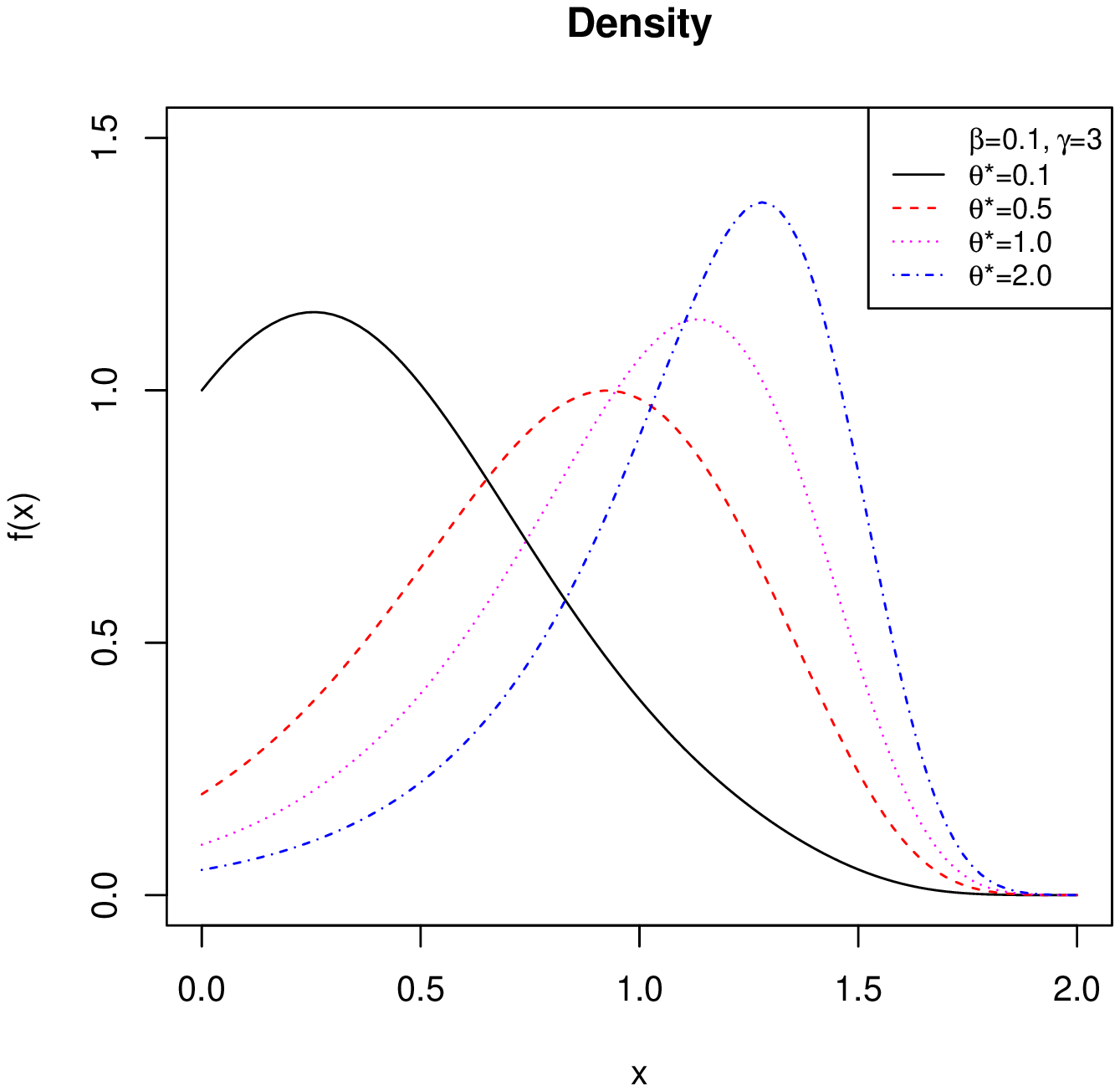}
\includegraphics[scale=0.35]{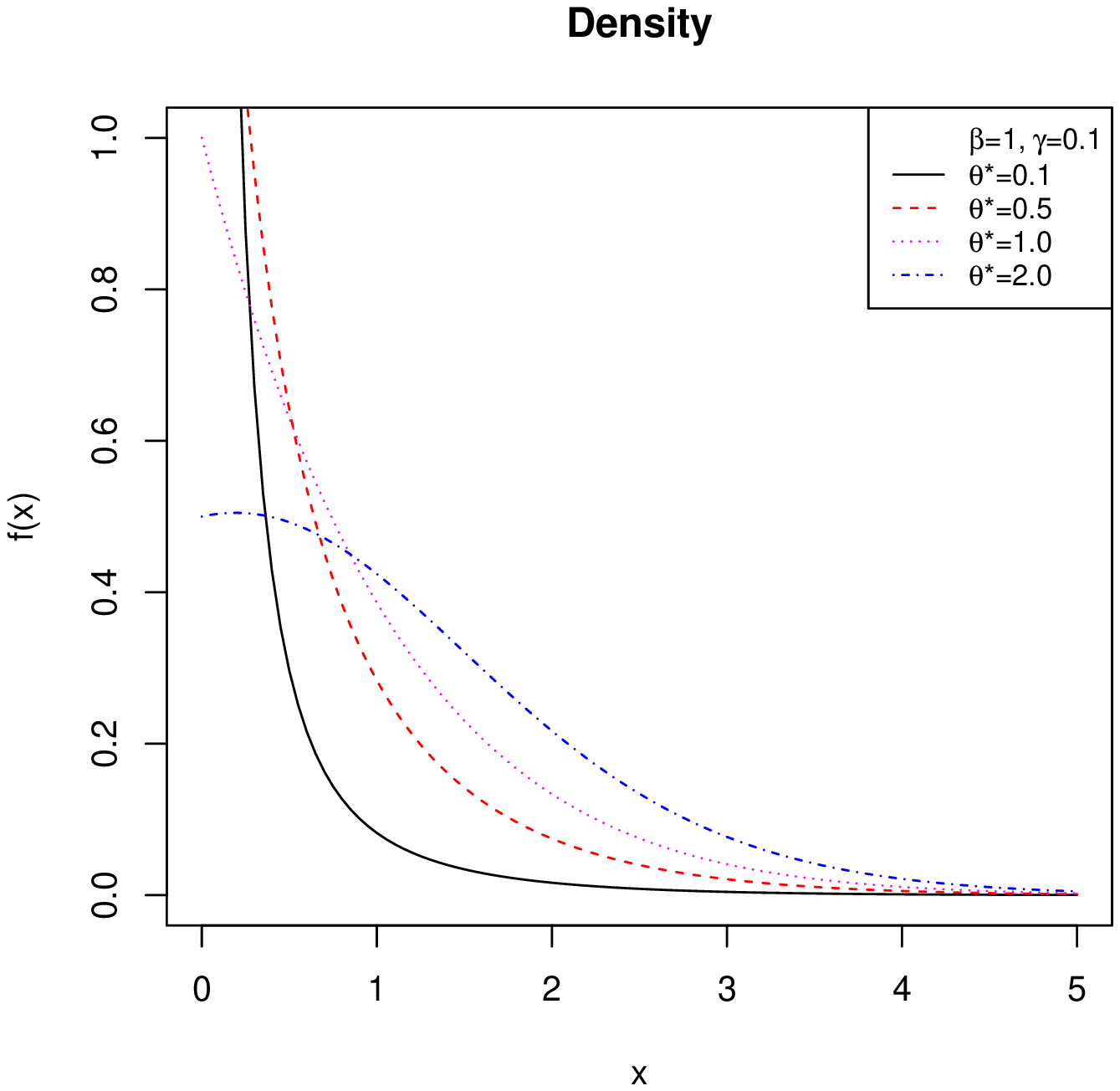}
\includegraphics[scale=0.35]{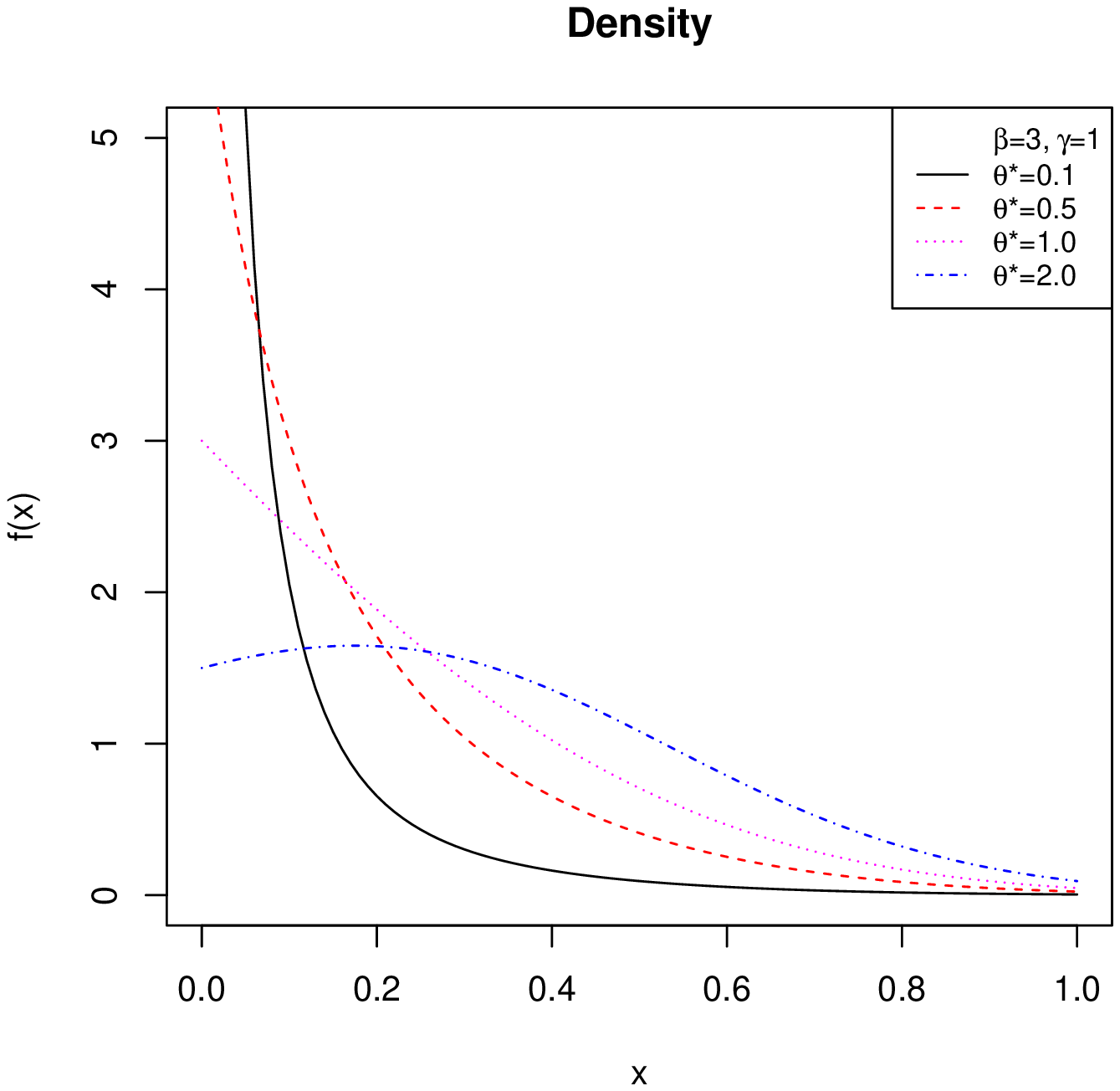}
%\includegraphics[scale=0.35]{denGG9.eps}
%\vspace{-0.8cm}
%\caption[]{Plots of density and hazard rate functions of GG for  different values $\beta $, $\gamma$ and $\theta^*$.}\label{fig.GG}
%\end{figure}
%
%$\ $
%\newpage
%
%\begin{figure}[t]
%\centering
%\includegraphics[scale=0.33]{hazGG1.eps}
\includegraphics[scale=0.35]{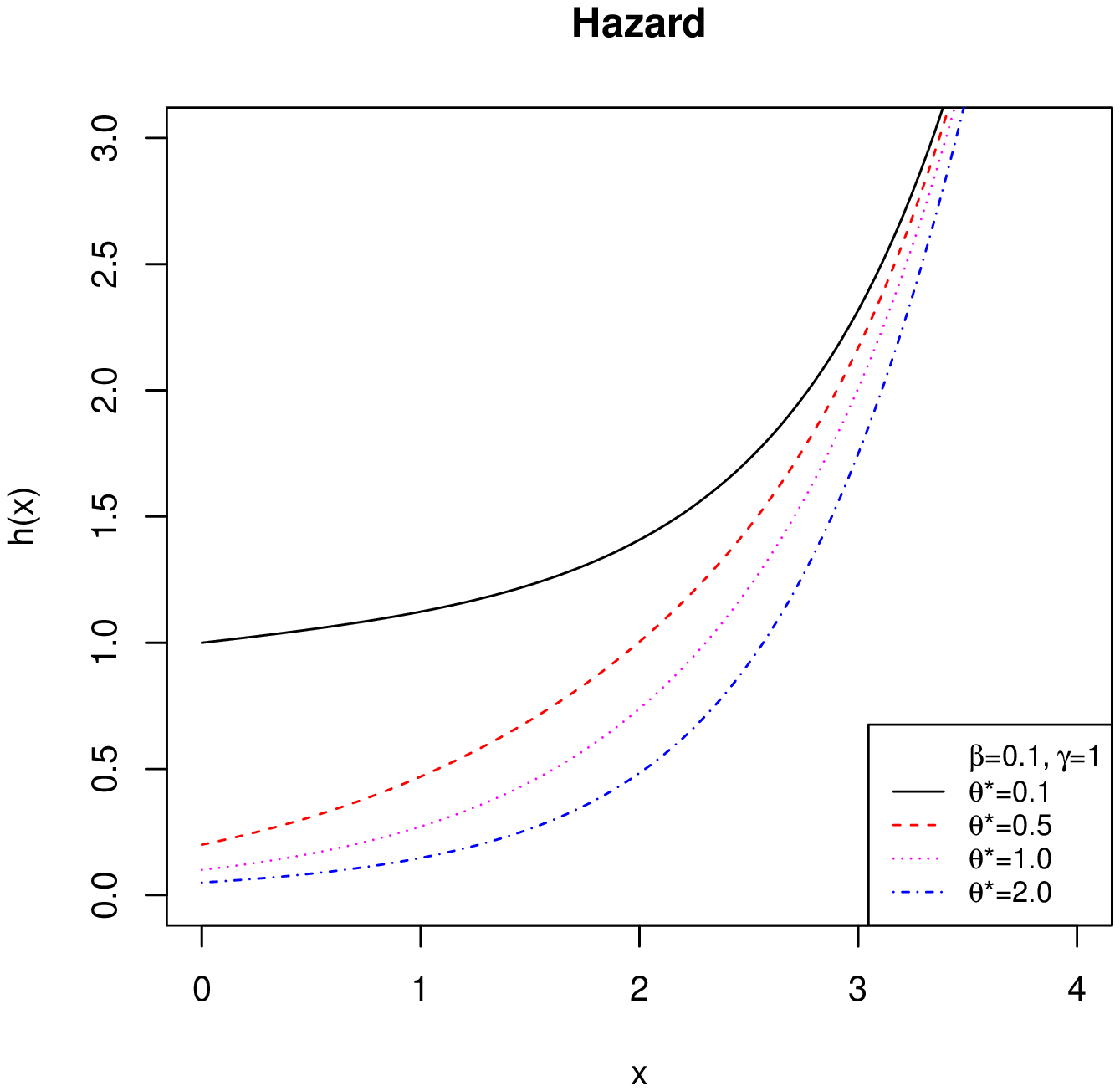}
\includegraphics[scale=0.35]{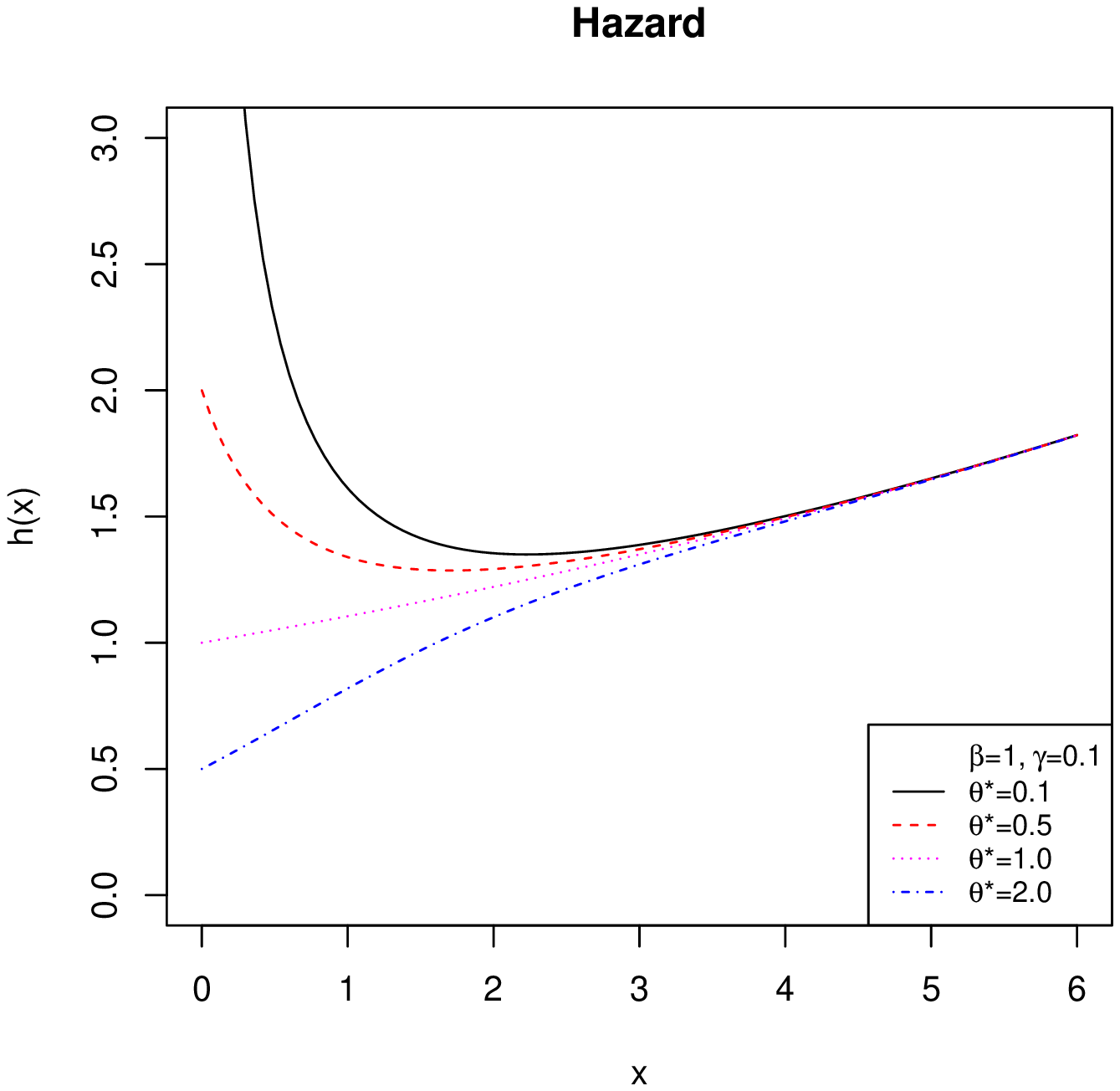}
\includegraphics[scale=0.35]{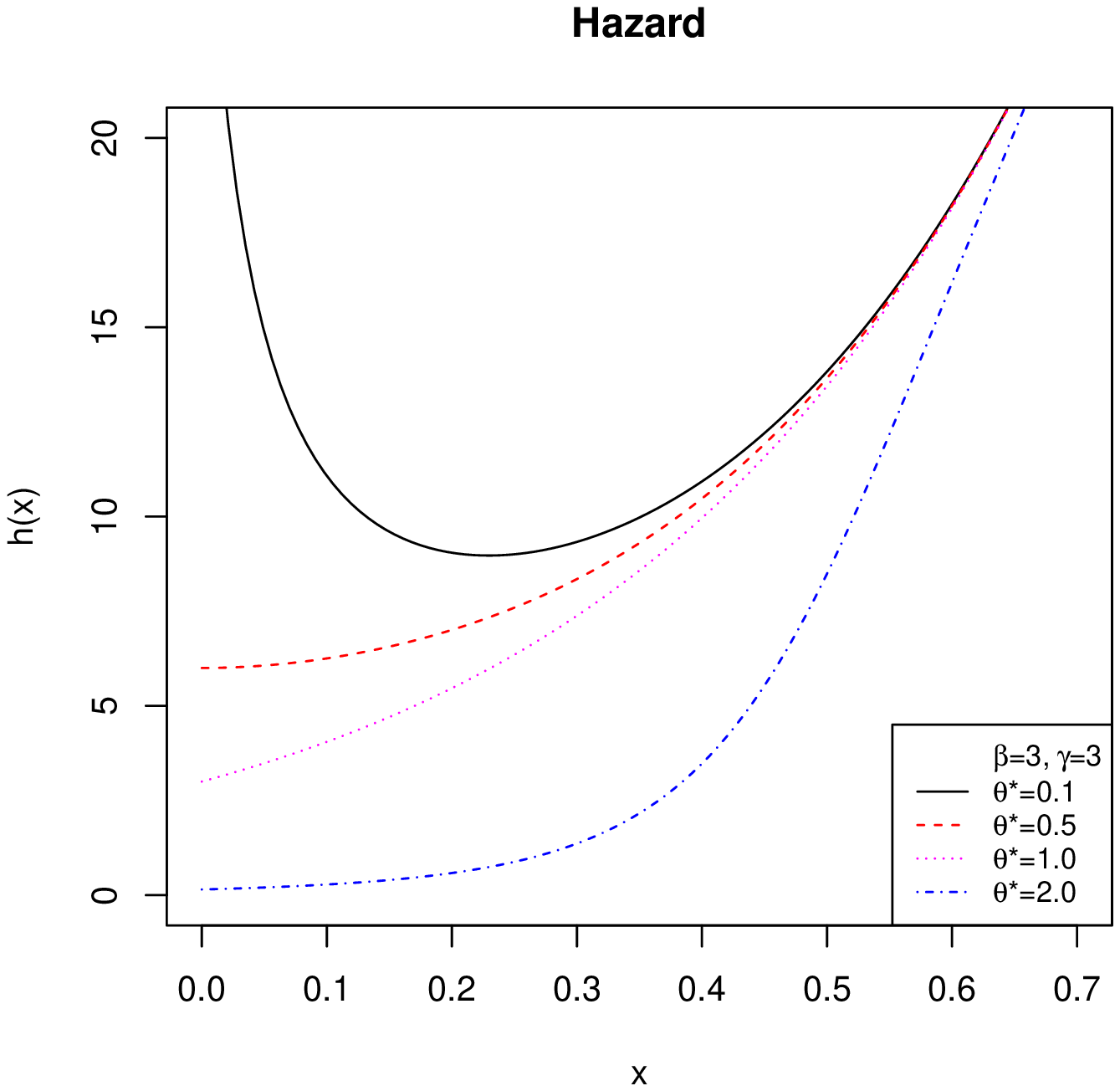}

\vspace{-0.8cm}
\caption[]{Plots of density and hazard rate functions of GG for  different values $\beta $, $\gamma$ and $\theta^*$.}\label{fig.GG}
\end{figure}

The geometric distribution (truncated at zero) is a special case of power series distributions with $a_{n}=1$ and $C(\theta)=\frac{\theta}{1-\theta} \ (0<\theta<1) $ . The  pdf and hazard rate function of Gompertz-geometric (GG) distribution is given respectively by
\begin{eqnarray}\label{eq.fGG}
f(x)=\frac{(1-\theta)\beta e^{\gamma x}e^{-\frac{\beta}{\gamma}(e^{\gamma x}-1)}}{(1-\theta e^{-\frac{\beta}{\gamma}(e^{\gamma x}-1)})^{2}},
\end{eqnarray}
and
\begin{eqnarray}
h(x)=\frac{\beta e^{\gamma x}}{1-\theta e^{-\frac{\beta}{\gamma}(e^{\gamma x}-1)}}.
\end{eqnarray}

\bigskip
\begin{remark}
When $\theta^{\ast}=1-\theta$, from \eqref{eq.fGG} we have
\begin{eqnarray}\label{eq.fGGe}
f(x)= \frac{\theta^*\beta e^{\gamma x}e^{-\frac{\beta}{\gamma}(e^{\gamma x}-1)}}{(1-(1-\theta^*) e^{-\frac{\beta}{\gamma}(e^{\gamma x}-1)})^{2}}.
\end{eqnarray}
Based on \cite{ma-ol-97}
%Marshal-Olkin (1997),
$f(x)$ in \eqref{eq.fGGe} also is density for all $\theta^*>0$ ($\theta<1$).

Note that when $\gamma\rightarrow 0^+$, the pdf of extended exponential geometric (EEG) distribution
 \citep[see][]{ad-di-05}
%(Adamidis et al., 2005)
 concludes from the pdf in \eqref{eq.fGGe} with $\theta^*>0$. The EEG hazard function is monotonically increasing for $\theta^*>1$; decreasing for $0<\theta^*<1$ and constant for $\theta^*=1$.
\end{remark}

\begin{remark}
If $\theta^*=1$, then the pdf in  \eqref{eq.fGGe} becomes the pdf of Gompertz distribution. Note that the hazard rate function of Gompertz distribution is increasing.
\end{remark}

The plots of density and hazard rate function of GG distribution for different values of $\beta$, $\gamma$ and $\theta^*$  are given in Figure \ref{fig.GG}. We can see that the hazard rate function of GG distribution is increasing  or bathtub.

%Now, when $\beta\longrightarrow o$ ,  the approximated value of $E(X)$  for GG distribution is given by
%$$
%E(X)\approx \frac{1-\theta}{\gamma}\sum\limits_{n=1}^{\infty} \theta^{n-1}e^{\frac{n\beta}{\gamma}}(\frac{n\beta}{\gamma}-\ln(\frac{n\beta}{\gamma})+\Gamma'(1)),
%$$

%\newpage

%
%$ \ $
%\newpage

%\begin{theorem}
%Consider the GG hazard function in (). Then it has minimum value in $x=\frac{1}{\gamma}\ln[1-\frac{\gamma}{\beta}\ln(\frac{\gamma}{\theta(\beta+\gamma)})]$.\\
%\noindent\textbf{Proof.}\;.See Appendix A.1.
%\end{theorem}

\subsection{ Gompertz - Poisson distribution}

The Poisson distribution (truncated at zero) is a special case of power series distributions with $a_{n}=\frac{1}{n!}$ and $C(\theta)=e^{\theta}-1 \ (\theta>0).$ The pdf and hazard rate function of Gompertz-Poisson (GP) distribution are given respectively by
\begin{eqnarray}
f(x)= \frac{\theta \beta e^{\gamma x}e^{-\frac{\beta}{\gamma}(e^{\gamma x}-1)}e^{\theta e^{-\frac{\beta}{\gamma}(e^{\gamma x}-1)}}}{e^\theta-1},
\end{eqnarray}
and
\begin{eqnarray}
h(x)=\frac{\theta \beta e^{\gamma x}e^{-\frac{\beta}{\gamma}(e^{\gamma x}-1)}}{1-e^{-\theta e^{-\frac{\beta}{\gamma}(e^{\gamma x}-1)}}}.
\end{eqnarray}
%Now, when $\beta\longrightarrow o$ ,  the approximated value of $E(X)$  for GP distribution is given by
%$$
%E(X)\approx \frac{1}{\gamma(e^{\theta}-1)}\sum\limits_{n=1}^{\infty} \frac{\theta^{n}}{n!}e^{\frac{n\beta}{\gamma}}(\frac{n\beta}{\gamma}-\ln(\frac{n\beta}{\gamma})+\Gamma'(1)),
%$$

The plots of density and hazard rate function of GP for different values of $\beta$, $\gamma$ and $\theta$  are given in Figure \ref{fig.GP}. We can see that the hazard rate function of GP distribution is increasing  or bathtub.

%\newpage

\begin{figure}[t]
\centering
\includegraphics[scale=0.35]{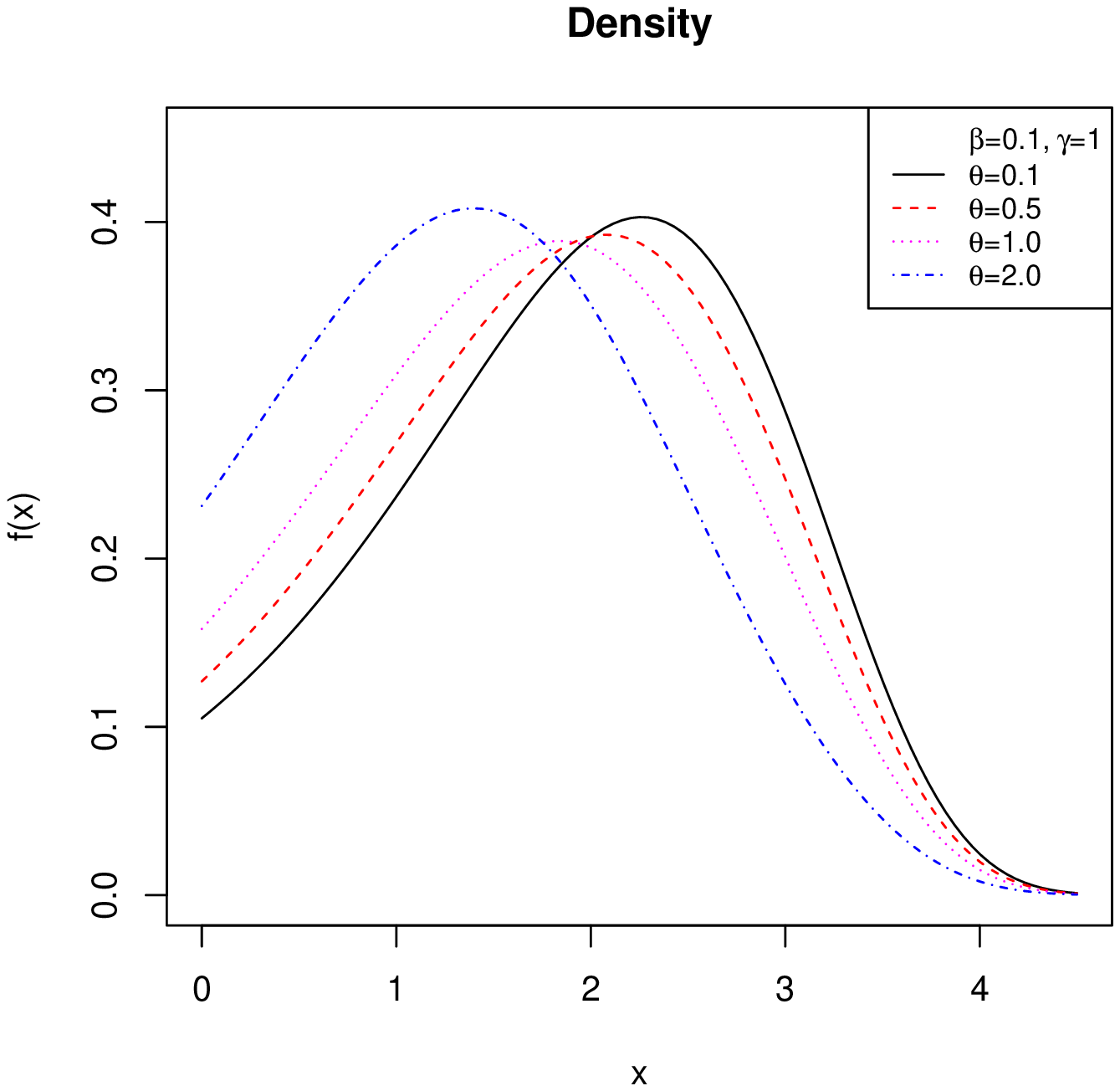}
\includegraphics[scale=0.35]{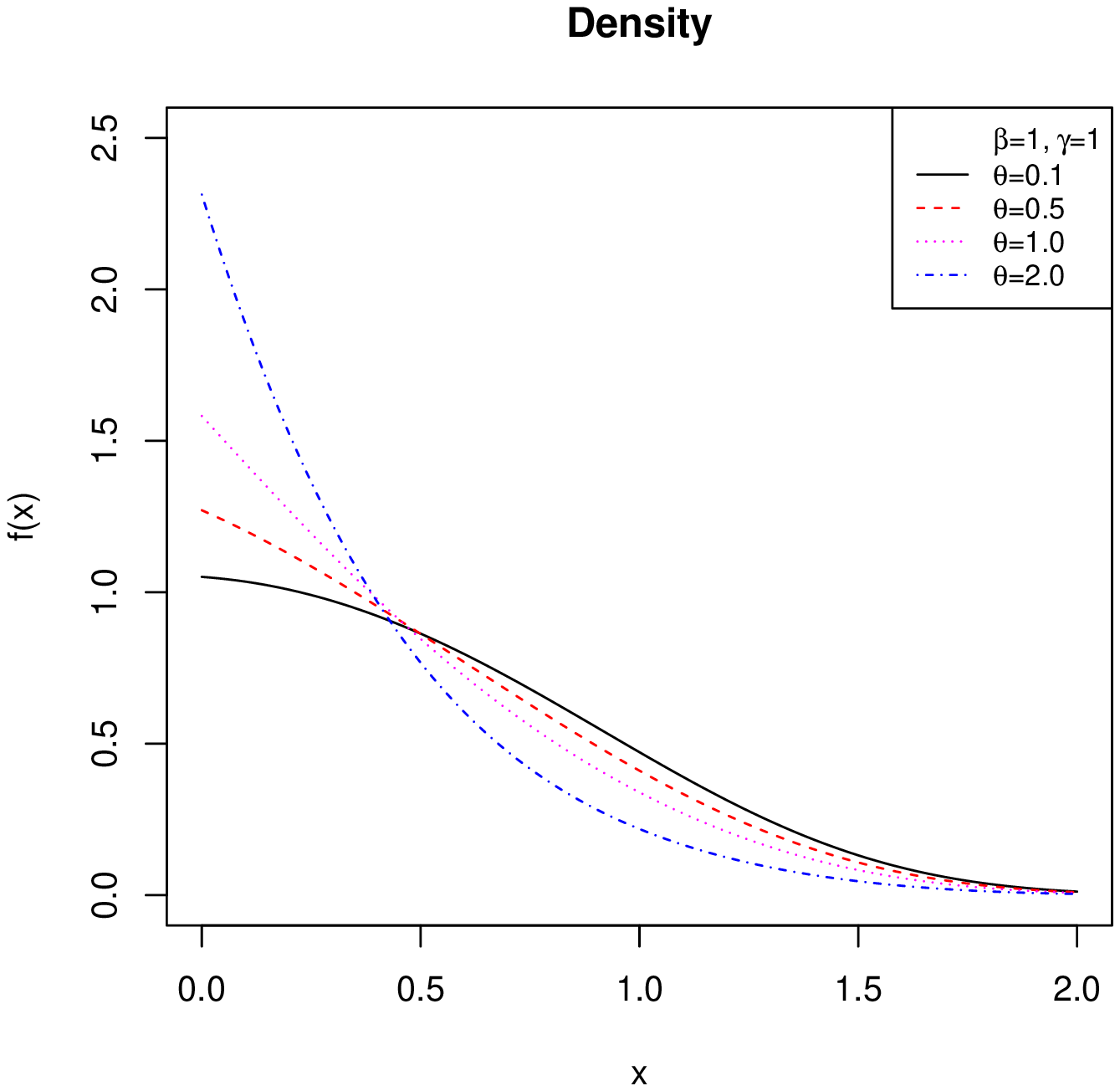}
\includegraphics[scale=0.35]{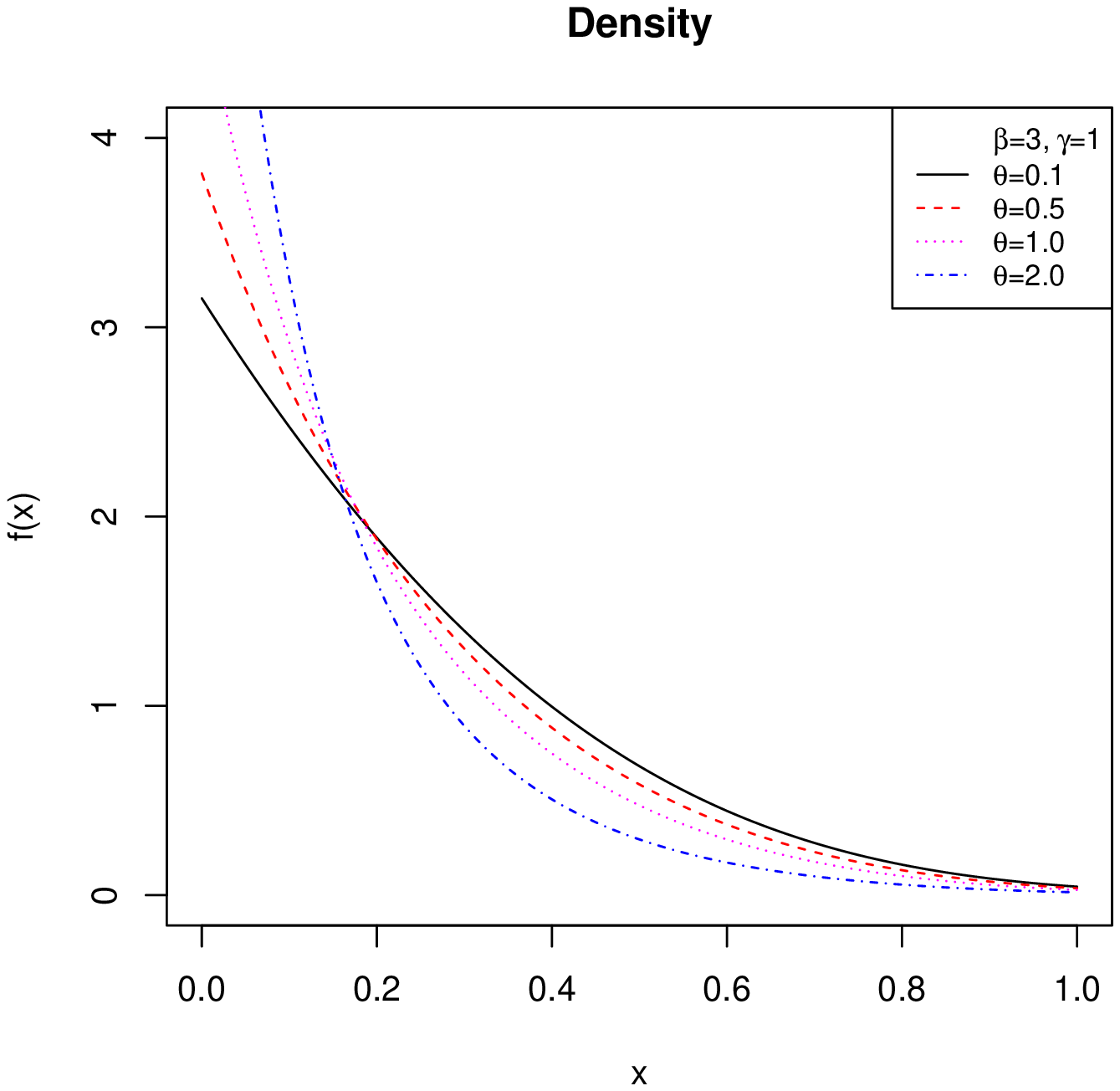}
%\includegraphics[scale=0.35]{denGP9.eps}
%\vspace{-0.8cm}
%\caption[]{Plots of density and hazard rate functions of GP for  different values $\beta $, $\gamma$ and $\theta$.}\label{fig.GP}
%\end{figure}
%
%$\ $
%\newpage
%
%\begin{figure}[t]
%\centering
%\includegraphics[scale=0.33]{hazGP1.eps}
\includegraphics[scale=0.35]{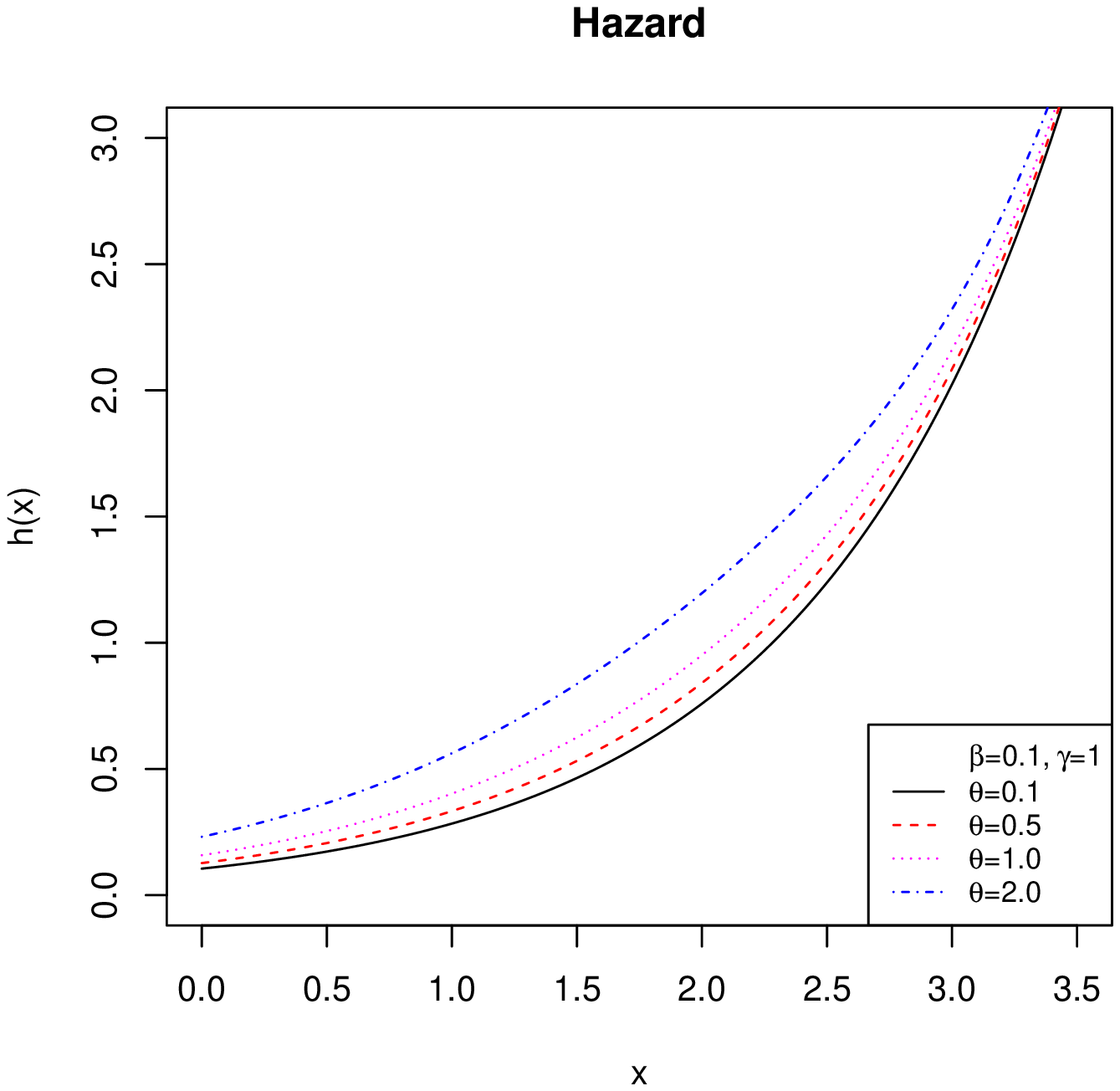}
\includegraphics[scale=0.35]{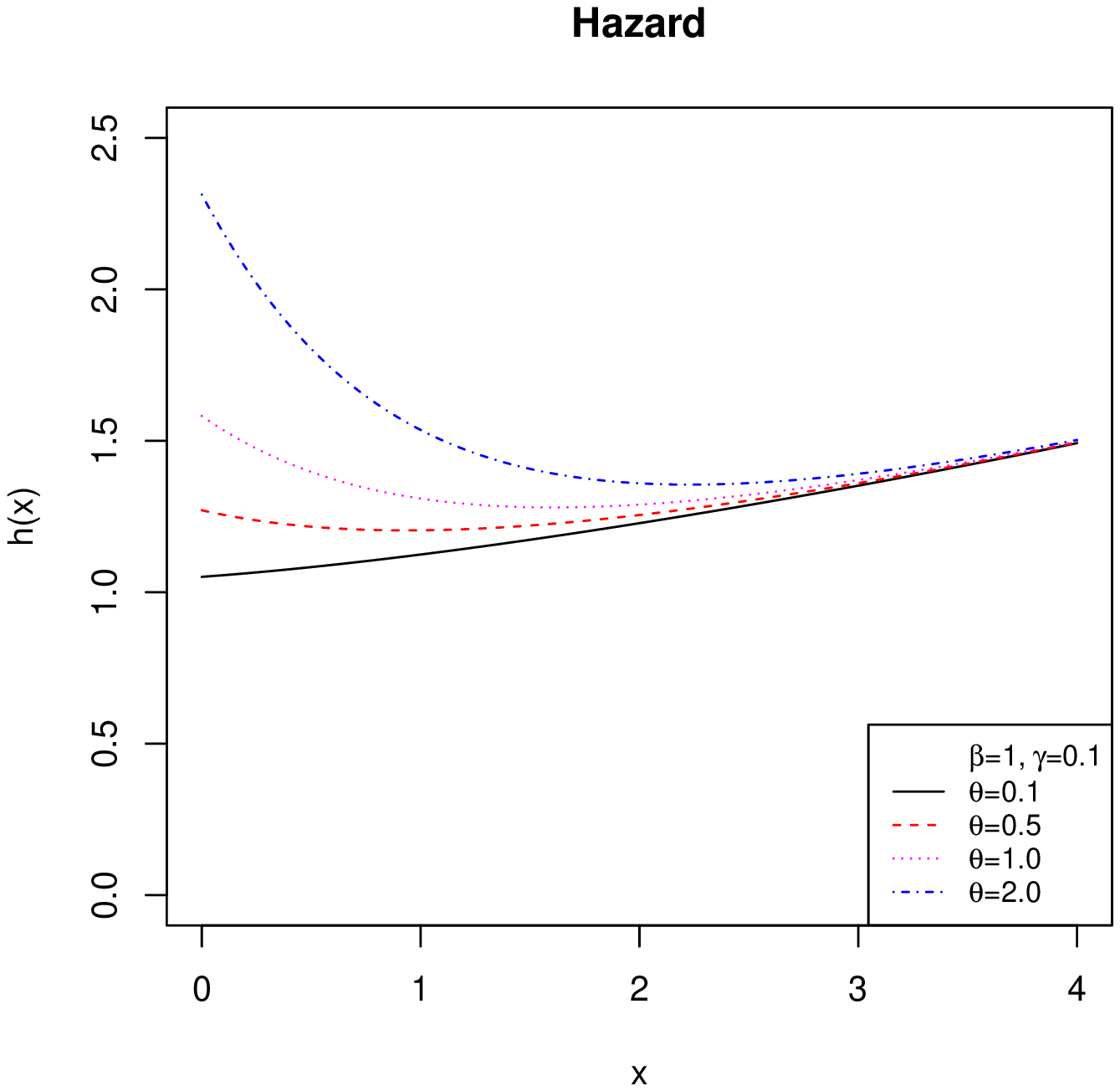}
\includegraphics[scale=0.35]{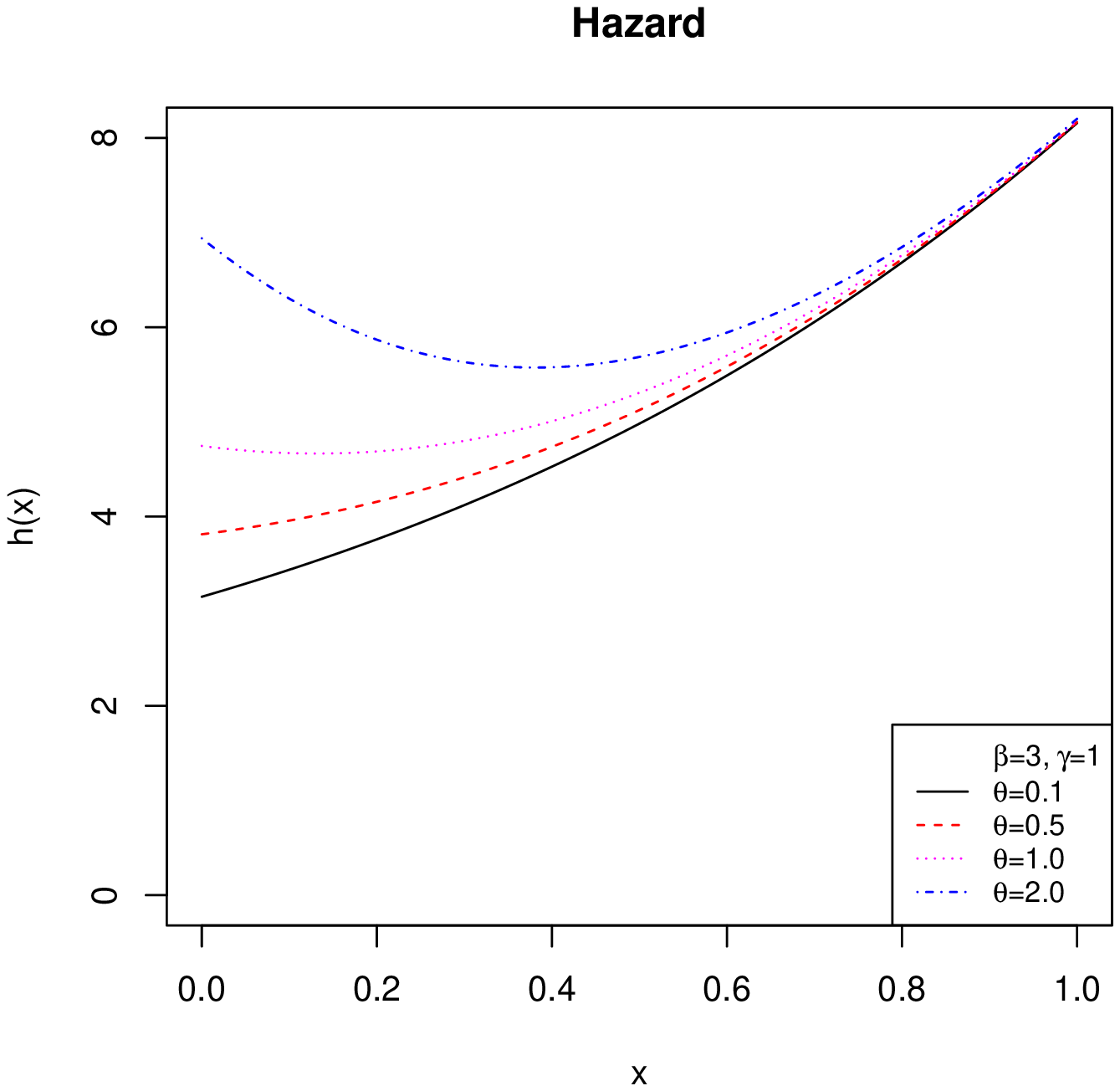}

\vspace{-0.8cm}
\caption[]{Plots of density and hazard rate functions of GP for  different values $\beta $, $\gamma$ and $\theta$.}\label{fig.GP}
\end{figure}

%$ \ $
%\newpage

\subsection{Gompertz - binomial distribution}

The binomial distribution (truncated at zero) is a special case of power series distributions with $a_{n}=\dbinom{m}{n}$ and $C(\theta)=(\theta+1)^{m}-1 \ (\theta>0),$ where $m$ $(n\leq m)$  is the number of replicas. The pdf and hazard rate function of
 Gompertz - binomial (GB) distribution are given respectively by
\begin{eqnarray}
f(x)=\frac{m \theta \beta e^{\gamma x}e^{-\frac{\beta}{\gamma}(e^{\gamma x}-1)}(\theta e^{-\frac{\beta}{\gamma}(e^{\gamma x}-1)}+1)^{m-1}}{(\theta+1)^{m}-1},
\end{eqnarray}
and
\begin{eqnarray}
h(x)=\frac{m \theta \beta e^{\gamma x}e^{-\frac{\beta}{\gamma}(e^{\gamma x}-1)}(\theta e^{-\frac{\beta}{\gamma}(e^{\gamma x}-1)}+1)^{m-1}}{(\theta e^{-\frac{\beta}{\gamma}(e^{\gamma x}-1)}+1)^{m}-1}.
\end{eqnarray}

The plots of density and hazard rate function of GB for $m=5$, and different values of $\beta$, $\gamma$ and $\theta$  are given in Figure \ref{fig.GB}. We can see that the hazard rate function of GB distribution is increasing  or bathtub. We can find that the GP distribution can be obtained as limiting of GB distribution if $m\theta\longrightarrow \lambda>0$, when $m\longrightarrow \infty$.

%Now, when $\beta\longrightarrow o$ ,  the approximated value of $E(X)$  for GB distribution is given by
%$$
%E(X)\approx \frac{1}{\gamma[(\theta+1)^{m}-1]}\sum\limits_{n=1}^{\infty} \dbinom{m}{n}\theta^{n}e^{\frac{n\beta}{\gamma}}[\frac{n\beta}{\gamma}-\ln(\frac{n\beta}{\gamma})+\Gamma'(1)],
%$$

%\newpage

\begin{figure}[t]
\centering
\includegraphics[scale=0.35]{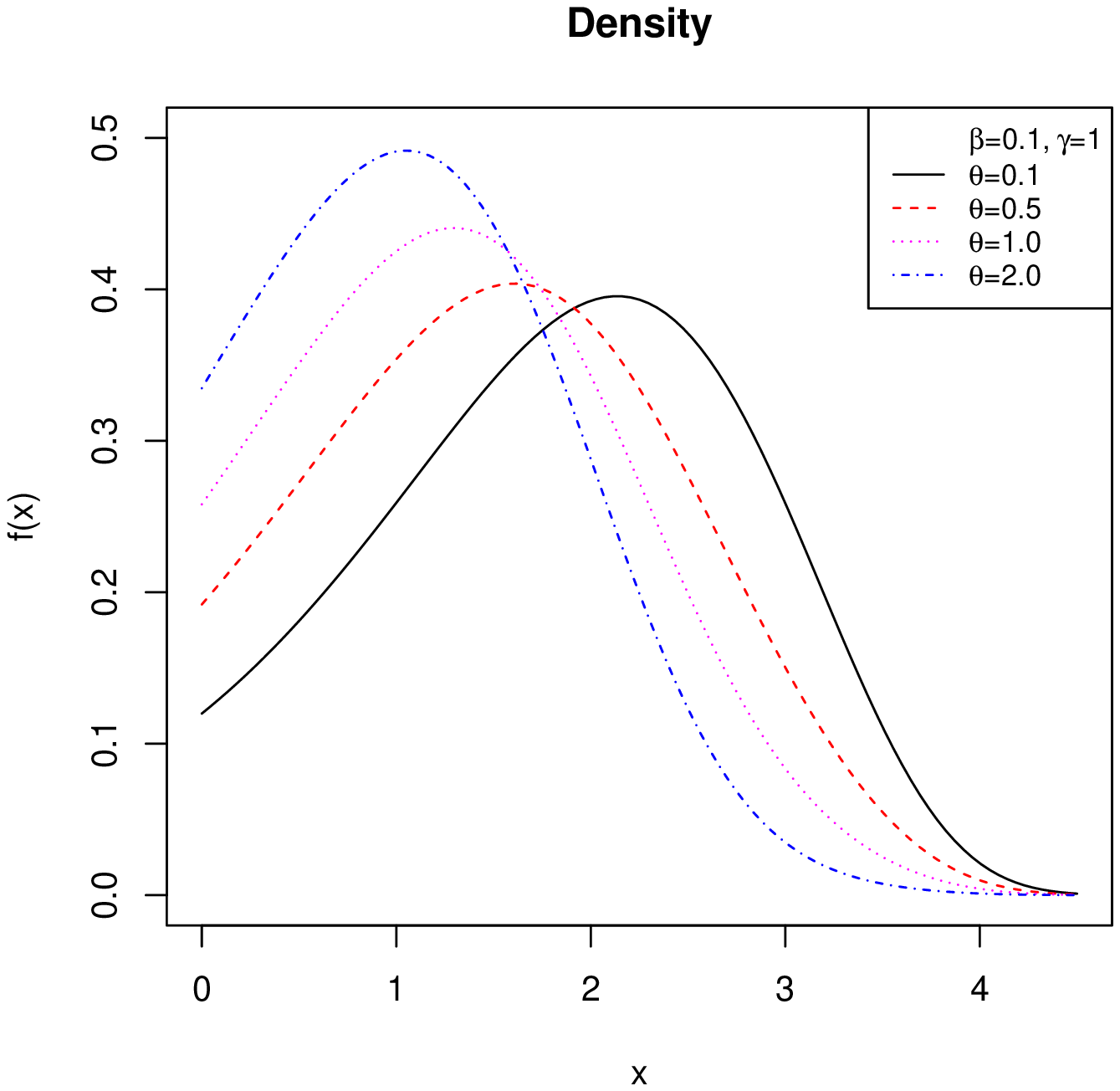}
\includegraphics[scale=0.35]{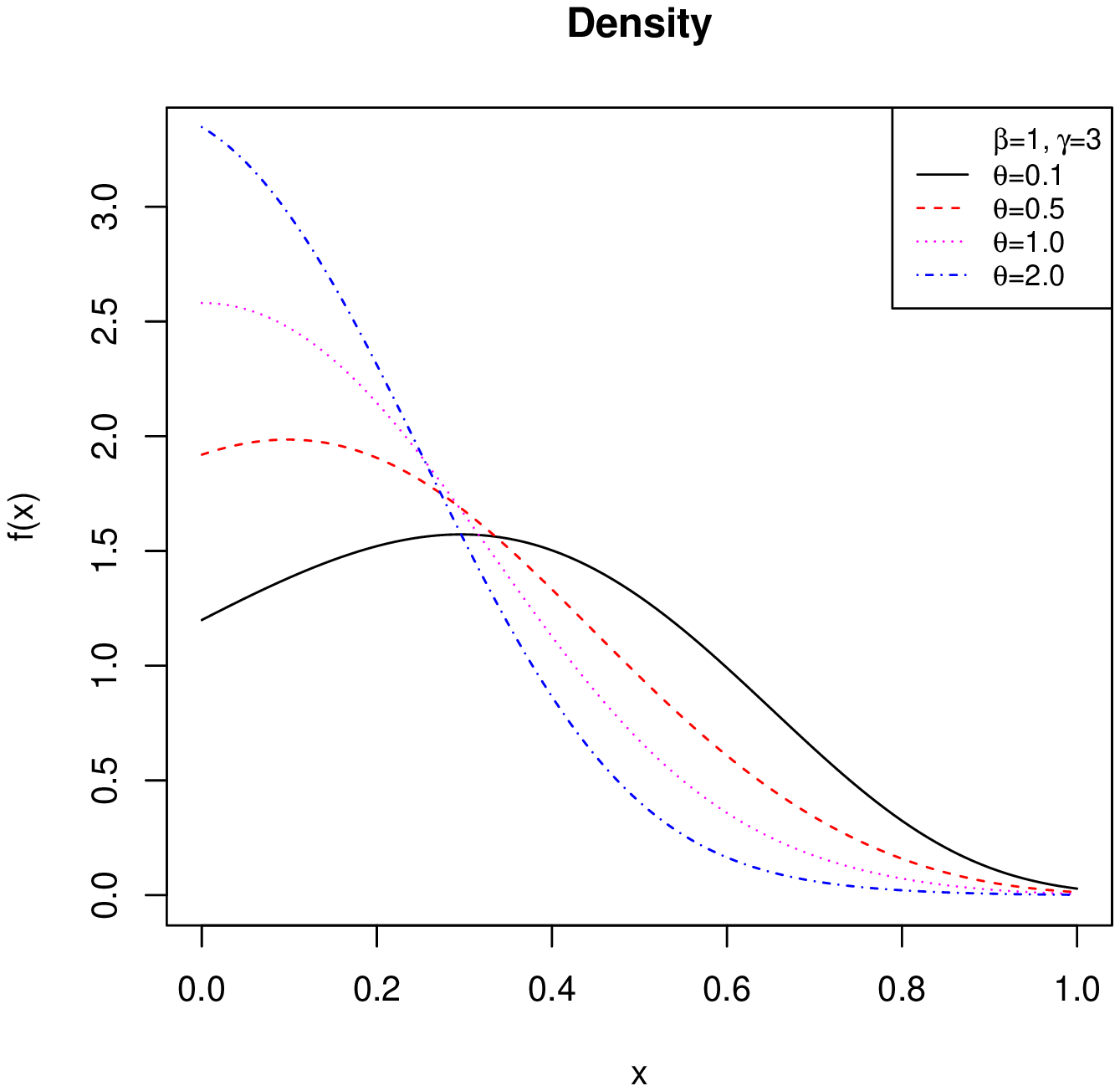}
\includegraphics[scale=0.35]{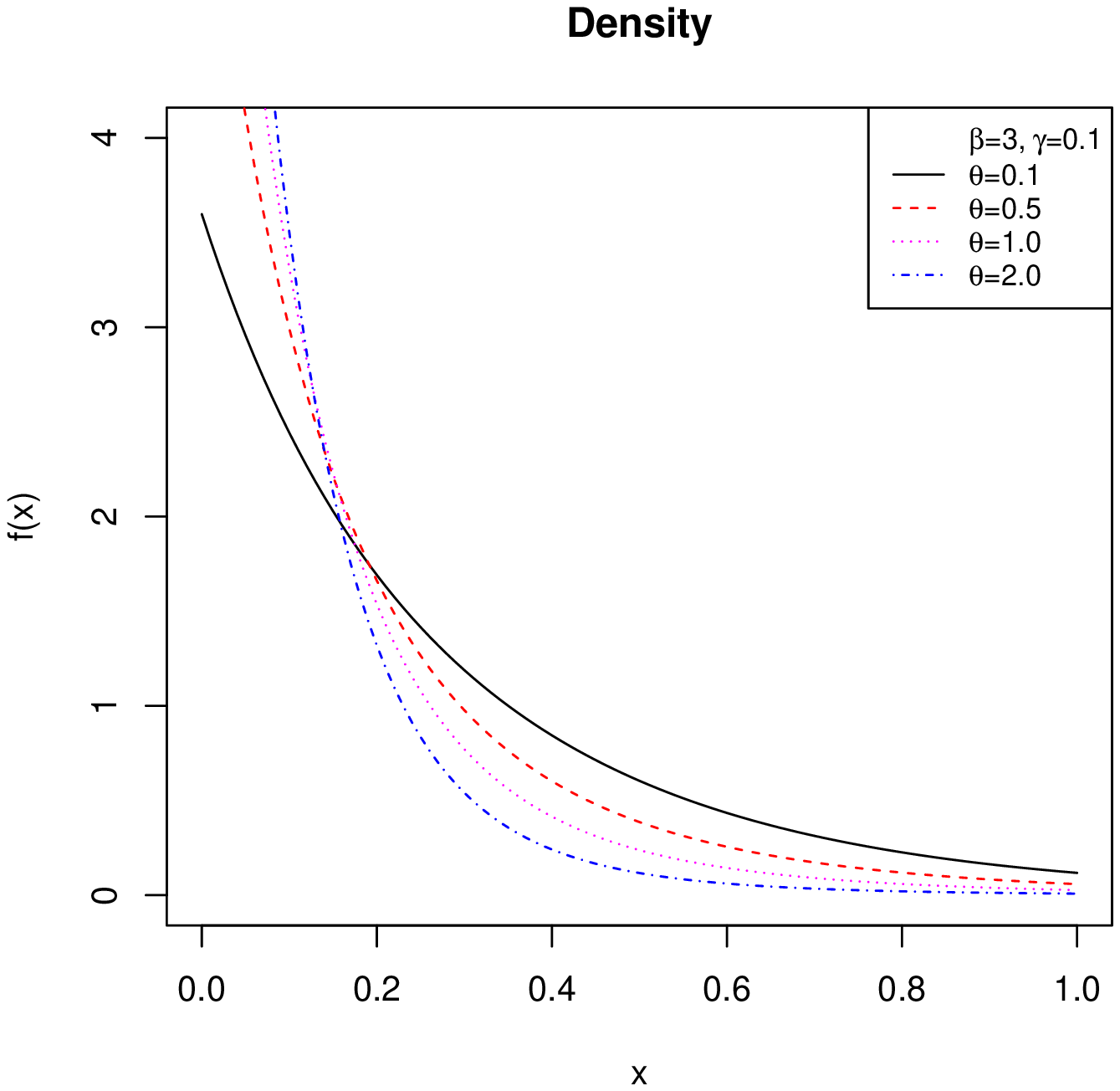}
%\includegraphics[scale=0.35]{denGB8.eps}
%\includegraphics[scale=0.35]{denGB9.eps}
%\vspace{-0.8cm}
%\caption[]{Plots of density and hazard rate functions of GB for $m=5$, and different values $\beta $, $\gamma$ and $\theta$.}\label{fig.GB}
%\end{figure}

%$\ $
%\newpage
%
%\begin{figure}[t]
%\centering
%\includegraphics[scale=0.33]{hazGB1.eps}
\includegraphics[scale=0.35]{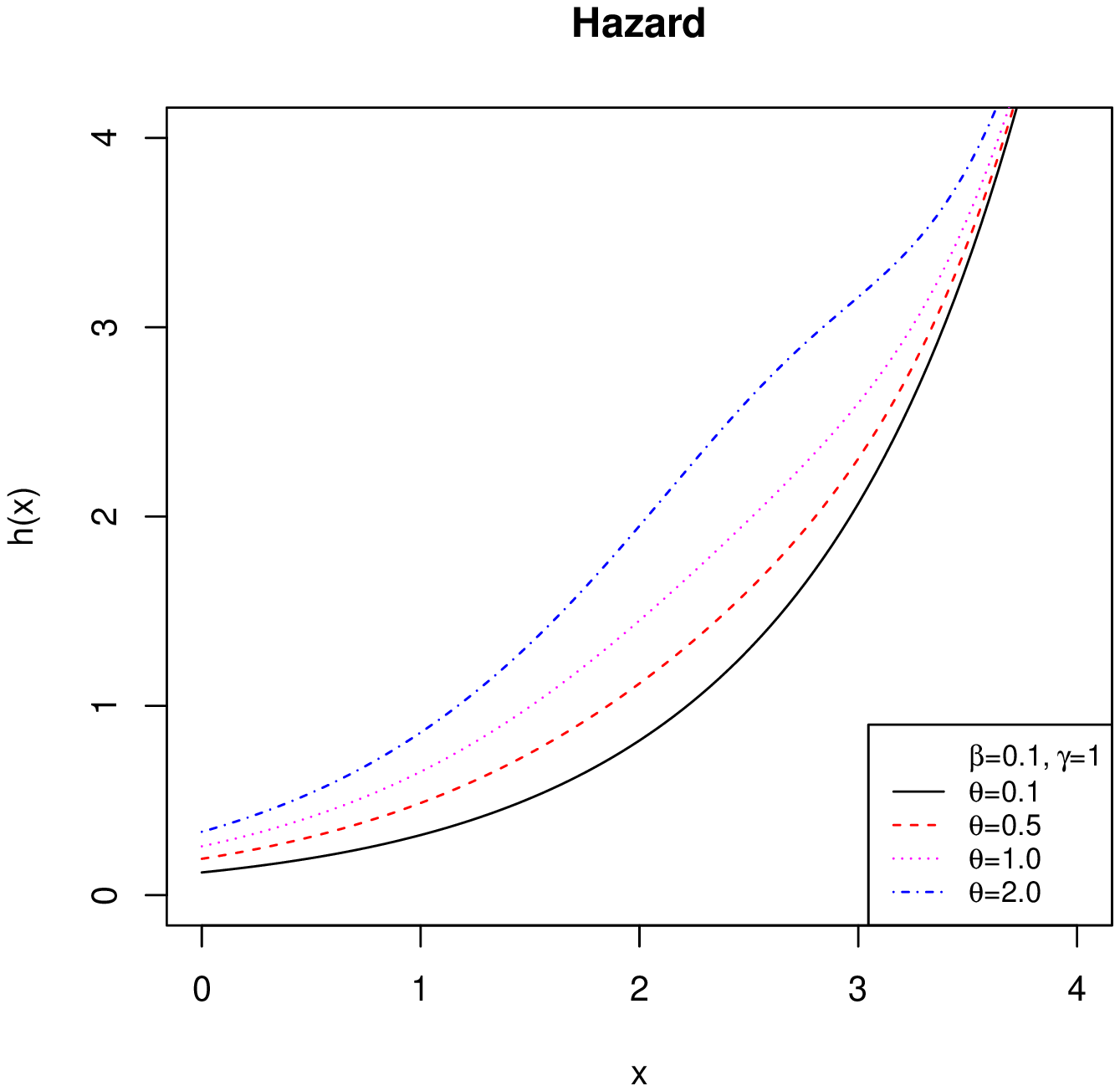}
\includegraphics[scale=0.35]{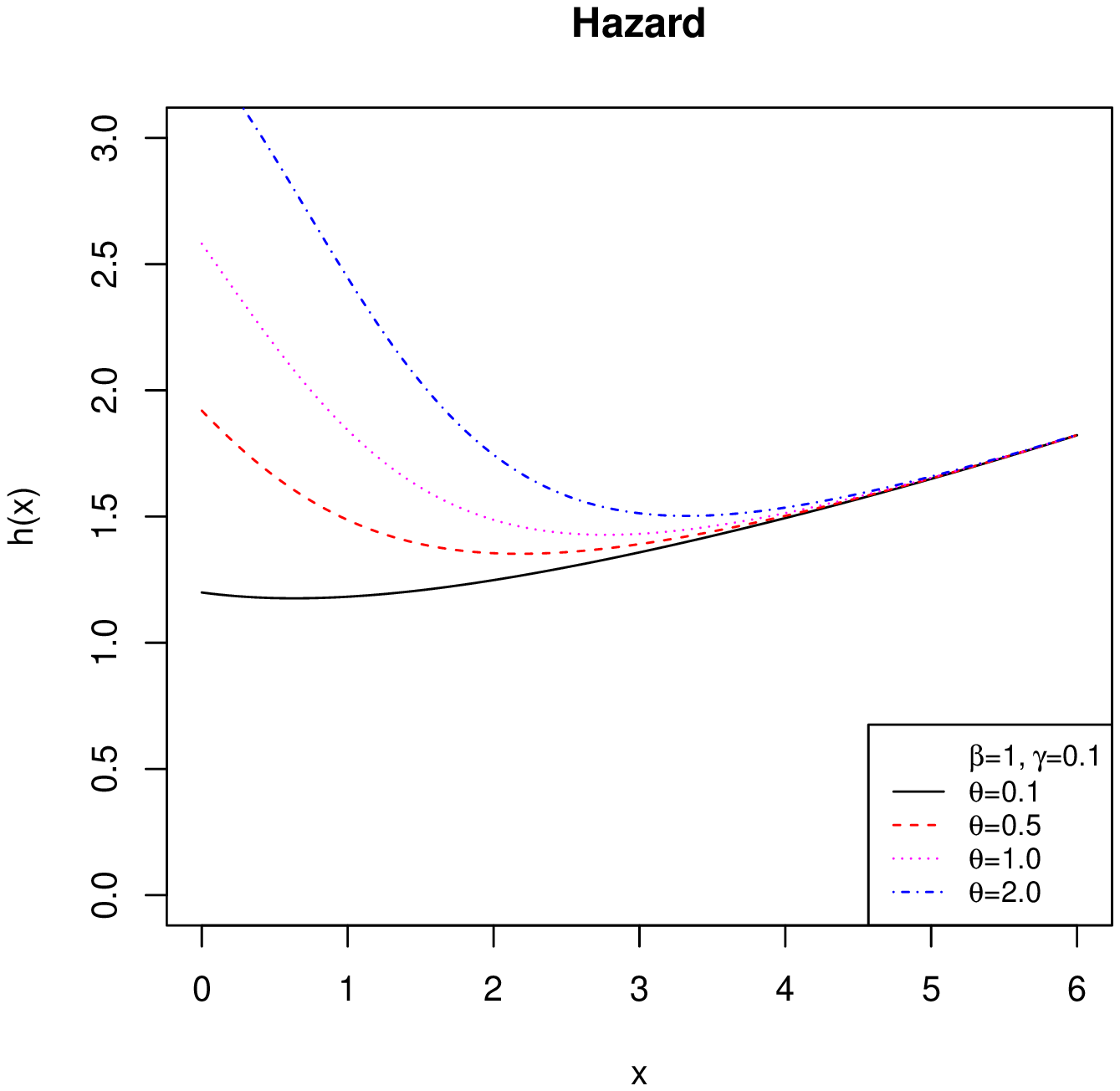}
\includegraphics[scale=0.35]{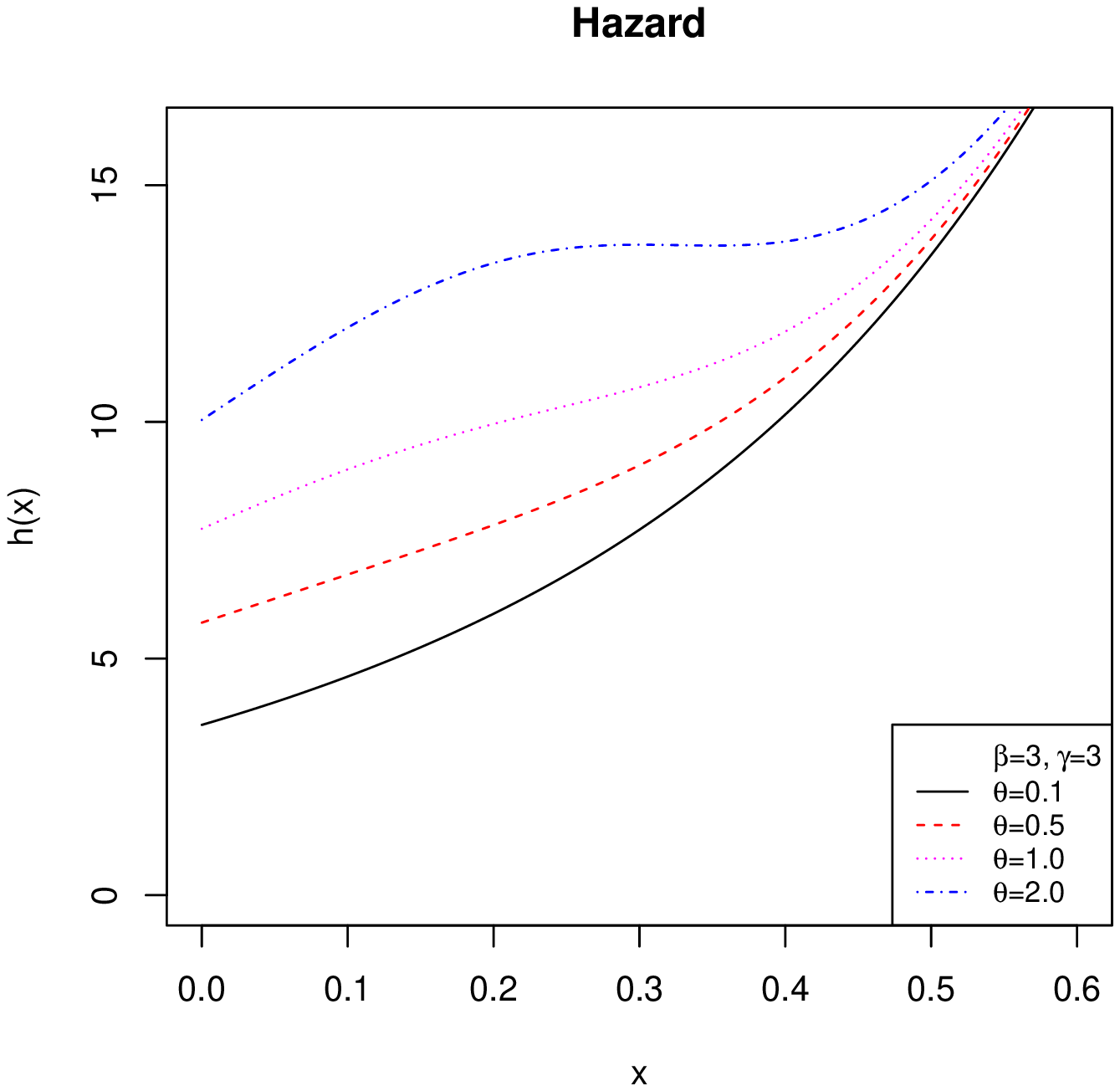}

\vspace{-0.8cm}
\caption[]{Plots of density and hazard rate functions of GB for $m=5$, and  different values $\beta $, $\gamma$ and $\theta$.}\label{fig.GB}
\end{figure}

%$ \ $
%\newpage

\subsection{ Gompertz - logarithmic distribution}
The logarithmic distribution (truncated at zero) is also a special case of power series distributions with $a_{n}=\frac{1}{n}$ and $C(\theta)=-\log(1-\theta) \ (0<\theta<1)$. The  pdf and hazard rate function of Gompertz - logarithmic (GL) distribution are given respectively by
\begin{eqnarray}
f(x)=  \frac{\theta \beta e^{\gamma x}e^{-\frac{\beta}{\gamma}(e^{\gamma x}-1)}}{(\theta e^{-\frac{\beta}{\gamma}(e^{\gamma x}-1)}-1)\log(1-\theta)},
\end{eqnarray}
and
\begin{eqnarray}
h(x)=  \frac{\theta \beta e^{\gamma x}e^{-\frac{\beta}{\gamma}(e^{\gamma x}-1)}}{(\theta e^{-\frac{\beta}{\gamma}(e^{\gamma x}-1)}-1)\log(1-\theta e^{-\frac{\beta}{\gamma}(e^{\gamma x}-1)})}.
\end{eqnarray}

The plots of density and hazard rate function of GL for different values of $\beta$, $\gamma$ and $\theta$  are given in Figure \ref{fig.GL}. We can see that the hazard rate function of GL distribution is increasing  or bathtub.

%Now, when $\beta\longrightarrow o$ ,  the approximated value of $E(X)$  for GP distribution is given by
%$$
%E(X)\approx \frac{-1}{\gamma \log(1-\theta)}\sum\limits_{n=1}^{\infty} \frac{\theta^{n}}{n}e^{\frac{n\beta}{\gamma}}[\frac{n\beta}{\gamma}-\ln(\frac{n\beta}{\gamma})+\Gamma'(1)],
%$$

%\newpage

\begin{figure}[t]
\centering
\includegraphics[scale=0.35]{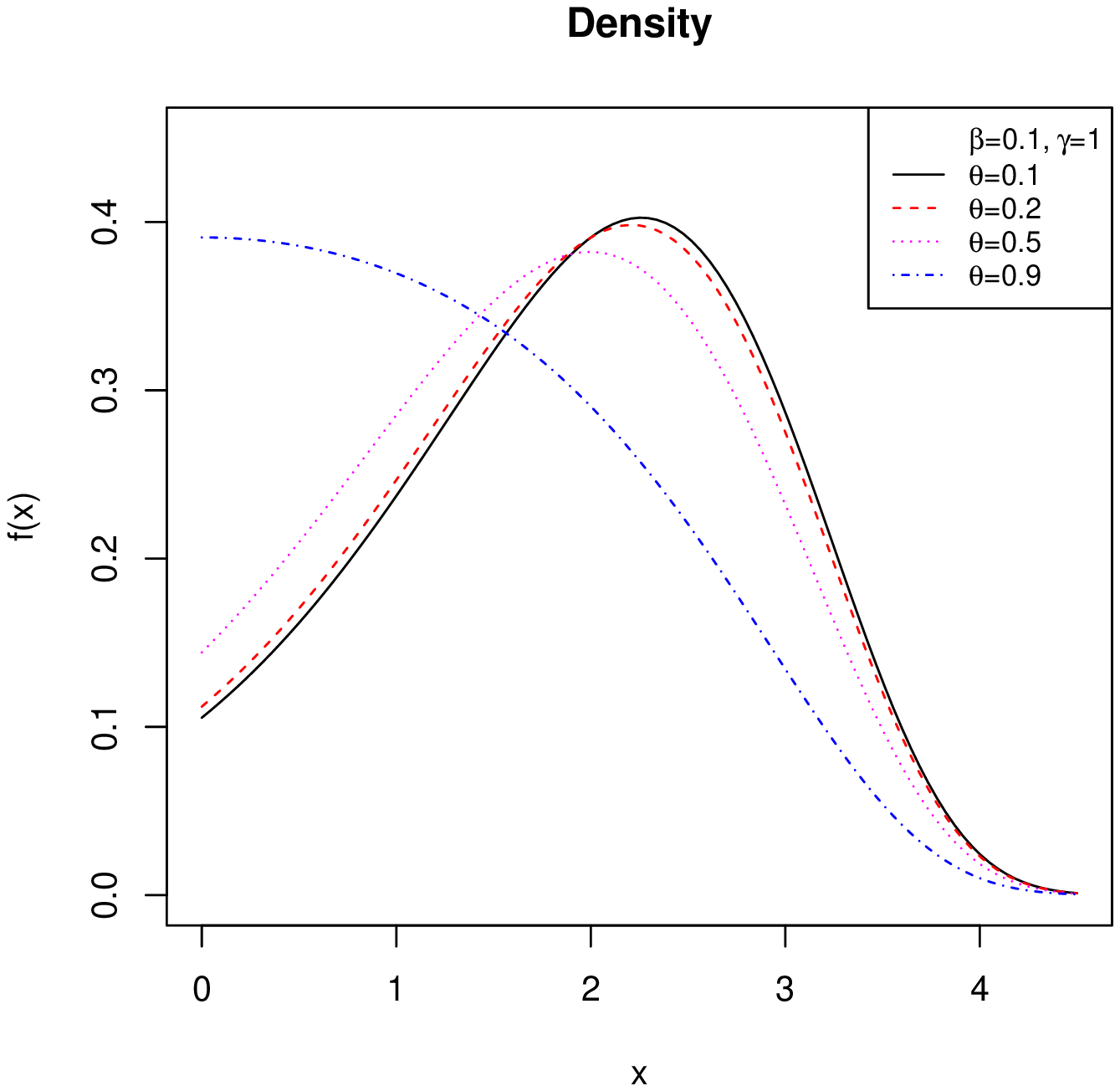}
\includegraphics[scale=0.35]{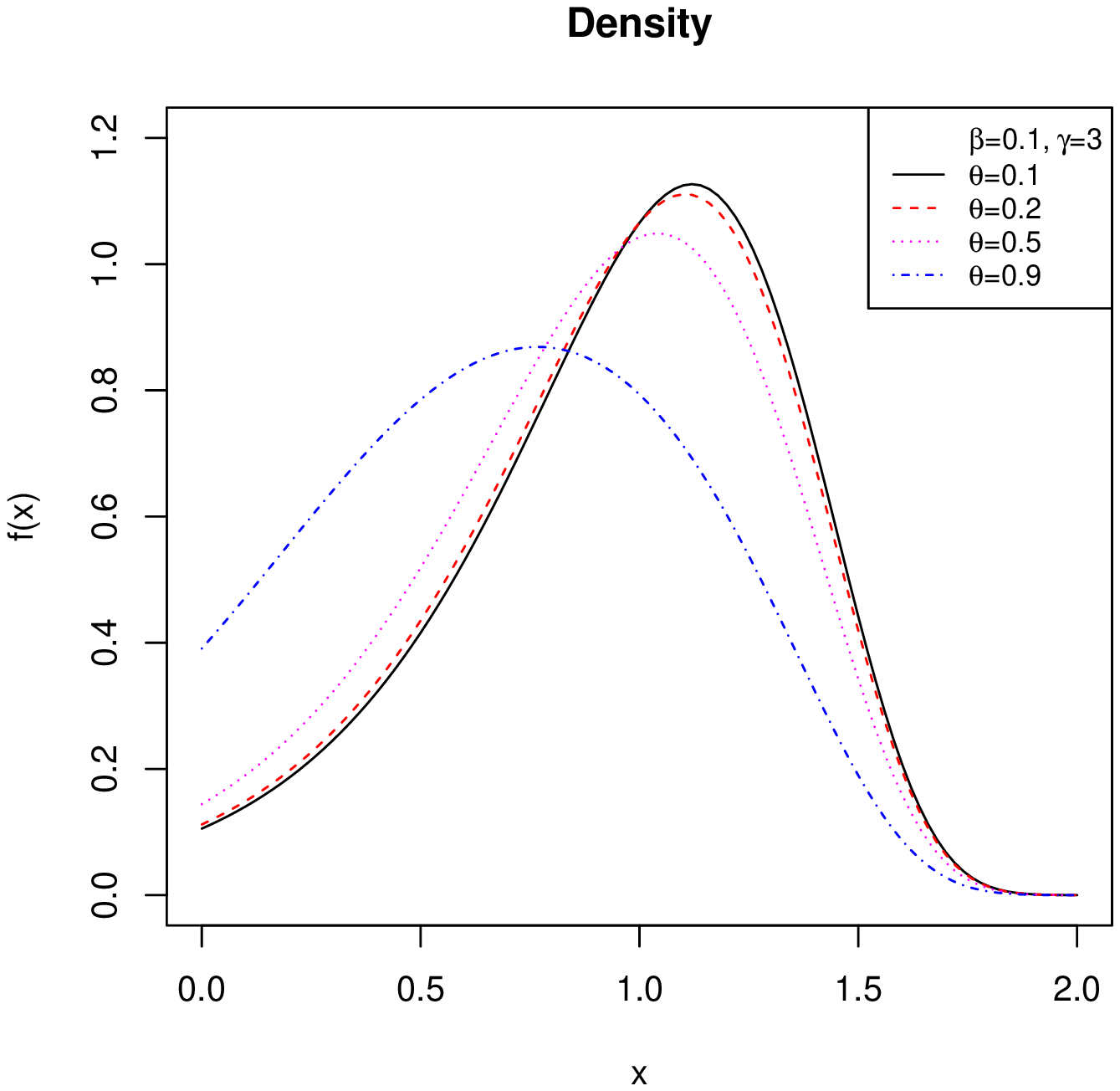}
\includegraphics[scale=0.35]{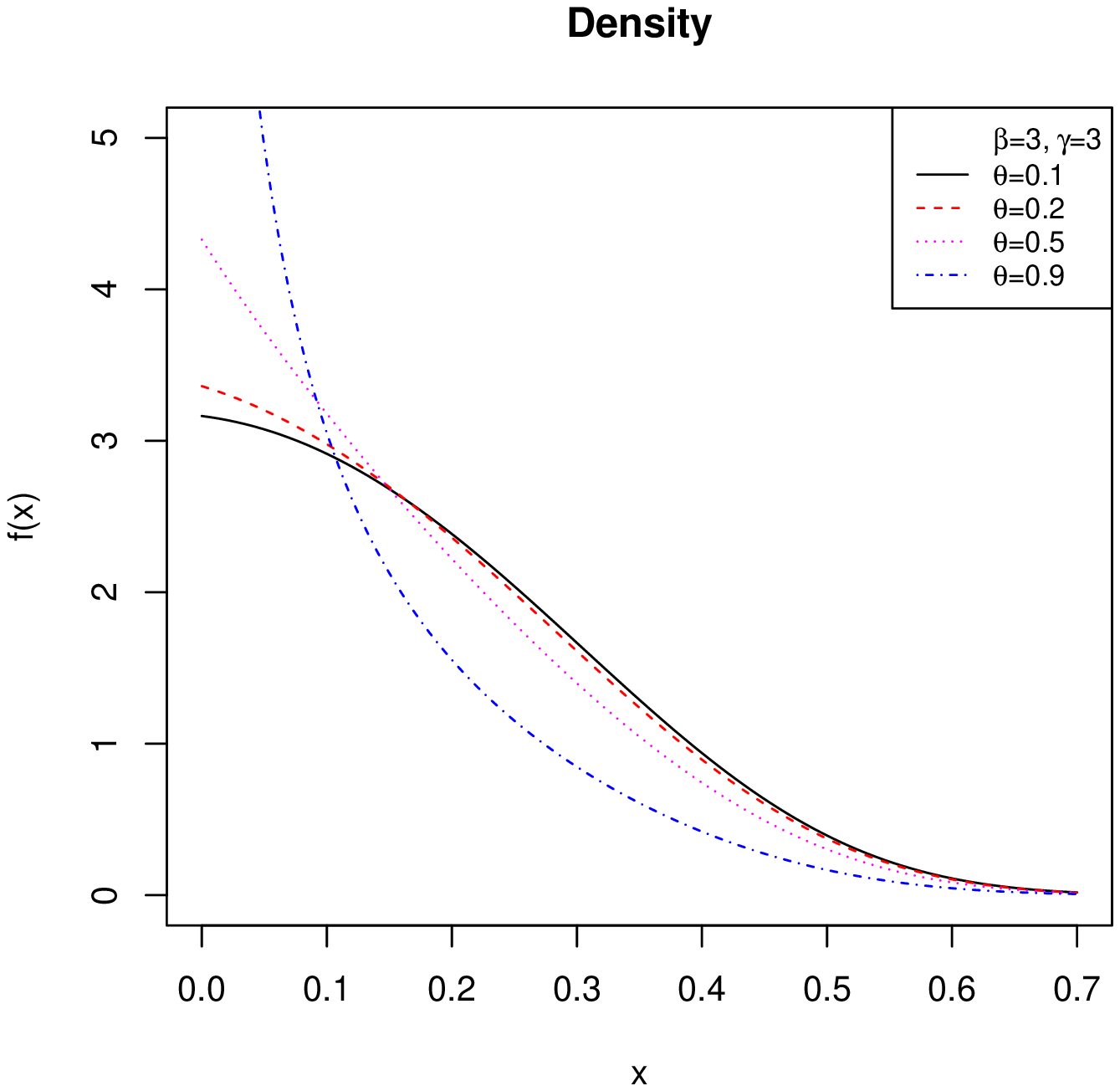}
%\vspace{-0.8cm}
%\caption[]{Plots of density and hazard rate functions of GL for  different values $\beta $, $\gamma$ and $\theta$.}\label{fig.GL}
%\end{figure}
%
%$\ $
%\newpage
%
%\begin{figure}[t]
%\centering
%\includegraphics[scale=0.33]{hazGL1.eps}
\includegraphics[scale=0.35]{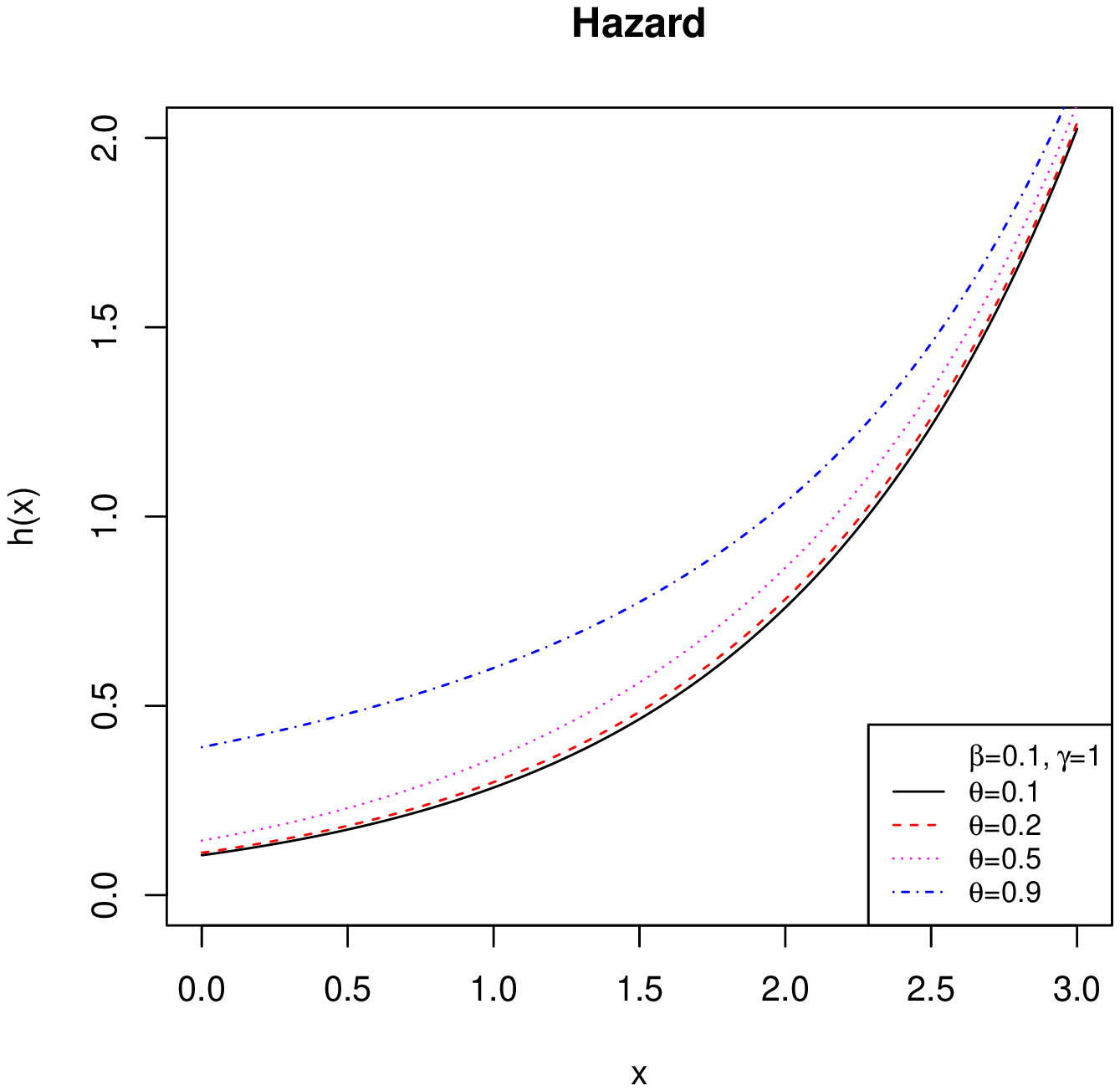}
\includegraphics[scale=0.35]{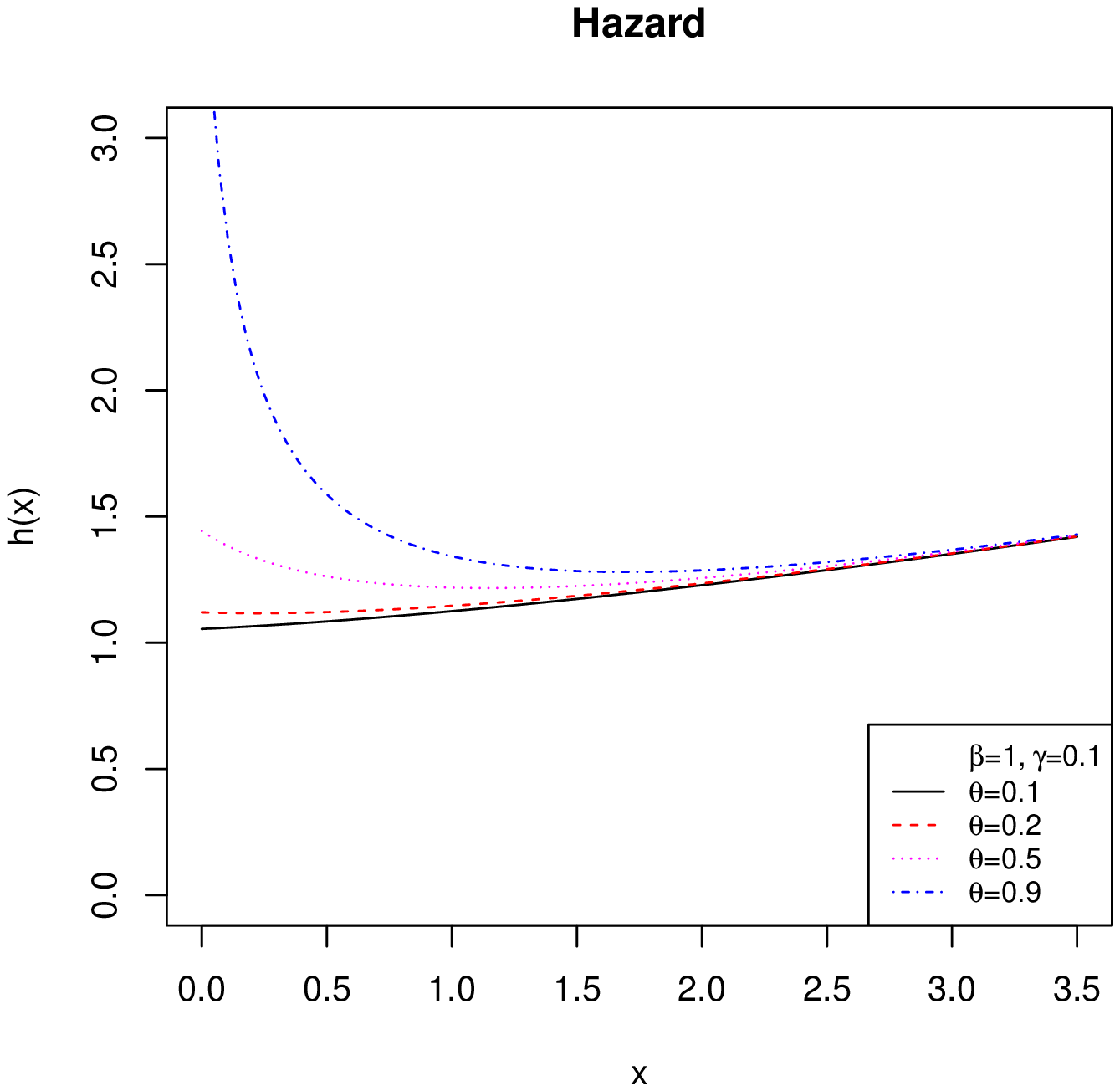}
\includegraphics[scale=0.35]{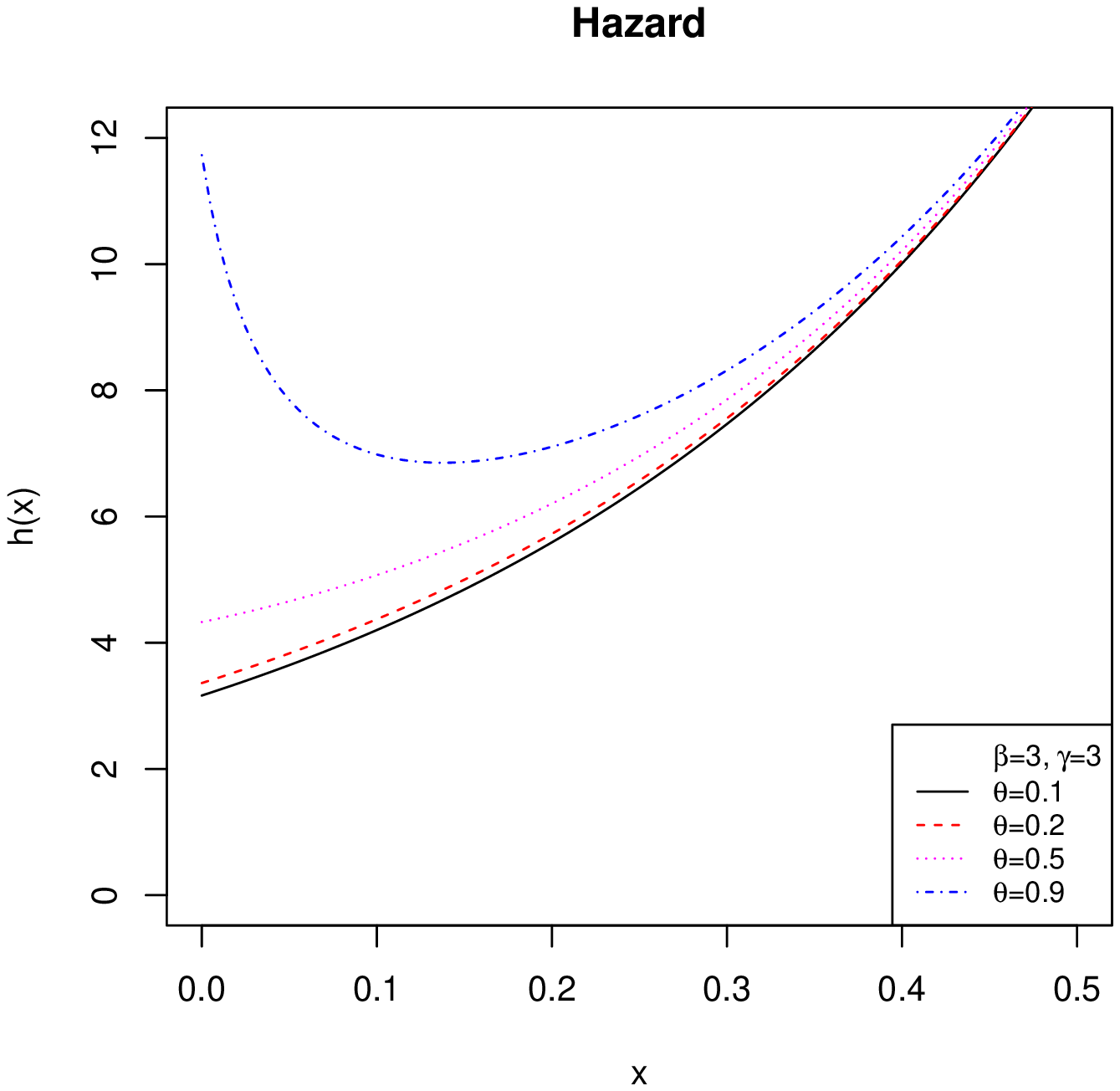}

\vspace{-0.8cm}
\caption[]{Plots of density and hazard rate functions of GL for different values $\beta $, $\gamma$ and $\theta$.}\label{fig.GL}
\end{figure}

\section{Estimation and inference}
\label{sec.est}
In this Section, we will derive the maximum likelihood estimators (MLE) of the unknown parameters  ${\boldsymbol \Theta}=(\beta, \gamma, \theta)^T$ of the $GPS(\beta, \gamma, \theta)$. Also, asymptotic confidence intervals of these parameters will be derived based on the Fisher information. At the end, we will propose an Expectation - Maximization (EM) algorithm for estimating the parameters.

\subsection{MLE for parameters}

 Let $X_{1},..., X_{n}$ be a random sample from  $GPS(\beta,\gamma,\theta)$, and let ${\boldsymbol x}=(x_{1},...,x_{n})$ be the observed values of this random sample. The log-likelihood function is given by
\begin{eqnarray*}
 l_{n} = l_{n}({\boldsymbol\Theta};{\boldsymbol x})= n\log(\theta)+n\log(\beta)+n\gamma \bar{x}+\sum\limits_{i=1}^{n}\log(t_i) +\sum\limits_{i=1}^{n}\log(C'(\theta t_i)) -n\log(C(\theta)),
\end{eqnarray*}
where $t_{i}=e^{-\frac{\beta}{\gamma}(e^{\gamma x_{i}}-1)}$.
Therefore, the  score function is given by $U({\boldsymbol\Theta};{\boldsymbol x})=(\frac{\partial l_{n}}{\partial\beta},\frac{\partial l_{n}}{\partial\gamma},\frac{\partial l_{n}}{\partial\theta})^{T}$, where
 \begin{eqnarray}
&& \frac{\partial l_{n}}{\partial \beta}=\frac{n}{\beta}+\frac{1}{\beta}\sum\limits_{i=1}^{n}\log(t_i)+\frac{\theta}{\beta}\sum\limits_{i=1}^{n}\frac{ t_i\log(t_i) C''(\theta t_{i})}{C'(\theta t_{i})},\label{eq.lb}\\                                                                                                                                                          && \frac{\partial l_{n}}{\partial \gamma}=n\bar{x}+\sum\limits_{i=1}^{n}d_i
+\theta\sum\limits_{i=1}^{n}\frac{b_i C''(\theta t_{i})}{C'(\theta t_{i})},\label{eq.lg}\\
&& \frac{\partial l_{n}}{\partial \theta}=\frac{n}{\theta}+\sum\limits_{i=1}^{n}\frac{
 t_{i}C''(\theta t_{i})}{C'(\theta t_{i})}-\frac{nC'(\theta)}{C(\theta)},\label{eq.lt}
  \end{eqnarray}
and $b_i=\frac{\partial t_{i}}{\partial \gamma}=t_i d_i$ and $d_i=\frac{\partial \log(t_{i})}{\partial \gamma}=\frac{1 }{\gamma}(-\log(t_i)+\gamma x_i\log(t_i)-\beta x_i)$.

 The MLE of ${\boldsymbol \Theta}$, say $\hat{\boldsymbol\Theta}$, is obtained by solving the nonlinear system $U({\boldsymbol\Theta};{\boldsymbol x})={\boldsymbol 0}$. We cannot get an explicit form for this nonlinear system of equations and they can be calculated by using a numerical method, like the Newton method or the bisection method.

For each element of the power series distributions (geometric, Poisson, logarithmic and binomial), we have the following theorems for the MLE's:

\begin{theorem}\label{th.lb}
Let ${\rm g_{1}}(\beta;\gamma,\theta,{\boldsymbol x})$ denote the function on RHS of the expression in (\ref{eq.lb}), where $\gamma$ and $\theta$ are the true values of the
parameters. Then, for a given $\gamma >0$, and $\theta >0$, the roots of ${\rm g}_1(\beta ;\gamma,\theta ,{\boldsymbol x}) =0$, lies in the interval
\[\left(\frac{n}{\frac{\theta C''(\theta)}{C'(\theta)}+1}{(-\sum^n_{i =1}\log(p_i))^{ -1}}\ \ ,\ \ \ n{(-\sum^n_{i =1}\log(p_i))^{ -1}}\right),\]
\end{theorem}
\begin{proof}
See Appendix A.1.
\end{proof}

\begin{theorem}\label{th.lg}
Let ${\rm g_{2}}(\gamma;\beta,\theta,x)$ denote the function on RHS of the expression in (\ref{eq.lg}), where $\beta$ and $\theta$ are the true values of the
parameters. Then, the equation  ${\rm g_{2}}(\gamma;\beta,\theta,{\boldsymbol x})=0$ has at least one root if
$$n\bar{x}-\frac{\beta }{2}\sum^n_{i=1}{x^2_i(1+\frac{\theta e^{-\beta x_i}C''(\theta e^{-\beta x_i})}{C'\left(\theta e^{-\beta x_i}\right)})}>0.$$
\end{theorem}
\begin{proof}
See Appendix A.2.
\end{proof}

\begin{theorem}\label{th.lt}
Let ${\rm g}_{3}(\theta;\beta,\gamma,{\boldsymbol x})$ denote the function on RHS of the expression in (\ref{eq.lt}), where $\beta$ and $\gamma$ are the true values of the
parameters.\\
a. The equation  ${\rm g}_{3}(\theta;\beta,\gamma,{\boldsymbol x})=0$ has at least one root if for all GG, GP and GL  distributions $\sum\limits_{i=1}^{n}t_{i}>\frac{n}{2}$.\\
b. If  ${\rm g}_{3}(p;\beta,\gamma,{\boldsymbol x})=\frac{\partial l_{n}}{\partial p}$, where $p=\frac{\theta}{\theta+1}$ and $p\in (0,1)$ then the equation ${\rm g}_{3}(\theta;\beta,\gamma,{\boldsymbol x})=0$ has at least one root for GB  distribution if $\sum\limits_{i=1}^{n}t_{i}>\frac{n}{2}$ and  $\sum\limits_{i=1}^{n}\frac{1}{t_{i}}>\frac{nm}{1-m}$.
\end{theorem}
\begin{proof}
See Appendix A.3.
\end{proof}

\begin{theorem}
The pdf, $f(x|{\boldsymbol \Theta})$, of GPS distribution satisfies on the regularity condistions, i.e.
\begin{itemize}
  \item[i.] the support of $f(x|{\boldsymbol \Theta})$ does not depend on  ${\boldsymbol \Theta}$,
  \item[ii.] $f(x|{\boldsymbol \Theta})$ is twice continuously differentiable with respect to ${\boldsymbol \Theta}$,
  \item[iii.] the differentiation and integration are interchangeable in the sense that
  \end{itemize}
      $$
        \frac{\partial}{\partial {\boldsymbol \Theta}} \int_{-\infty}^{\infty}f(x|{\boldsymbol \Theta})dx= \int_{-\infty}^{\infty} \frac{\partial}{\partial {\boldsymbol \Theta}}f(x|{\boldsymbol \Theta})dx, \ \ \ \ \ \ \frac{\partial^2}{\partial {\boldsymbol \Theta}\partial{\boldsymbol \Theta}^T} \int_{-\infty}^{\infty}f(x|{\boldsymbol \Theta})dx= \int_{-\infty}^{\infty} \frac{\partial^2}{\partial {\boldsymbol \Theta}\partial{\boldsymbol \Theta}^T}f(x|{\boldsymbol \Theta})dx.
      $$

\end{theorem}
\begin{proof}
The proof is obvious and for more details, see \cite{ca-be-01} Section 10.
\end{proof}
%
%\begin{theorem}\label{th.lb}
%Under the regularity conditions, the MLE of ${\boldsymbol \Theta}$, say $\hat{\boldsymbol\Theta}$,  is asymptotically normal so that
%$$ \sqrt{n}(\hat{\boldsymbol\Theta}-{\boldsymbol \Theta})\longrightarrow N(0,I_n^{-1}(\hat{\boldsymbol\Theta}))$$
%\end{theorem}
%\begin{proof}
%See Lehmann,E.L. and Casella,G.(1998). Theory of Point estimation , 2nd Edition,Springer,New York.
%\end{proof}

\bigskip

Now, we derive  asymptotic confidence intervals for the parameters of GPS distribution. It is well-known that under regularity conditions \citep[see][Section 10]{ca-be-01}, the asymptotic distribution of
$\sqrt{n}(\hat{\boldsymbol\Theta}-{\boldsymbol\Theta})$ is multivariate normal with mean ${\boldsymbol 0}$ and variance-covariance matrix $J_n^{-1}({\boldsymbol\Theta})$, %the vector $\hat{\boldsymbol\Theta}$ has an asymptotic
where $J_n({\boldsymbol\Theta})=\lim_{n\rightarrow 0} I_n({\boldsymbol\Theta})$, and $I_n({\boldsymbol\Theta})$ is
the $3\times3$  observed information matrix, i.e.
\[I_n\left({\boldsymbol\Theta} \right)=-\left[ \begin{array}{ccc}
I_{\beta \beta} & I_{\beta \gamma } & I_{\beta \theta } \\
I_{\beta \gamma  } & I_{\gamma \gamma } & I_{\gamma \theta } \\
I_{\beta \theta  } & I_{\gamma \theta } & I_{\theta \theta}
\end{array} \right],\]
whose elements are given in Appendix B. Therefore, an $100(1-\alpha)$ asymptotic confidence interval for each parameter, ${\boldsymbol\Theta}_{r}$, is given by
\begin{equation}\label{eq.CI}
ACI_{r}=(\hat{\boldsymbol\Theta}_{r} -Z_{\alpha/2}\sqrt{\hat{I}_{rr}}, \hat{\boldsymbol\Theta}_{r}+Z_{\frac{\alpha}{2}}\sqrt{\hat{I}_{rr}}),
\end{equation}
 where $ \hat{I}_{rr}$ is the $(r,r)$  diagonal element of $I_{n}^{-1}(\hat{\boldsymbol\Theta})$ for $r=1,2,3$ and $Z_{\alpha/2}$ is the quantile
$\frac{\alpha}{2}$ of the standard normal distribution.

In some cases, a censoring time $C_i$ is assumed in collecting the lifetime data $X_i$, where  $C_i$ and   $X_i$ are independent. Suppose that the data consist of $n$ independent observations $x_i=\min(X_i, C_i)$ and $\delta_i=I(X_i\leq C_i)$ is such that $\delta_i=1$  if $X_i$ is a time to event and  $\delta_i=0$ if it is right censored for $i=1,\dots,n$. The censored likelihood function is
\begin{equation}\label{eq.Lce}
L_S({\boldsymbol\Theta})\varpropto \prod_{i=1}^n [f(x_i|{\boldsymbol\Theta})]^{\delta_i}[S(x_i|{\boldsymbol\Theta})]^{1-\delta_i},
\end{equation}
where $f(x_i|{\boldsymbol\Theta})$ and $S(x_i|{\boldsymbol\Theta})$ are the density function and survival function of GPS distribution. A similar  procedure to the above  can be used for constructing confidence interval for the parameters of the GPS model with a censoring  time.

\subsection{ EM-algorithm}

The EM algorithm is a very powerful tool in handling the incomplete data problem
 \citep[see][]{de-la-ru-77}.
%(Dempster et al., 1977).
It is an iterative method, and  there are two steps in each iteration: Expectation step or the E-step and the Maximization step or the M-step. The EM algorithm is especially useful if the complete data set is easy to analyze. In this Section, we develop an EM-algorithm for obtaining the MLE's for the  parameters of GPS distribution.

We define a hypothetical complete-data distribution with a joint probability density function in the form
$$
g(x_i,z_i;{\boldsymbol \Theta})=z_i \beta e^{\gamma x_i}e^{-\frac{z_i\beta}{\gamma}(e^{\gamma x_i}-1)}\frac{a_{z_i}\theta^{z_i}}{C(\theta)},
$$
where $\beta$, $\gamma$, $\theta>0$, $x_i>0$ and $z_i\in N$. Therefore, the log-likelihood for the complete-data is
\begin{eqnarray}\label{eq.ls}
l^{\ast}({\boldsymbol y};{\boldsymbol \Theta})\propto n\bar{z}\log(\theta)+n\log(\beta)+n\gamma \bar{x}-\frac{\beta}{\gamma}\sum\limits_{i=1}^{n}z_{i}( e^{\gamma x_{i}}-1)-n\log(C(\theta)),
\end{eqnarray}
where ${\boldsymbol y}=(x_1,...,x_n, z_1,...,z_n)$, $\bar{z}=n^{-1}\sum\limits_{i=1}^{n}z_{i}$, and $\bar{x}=n^{-1}\sum\limits_{i=1}^{n}x_{i}$.
On differentiation  \eqref{eq.ls} with respect to parameters $\beta$, $\gamma$, and $\theta$, we obtain the components of the score function,
$U({\boldsymbol y};{\boldsymbol\Theta})=(\frac{\partial l^{\ast}_{n}}{\partial\beta},\frac{\partial l^{\ast}_{n}}{\partial\gamma},\frac{\partial l^{\ast}_{n}}{\partial\theta})^{\prime}$, as
\begin{eqnarray*}
\frac{\partial l^{\ast}_{n}}{\partial\beta}&=&\frac{n}{\beta}-\frac{1}{\gamma}\sum\limits_{i=1}^{n}z_{i}( e^{\gamma x_{i}}-1),\\
\frac{\partial l^{\ast}_{n}}{\partial\gamma}&=&n\bar{x}+\frac{\beta}{\gamma^{2}}\sum\limits_{i=1}^{n} z_{i}(e^{\gamma x_{i}}-1)-\frac{\beta}{\gamma}\sum\limits_{i=1}^{n} z_{i}x_{i}e^{\gamma x_{i}},\\%=n\bar{x}+\frac{\beta}{\gamma^{2}}\sum\limits_{i=1}^{n} z_{i}(e^{\gamma x_{i}}-1-\gamma x_{i}e^{\gamma x_{i}}),\\
\frac{\partial l^{\ast}_{n}}{\partial\theta}&=&\frac{n\bar{z}}{\theta}-n\frac{C'(\theta)}{C(\theta)}.
\end{eqnarray*}

From a nonlinear system of equations $U({\boldsymbol y};{\boldsymbol\Theta})={\boldsymbol 0}$, we obtain the iterative procedure of the EM-algorithm as
\begin{eqnarray*}
&&\hat{\beta}^{(t+1)}=\frac{n\gamma^{(t)}}{\sum\limits_{i=1}^{n}\hat{z}_{i}^{(t)}(e^{\hat{\gamma }^{(t)} x_{i}}-1)},
\qquad \hat{\theta}^{(t+1)}=\frac{C(\hat{\theta}^{(t+1)})}{nC'(\hat{\theta}^{(t+1)})}\sum\limits_{i=1}^{n}\hat{z}_{i}^{(t)},\\
&&n\bar{x}(\hat{\gamma}^{(t+1)})^{2}+\hat{\beta}^{(t)}\sum\limits_{i=1}^{n}\hat{z}^{(t)}_{i}(e^{\hat{\gamma}^{(t+1)} x_{i}}-1)-\hat{\gamma}^{(t+1)}\hat{\beta}^{(t)}\sum\limits_{i=1}^{n}{\hat z}^{(t)}_{i} x_{i}e^{\hat{\gamma}^{(t+1)} x_{i}}=0,
\end{eqnarray*}
where $\hat{\theta}^{(t+1)}$ and $\hat{\gamma}^{(t+1)}$ are found numerically.
Here, for $i=1,2,...,n$, we have that
$$
\hat{z}_{i}^{(t)}=1+\frac{\hat{\theta}^{(t)}e^{-\frac{\hat{\beta}^{(t)}}{\hat{\gamma}^{(t)}}(e^{\hat{\gamma}^{(t)} x_{i}}-1)}C''(\hat{\theta}^{(t)}e^{-\frac{\hat{\beta}^{(t)}}{\hat{\gamma}^{(t)}}(e^{\hat{\gamma}^{(t)} x_{i}}-1)})}{C'(\hat{\theta}^{(t)}e^{-\frac{\hat{\beta}^{(t)}}{\hat{\gamma}^{(t)}}(e^{\hat{\gamma}^{(t)} x_{i}}-1)})}.
$$

In this part, we use the results of \cite{lou-82}
%Louis (1982)
to obtain the standard errors of the estimators from the EM-algorithm. The elements of the $3\times 3$ observed information matrix $I_{c}({\boldsymbol\Theta};{\boldsymbol y})=-[\frac{\partial U({\boldsymbol y};{\boldsymbol\Theta})}{\partial{\boldsymbol\Theta}}]$ are given by
\begin{eqnarray*}
&&-\frac{\partial^{2} l^{\ast}_{n}}{\partial\beta^{2}}=\frac{n}{\beta^{2}},
\quad-\frac{\partial^{2} l^{\ast}_{n}}{\partial\beta \partial\gamma}=-\frac{\partial^{2} l^{\ast}_{n}}{\partial\gamma \partial \beta}=-\frac{1}{\gamma^{2}}\sum\limits_{i=1}^{n} z_{i}(e^{\gamma x_{i}}-1)+\frac{1}{\gamma}\sum\limits_{i=1}^{n} z_{i}x_{i}e^{\gamma x_{i}},\\
&&\ \frac{\partial^{2} l^{\ast}_{n}}{\partial\beta \partial\theta}=\frac{\partial^{2} l^{\ast}_{n}}{\partial\theta\partial\beta}=\frac{\partial^{2} l^{\ast}_{n}}{\partial\theta\partial\gamma}=\frac{\partial^{2} l^{\ast}_{n}}{\partial\gamma\partial\theta}=0, \quad -\frac{\partial^{2} l^{\ast}_{n}}{\partial\theta^{2}}=\frac{ n\bar{z}}{\theta^{2}}+\frac{nC''(\theta)}{C(\theta)}-\frac{n(C'(\theta))^{2}}{(C(\theta))^{2}},\\
&&-\frac{\partial^{2} l^{\ast}_{n}}{\partial\gamma^{2}}=\frac{2\beta}{\gamma^{3}}\sum\limits_{i=1}^{n} z_{i}(e^{\gamma x_{i}}-1)-\frac{2\beta}{\gamma^2}\sum\limits_{i=1}^{n} z_{i}x_{i}e^{\gamma x_{i}}+\frac{\beta}{\gamma}\sum\limits_{i=1}^{n} z_{i}x_{i}^{2}e^{\gamma x_{i}}.
\end{eqnarray*}
Taking the conditional expectation of $I_c({\boldsymbol\Theta};{\boldsymbol y})$ given ${\boldsymbol x}$, we obtain the $3\times3$ matrix
\begin{eqnarray*}
{\mathcal I}_{c}({\boldsymbol\Theta};{\boldsymbol x})=E(I_{c}({\boldsymbol\Theta};{\boldsymbol y})|{\boldsymbol x})=[c_{ij}],
\end{eqnarray*}
where
\begin{eqnarray*}
&&c_{11}=\frac{n}{\beta^{2}}, \quad c_{12}=c_{21}=-\frac{1}{\gamma^{2}}\sum\limits_{i=1}^{n}E(Z_{i}|x_i)(e^{\gamma x_{i}}-1)+\frac{1}{\gamma}\sum\limits_{i=1}^{n}E(Z_{i}|x_i)x_{i}e^{\gamma x_{i}},\\
&&c_{13}=c_{31}=c_{23}=c_{32}=0, \quad
c_{33}=\frac{1}{\theta^{2}}\sum\limits_{i=1}^{n}E(Z_{i}|x_i)+\frac{nC''(\theta)}{C(\theta)}-\frac{n(C'(\theta))^{2}}{(C(\theta))^{2}},\\
&&c_{22}=\frac{2\beta}{\gamma^{3}}\sum\limits_{i=1}^{n} E(Z_{i}|x_i)(e^{\gamma x_{i}}-1)-\frac{2\beta}{\gamma^2}\sum\limits_{i=1}^{n}
E(Z_{i}|x_i)x_{i}e^{\gamma x_{i}}+\frac{\beta}{\gamma}\sum\limits_{i=1}^{n}E(Z_{i}|x_i)x_{i}^{2}e^{\gamma x_{i}},
\end{eqnarray*}
and
 $$E(Z_{i}|x_i)=1+\frac{\theta e^{-\frac{\beta}{\gamma}(e^{\gamma x_{i}}-1)} C''(\theta e^{-\frac{\beta}{\gamma}(e^{\gamma x_{i}}-1)})}{C'(\theta e^{-\frac{\beta}{\gamma}(e^{\gamma x_{i}}-1)})}.$$

Moving now to the computation of ${\mathcal I}_{m}(\Theta;\boldsymbol{x})$ as
$$
{\mathcal I}_{m}({\boldsymbol\Theta};{\boldsymbol x})=Var[U({\boldsymbol y};{\boldsymbol\Theta})|{\boldsymbol x}]=[v_{ij}],
$$
where
\begin{eqnarray*}
&&v_{11}=\frac{1}{\gamma^2}\sum\limits_{i=1}^{n}(e^{\gamma x_{i}}-1)^{2}Var(Z_{i}|x_i), \quad v_{13}=v_{31}=-\frac{1}{\gamma\theta}\sum\limits_{i=1}^{n}(e^{\gamma x_{i}}-1)Var(Z_{i}|x_i),\\
&&v_{12}=v_{21}=-\frac{\beta}{\gamma^{3}}\sum\limits_{i=1}^{n}(e^{\gamma x_{i}}-1)(e^{\gamma x_{i}}-1-\gamma x_{i}e^{\gamma x_{i}})
Var(Z_{i}|x),\\
&&v_{22}=\frac{\beta^2}{\gamma^{4}}\sum\limits_{i=1}^{n}(
e^{\gamma x_{i}}-1-\gamma x_{i}e^{\gamma x_{i}})^2Var(Z_{i}|x_i),\\
&&v_{23}=v_{32}=\frac{\beta}{\theta\gamma^{2}}\sum\limits_{i=1}^{n}(e^{\gamma x_{i}}-1-\gamma x_{i}e^{\gamma x_{i}})Var(Z_{i}|x_i), \quad
v_{33}=\frac{1}{\theta^{2}}\sum\limits_{i=1}^{n}Var(Z_{i}|x_i),
\end{eqnarray*}
and
\begin{eqnarray}
Var(Z|x)&=&E(Z^{2}|x)-(E(Z|x))^{2}\nonumber\\
&=&\frac{1}{C'(\theta_{\ast})}\sum\limits_{z=1}^{\infty}a_{z}z^{3}\theta_{\ast}^{z-1}-\frac{1}{[C'(\theta_{\ast})]^{2}}[C'(\theta_{\ast})+\theta_{\ast}C''(\theta_{\ast})]^{2}\nonumber\\
&=&\frac{1}{C'(\theta_{\ast})}[\theta_{\ast}^{2}C'''(\theta_{\ast})+C'(\theta_{\ast})+3\theta_{\ast}
C''(\theta_{\ast})]-\frac{1}{[C'(\theta_{\ast})]^{2}}[C'(\theta_{\ast})+\theta_{\ast}C''(\theta_{\ast})]^{2},
\nonumber
\end{eqnarray}
in which $\theta_{\ast}=\theta e^{-\frac{\beta}{\gamma}(e^{\gamma x}-1)}$. Therefore, we obtain the observed information as
$$
I(\hat{{\boldsymbol\Theta}};{\boldsymbol x})=\mathcal{I}_{c}(\hat{\boldsymbol\Theta};{\boldsymbol x})-\mathcal{I}_{m}(\hat{\boldsymbol\Theta};{\boldsymbol x}).
$$
The standard errors of the MLE's of the EM-algorithm are the square root of the diagonal elements of the $I^{-1}(\hat{\boldsymbol\Theta};{\boldsymbol x})$.

\section{Simulation}
\label{sec.sim}

This section presents the results of three simulation studies. First,
 a simulation study is performed for evaluation of  parameter estimation based on the EM algorithm. No restriction has been imposed on the maximum number of iterations and convergence is assumed when the absolute difference between successive estimates are less that $10^{-4}$.

Here, we consider the GG distribution and generate $N=1000$ random samples with different set of parameters  for $n=30, 50, 100, 200$. In each random sample, the estimation of parameters as well as the Fisher information matrix are obtained. Then, the average value of estimations (AE), mean square errors (MSE), variance of estimations (VS), the average value of inverse of Fisher information (EF) matrices, and coverage probabilities (CP) of the 95\% confidence interval in \eqref{eq.CI}  are computed.   The results are given in Table \ref{tab.sim1}, and we can conclude that
%$$
%\hat{\theta}^{(t+1)}=\frac{C(\hat{\theta}^{(t+1)})}{nC'(\hat{\theta}^{(t+1)})}\sum\limits_{i=1}^{n}\hat{z}_{i}^{(t)}
%\ \ \ \ \ \ \ \ \
%\hat{\theta}^{(t+1)}=1-\frac{n}{\sum\limits_{i=1}^{n}\hat{z}_{i}^{(t)}}$$

\begin{sidewaystable}
%\begin{table}[h]
\begin{center}
%{\small
\caption{The average MLE's, mean square errors, variance of estimations, the average value of Fisher information, and coverage probability based on
 EM estimators for GG distribution}\label{tab.sim1}
\begin{tabular}{|c|ccc|ccc|ccc|ccc|ccc|ccc|} \hline
  &      \multicolumn{3}{|c|}{Parameter}       &  \multicolumn{3}{|c|}{AE}  & \multicolumn{3}{|c|}{MSE}
   &\multicolumn{3}{|c|}{VS} & \multicolumn{3}{|c|}{EF} &  \multicolumn{3}{|c|}{CP}  \\ \hline
$n$ & $\beta$ & $\gamma$ & $\theta$ & $\hat{\beta}$ & $\hat{\gamma}$ & $\hat{\theta}$  & $\hat{\beta}$ & $\hat{\gamma}$ & $\hat{\theta}$ &
 $\hat{\beta}$ & $\hat{\gamma}$ & $\hat{\theta}$ & $\hat{\beta}$ & $\hat{\gamma}$ & $\hat{\theta}$ &  ${\beta}$ & ${\gamma}$ & ${\theta}$  \\ \hline
30 & 0.5 & 2.0 & 0.9 & 0.490 & 2.760 & 0.891 & 0.914 & 4.601 & 0.466 & 0.102 & 1.553 & 0.008 & 1.748 & 6.111 & 0.089 & 0.90 & 0.94 & 0.91 \\
50 & 0.5 & 2.0 & 0.9 & 0.458 & 2.582 & 0.903 & 0.939 & 3.538 & 0.460 & 0.092 & 1.084 & 0.005 & 1.593 & 4.836 & 0.076 & 0.92 & 0.94 & 0.91 \\
100 & 0.5 & 2.0 & 0.9 & 0.446 & 2.451 & 0.908 & 1.004 & 2.796 & 0.461 & 0.108 & 0.666 & 0.005 & 0.882 & 2.723 & 0.041 & 0.95 & 0.95 & 0.95 \\
200 & 0.5 & 2.0 & 0.9 & 0.470 & 2.283 & 0.904 & 0.941 & 2.170 & 0.454 & 0.112 & 0.442 & 0.005 & 0.679 & 1.716 & 0.027 & 0.95 & 0.95 & 0.96 \\ \hline

30 & 0.5 & 2.0 & 0.1 & 0.406 & 2.711 & 0.207 & 1.020 & 4.774 & 0.994 & 0.039 & 0.745 & 0.107 & 1.774 & 5.318 & 5.722 & 0.89 & 0.96 & 0.87 \\
50 & 0.5 & 2.0 & 0.1 & 0.427 & 2.587 & 0.187 & 1.004 & 4.149 & 1.004 & 0.039 & 0.531 & 0.102 & 1.143 & 3.119 & 3.222 & 0.90 & 0.94 & 0.88 \\
100 & 0.5 & 2.0 & 0.1 & 0.457 & 2.418 & 0.131 & 0.951 & 3.371 & 1.030 & 0.032 & 0.311 & 0.086 & 1.790 & 5.490 & 6.589 & 0.92 & 0.96 & 0.90 \\
200 & 0.5 & 2.0 & 0.1 & 0.485 & 2.300 & 0.192 & 0.914 & 2.948 & 1.036 & 0.027 & 0.211 & 0.076 & 0.835 & 2.103 & 2.602 & 0.92 & 0.95 & 0.92 \\ \hline

30 & 1.0 & 2.0 & 0.9 & 0.859 & 3.441 & 0.915 & 0.764 & 8.083 & 0.401 & 0.213 & 3.115 & 0.005 & 4.490 & 14.220 & 0.060 & 0.91 & 0.93 & 0.93 \\
50 & 1.0 & 2.0 & 0.9 & 0.911 & 3.123 & 0.913 & 0.924 & 5.636 & 0.399 & 0.466 & 2.097 & 0.006 & 4.339 & 9.036 & 0.051 & 0.91 & 0.94 & 0.92 \\
100 & 1.0 & 2.0 & 0.9 & 0.913 & 2.684 & 0.903 & 1.223 & 3.474 & 0.417 & 0.854 & 1.385 & 0.011 & 3.445 & 5.108 & 0.042 & 0.92 & 0.96 & 0.92 \\
200 & 1.0 & 2.0 & 0.9 & 1.033 & 2.378 & 0.893 & 1.234 & 2.369 & 0.420 & 0.964 & 1.024 & 0.011 & 2.393 & 3.437 & 0.027 & 0.92 & 0.95 & 0.93 \\ \hline

30 & 1.0 & 2.0 & 0.1 & 0.261 & 2.972 & 0.274 & 0.962 & 5.359 & 1.006 & 0.128 & 0.998 & 0.088 & 6.823 & 10.343 & 6.991 & 0.89 & 0.93 & 0.91 \\
50 & 1.0 & 2.0 & 0.1 & 0.214 & 2.814 & 0.228 & 0.912 & 4.528 & 1.057 & 0.133 & 0.759 & 0.089 & 4.393 & 5.565 & 3.258 & 0.89 & 0.92 & 0.91 \\
100 & 1.0 & 2.0 & 0.1 & 0.185 & 2.556 & 0.179 & 0.824 & 3.360 & 1.103 & 0.125 & 0.462 & 0.083 & 2.599 & 3.426 & 2.167 & 0.91 & 0.94 & 0.93 \\
200 & 1.0 & 2.0 & 0.1 & 0.155 & 2.411 & 0.117 & 0.771 & 2.829 & 1.173 & 0.107 & 0.336 & 0.076 & 1.841 & 2.406 & 1.618 & 0.92 & 0.93 & 0.93 \\ \hline
\end{tabular}
%}
\end{center}
%\end{table}
\end{sidewaystable}

\begin{sidewaystable}
%\begin{table}[h]
\begin{center}
%{\small
\caption{The average MLE's, mean square errors, variance of estimations, the average value of Fisher information, and coverage probability based on
 MLE estimators for GG distribution with censored data}\label{tab.sim2}
\begin{tabular}{|c|ccc|ccc|ccc|ccc|ccc|ccc|} \hline
  &      \multicolumn{3}{|c|}{Parameter}       &  \multicolumn{3}{|c|}{AE}  & \multicolumn{3}{|c|}{MSE}
   &\multicolumn{3}{|c|}{VS} & \multicolumn{3}{|c|}{EF} &  \multicolumn{3}{|c|}{CP}  \\ \hline
$n$ & $\beta$ & $\gamma$ & $\theta$ & $\hat{\beta}$ & $\hat{\gamma}$ & $\hat{\theta}$  & $\hat{\beta}$ & $\hat{\gamma}$ & $\hat{\theta}$ &
 $\hat{\beta}$ & $\hat{\gamma}$ & $\hat{\theta}$ & $\hat{\beta}$ & $\hat{\gamma}$ & $\hat{\theta}$ &  ${\beta}$ & ${\gamma}$ & ${\theta}$  \\ \hline

30 & 0.5 & 2.0 & 0.9 & 1.141 & 2.536 & 0.705 & 1.613 & 2.092 & 0.132 & 1.203 & 1.807 & 0.094 & 16.115 & 14.718 & 5.692 & 0.86 & 0.93 & 0.84 \\
50 & 0.5 & 2.0 & 0.9 & 0.898 & 2.300 & 0.778 & 0.866 & 1.431 & 0.071 & 0.709 & 1.342 & 0.056 & 8.758 & 6.669 & 2.715 & 0.88 & 0.94 & 0.87 \\
100 & 0.5 & 2.0 & 0.9 & 0.789 & 2.062 & 0.822 & 0.651 & 0.862 & 0.040 & 0.568 & 0.859 & 0.034 & 7.289 & 5.956 & 2.069 & 0.88 & 0.95 & 0.88 \\
200 & 0.5 & 2.0 & 0.9 & 0.670 & 2.042 & 0.856 & 0.326 & 0.633 & 0.018 & 0.297 & 0.632 & 0.016 & 4.143 & 5.705 & 1.008 & 0.91 & 0.95 & 0.90 \\ \hline

30 & 0.5 & 2.0 & 0.1 & 0.399 & 2.367 & 0.232 & 0.059 & 0.568 & 0.138 & 0.049 & 0.434 & 0.121 & 4.323 & 8.031 & 6.697 & 0.91 & 0.92 & 0.90 \\
50 & 0.5 & 2.0 & 0.1 & 0.389 & 2.343 & 0.260 & 0.061 & 0.485 & 0.154 & 0.049 & 0.368 & 0.128 & 3.002 & 6.738 & 4.966 & 0.91 & 0.92 & 0.91 \\
100 & 0.5 & 2.0 & 0.1 & 0.386 & 2.304 & 0.293 & 0.066 & 0.422 & 0.174 & 0.053 & 0.330 & 0.137 & 3.390 & 4.996 & 4.715 & 0.89 & 0.93 & 0.89 \\
200 & 0.5 & 2.0 & 0.1 & 0.377 & 2.316 & 0.307 & 0.063 & 0.363 & 0.179 & 0.048 & 0.263 & 0.136 & 1.195 & 4.790 & 3.744 & 0.89 & 0.95 & 0.90 \\ \hline

30 & 1.0 & 2.0 & 0.9 & 1.995 & 2.849 & 0.746 & 4.086 & 3.014 & 0.091 & 3.098 & 2.294 & 0.068 & 16.988 & 18.095 & 3.081 & 0.87 & 0.90 & 0.85 \\
50 & 1.0 & 2.0 & 0.9 & 1.667 & 2.594 & 0.799 & 2.727 & 2.201 & 0.052 & 2.284 & 1.851 & 0.042 & 15.772 & 17.097 & 2.365 & 0.86 & 0.91 & 0.84 \\
100 & 1.0 & 2.0 & 0.9 & 1.308 & 2.208 & 0.851 & 1.428 & 1.120 & 0.025 & 1.335 & 1.078 & 0.022 & 13.458 & 15.288 & 1.799 & 0.85 & 0.94 & 0.82 \\
200 & 1.0 & 2.0 & 0.9 & 1.191 & 2.064 & 0.871 & 0.902 & 0.552 & 0.013 & 0.866 & 0.548 & 0.012 & 10.675 & 9.313 & 1.069 & 0.88 & 0.93 & 0.87 \\ \hline

30 & 1.0 & 2.0 & 0.1 & 0.616 & 2.994 & 0.367 & 0.296 & 2.448 & 0.215 & 0.149 & 1.460 & 0.144 & 8.004 & 9.443 & 6.819 & 0.84 & 0.93 & 0.84 \\
50 & 1.0 & 2.0 & 0.1 & 0.628 & 2.882 & 0.389 & 0.308 & 2.110 & 0.232 & 0.170 & 1.333 & 0.148 & 6.905 & 4.961 & 5.584 & 0.80 & 0.95 & 0.81 \\
100 & 1.0 & 2.0 & 0.1 & 0.630 & 2.816 & 0.407 & 0.325 & 1.762 & 0.253 & 0.188 & 1.098 & 0.159 & 5.158 & 4.707 & 4.130 & 0.74 & 0.89 & 0.75 \\
200 & 1.0 & 2.0 & 0.1 & 0.631 & 2.723 & 0.413 & 0.326 & 1.331 & 0.260 & 0.190 & 0.809 & 0.162 & 5.089 & 3.507 & 3.174 & 0.73 & 0.91 & 0.74 \\ \hline

\end{tabular}
%}
\end{center}
%\end{table}
\end{sidewaystable}

%\begin{sidewaystable}
\begin{table}[]
\begin{center}
%{\small
\caption{The number of cases that the criteria value of fitted distribution is smaller than the criteria value of fitted GG distribution}\label{tab.sim3}
\begin{tabular}{|c|ccc|cccc|cccc|cccc|} \hline
 & \multicolumn{3}{|c|}{Parameter} & \multicolumn{4}{|c|}{AIC} & \multicolumn{4}{|c|}{AICC} & \multicolumn{4}{|c|}{BIC} \\ \hline
$n$ & $\beta$ & $\gamma$ & $\theta$ & Gompertz & GP & GB & GL & Gompertz & GP & GB & GL & Gompertz & GP & GB & GL \\ \hline
30 & 0.5 & 2.0 & 0.9 & 761 & 81 & 71 & 457 & 821 & 81 & 71 & 457 & 902 & 81 & 71 & 457 \\
50 & 0.5 & 2.0 & 0.9 & 648 & 95 & 84 & 419 & 703 & 95 & 84 & 419 & 901 & 95 & 84 & 419 \\
100 & 0.5 & 2.0 & 0.9 & 360 & 60 & 48 & 375 & 387 & 60 & 48 & 375 & 776 & 60 & 48 & 375 \\
200 & 0.5 & 2.0 & 0.9 & 146 & 39 & 30 & 363 & 152 & 39 & 30 & 363 & 541 & 39 & 30 & 363 \\ \hline

30 & 0.5 & 2.0 & 0.1 & 945 & 18 & 24 & 350 & 959 & 18 & 24 & 350 & 978 & 18 & 24 & 350 \\
50 & 0.5 & 2.0 & 0.1 & 933 & 19 & 41 & 386 & 946 & 19 & 41 & 386 & 986 & 19 & 41 & 386 \\
100 & 0.5 & 2.0 & 0.1 & 917 & 23 & 52 & 418 & 924 & 23 & 52 & 418 & 990 & 23 & 52 & 418 \\
200 & 0.5 & 2.0 & 0.1 & 894 & 33 & 73 & 408 & 899 & 33 & 73 & 408 & 988 & 33 & 73 & 408 \\ \hline

30 & 1.0 & 2.0 & 0.9 & 492 & 123 & 102 & 460 & 588 & 123 & 102 & 460 & 706 & 123 & 102 & 460 \\
50 & 1.0 & 2.0 & 0.9 & 279 & 139 & 120 & 397 & 308 & 139 & 120 & 397 & 539 & 139 & 120 & 397 \\
100 & 1.0 & 2.0 & 0.9 & 84 & 112 & 82 & 363 & 89 & 112 & 82 & 363 & 282 & 112 & 82 & 363 \\
200 & 1.0 & 2.0 & 0.9 & 12 & 61 & 43 & 328 & 15 & 61 & 43 & 328 & 80 & 61 & 43 & 328 \\ \hline

30 & 1.0 & 2.0 & 0.1 & 965 & 25 & 28 & 333 & 973 & 25 & 28 & 333 & 984 & 25 & 28 & 333 \\
50 & 1.0 & 2.0 & 0.1 & 954 & 16 & 32 & 352 & 965 & 16 & 32 & 352 & 997 & 16 & 32 & 352 \\
100 & 1.0 & 2.0 & 0.1 & 954 & 25 & 64 & 364 & 958 & 25 & 64 & 364 & 992 & 25 & 64 & 364 \\
200 & 1.0 & 2.0 & 0.1 & 921 & 36 & 64 & 387 & 927 & 36 & 64 & 387 & 994 & 36 & 64 & 387 \\ \hline

\end{tabular}
%}
\end{center}
\end{table}
%\end{sidewaystable}

%\begin{sidewaystable}
\begin{table}
\begin{center}
%{\small
\caption{The number of cases that the criteria value of fitted distribution is smaller than the criteria value of fitted Gompertz distribution}\label{tab.sim4}
\begin{tabular}{|c|ccc|cccc|cccc|cccc|} \hline
 & \multicolumn{3}{|c|}{Parameter} & \multicolumn{4}{|c|}{AIC} & \multicolumn{4}{|c|}{AICC} & \multicolumn{4}{|c|}{BIC} \\ \hline
$n$ & $\beta$ & $\gamma$ & $\theta$ & GG & GP & GB & GL & GG & GP & GB & GL & GG & GP & GB & GL \\ \hline

30 & 0.5 & 2.0 & 0.9 & 54 & 13 & 6 & 231 & 34 & 7 & 5 & 180 & 16 & 5 & 1 & 92 \\ 
50 & 0.5 & 2.0 & 0.9 & 71 & 39 & 21 & 173 & 58 & 25 & 14 & 156 & 20 & 3 & 1 & 68 \\ 
100 & 0.5 & 2.0 & 0.9 & 59 & 39 & 34 & 161 & 53 & 37 & 32 & 150 & 12 & 3 & 4 & 47 \\ 
200 & 0.5 & 2.0 & 0.9 & 68 & 68 & 69 & 145 & 65 & 64 & 62 & 141 & 4 & 3 & 1 & 25 \\ \hline

30 & 0.5 & 2.0 & 0.1 & 47 & 8 & 1 & 108 & 28 & 3 & 0 & 81 & 12 & 0 & 0 & 44 \\ 
50 & 0.5 & 2.0 & 0.1 & 57 & 6 & 3 & 122 & 49 & 5 & 1 & 110 & 13 & 0 & 0 & 42 \\ 
100 & 0.5 & 2.0 & 0.1 & 97 & 13 & 21 & 116 & 86 & 10 & 15 & 115 & 13 & 0 & 1 & 25 \\ 
200 & 0.5 & 2.0 & 0.1 & 129 & 35 & 31 & 95 & 122 & 35 & 30 & 95 & 17 & 1 & 0 & 20 \\ \hline

30 & 1.0 & 2.0 & 0.9 & 52 & 7 & 3 & 48 & 39 & 0 & 0 & 43 & 18 & 0 & 0 & 29 \\ 
50 & 1.0 & 2.0 & 0.9 & 53 & 20 & 10 & 37 & 44 & 10 & 5 & 34 & 13 & 0 & 0 & 17 \\ 
100 & 1.0 & 2.0 & 0.9 & 74 & 30 & 40 & 25 & 68 & 24 & 32 & 24 & 10 & 1 & 0 & 8 \\ 
200 & 1.0 & 2.0 & 0.9 & 76 & 42 & 55 & 7 & 71 & 37 & 52 & 7 & 5 & 2 & 1 & 3 \\ \hline

30 & 1.0 & 2.0 & 0.1 & 34 & 0 & 1 & 93 & 20 & 0 & 0 & 79 & 5 & 0 & 0 & 51 \\ 
50 & 1.0 & 2.0 & 0.1 & 47 & 3 & 0 & 88 & 42 & 2 & 0 & 73 & 15 & 0 & 0 & 26 \\ 
100 & 1.0 & 2.0 & 0.1 & 65 & 0 & 1 & 100 & 61 & 0 & 1 & 89 & 8 & 0 & 0 & 21 \\ 
200 & 1.0 & 2.0 & 0.1 & 86 & 4 & 4 & 91 & 80 & 4 & 4 & 86 & 13 & 0 & 0 & 8 \\ \hline

\end{tabular}
%}
\end{center}
\end{table}
%\end{sidewaystable}

\newpage

\noindent i) convergence has been achieved in all cases and this emphasizes the
numerical stability of the EM-algorithm,
ii) the differences between the average estimates and the true values are almost small,
iii) the MSE,  variance of estimations, and variance based on Fisher information matrices  decrease when the sample size increases,
iv) the coverage probabilities of the confidence intervals for the parameters based on asymptotic approach are satisfactory and especially are close to the confidence coefficient, $0.95$ when the
sample size large.

In the second simulation, we  consider the GG distribution and generate $N=1000$ random samples with different set of parameters  for $n=30, 50, 100, 200$ and censoring percentage $p=0.3$. Using the censored likelihood function in \eqref{eq.Lce}, we obtained the MLE of parameters as well as the Fisher information matrix. Then, the AE, MSE, VS, EF matrices, and CP of the 95\% confidence interval are computed.   The results are given in Table \ref{tab.sim2}, and conclusions are similar to the first simulation. Only,  the variances based the average value of Fisher information matrix are very large.

At the end,   we performed a simulation study directed to model misspecification. We  consider the GG distribution and generate $N=1000$ random samples with different set of parameters  for $n=30, 50, 100, 200$. In each sample, considered distributions (Gompertz, GG, GP, GB with $m=5$, GL)  were fitted.
 The MLE of parameters, and then AIC (Akaike Information Criterion), AICC (AIC with correction)  and BIC (Bayesian Information Criterion) are calculated.
Using each criteria (AIC, AICC, BIC), the preferred distribution is the one with the smaller value. We computed the cases that the Gompertz, GP, GB, and GL distributions were preferred with respect to GG distribution. The results are given in Table \ref{tab.sim3} and we can conclude that when the real model is GG distribution 
i) it is usually possible to discriminate between GG distribution and three subclasses of GPS (GP, GB and GL), ii) when the sample size is large and the parameter $\theta$ far from away from 0, we can discriminate between GG distribution and Gompertz distribution. In fact, when $\theta$ is close to 0, the GPS model becomes to the Gompertz distribution (See Proposition 2).

Also, we study model misspecification using generating random sample from the Gompertz distribution and  computed the cases that the GG, GP, GB, and GL distributions were preferred with respect to Gompertz distribution. The results are given in Table \ref{tab.sim4} and we can conclude that
 it is usually possible to discriminate between Gompertz distribution and the subclasses of GPS (GG, GP, GB and GL) when the real model is Gompertz distribution.

\section{A numerical example}
\label{sec.ex}

In this Section, we consider a real data set and  fit the Gompertz, GG, GP, GB (with $m=5$), and GL distributions. The data obtained from
\cite{sm-na-87}
%Smith and Naylor (1987)
represent the strengths of 1.5 cm glass fibres, measured at the National Physical Laboratory, England. This data is also studied by
\cite{ba-sa-co-10}:
%Barreto-Souza et al. (2010)

            0.55, 0.93, 1.25, 1.36, 1.49, 1.52, 1.58, 1.61, 1.64, 1.68, 1.73, 1.81, 2.00, 0.74, 1.04, 1.27,

            1.39, 1.49, 1.53, 1.59, 1.61, 1.66, 1.68, 1.76, 1.82, 2.01, 0.77, 1.11, 1.28, 1.42, 1.50, 1.54,

            1.60, 1.62, 1.66, 1.69, 1.76, 1.84, 2.24, 0.81, 1.13, 1.29, 1.48, 1.50, 1.55, 1.61, 1.62, 1.66,

            1.70, 1.77, 1.84, 0.84, 1.24, 1.30, 1.48, 1.51, 1.55, 1.61, 1.63, 1.67, 1.70, 1.78, 1.89.

The MLE's of the parameters (with standard errors) for the distributions are given in Table \ref{table.EX}. Note that the MLE of $\theta$ for GL distribution is very close to 0. Therefore, the MLE's of the GL and Gompertz distributions are very close. In this table, we also consider the estimation of parameters for three parameters Weibull distribution (TW) with the following density function which is considered by \cite{sm-na-87}
$$
f_{TW}(x)=\lambda \gamma (x-\theta)^{\gamma-1} \exp(-\lambda(x-\theta)^\gamma), \ \ \ \ x>\theta,  \ \ \lambda>0, \ \  \gamma>0, \ \ \theta\in R.
$$
We give the estimation of $\beta=\log(\lambda)$ for the TW distribution because the MLE of $\lambda$ is very close to 0.

To test the goodness-of-fit of the distributions, we calculated the maximized log-likelihood, the Kolmogorov-Smirnov (K-S) statistic
with its respective p-value, the AIC, %(Akaike Information Criterion),
 AICC %(AIC with correction)
 and BIC %(Bayesian Information Criterion)
 for the six distributions. The results show that the GG distribution yields the best fit among the TW,  GP, GB, GL, and Gompertz distributions. Also, the GG, GP, and GB distribution are better than Gompertz and TW distributions. The plots of the densities (together with the data histogram)  and cumulative distribution functions in Figure \ref{plot.EX} confirm this conclusion. Also, Plots of the QQ-plot of fitted distributions are given in  Figure \ref{fig.qex}.

\begin{figure}
\vspace{-0.0cm}
\centering
\includegraphics[scale=0.53]{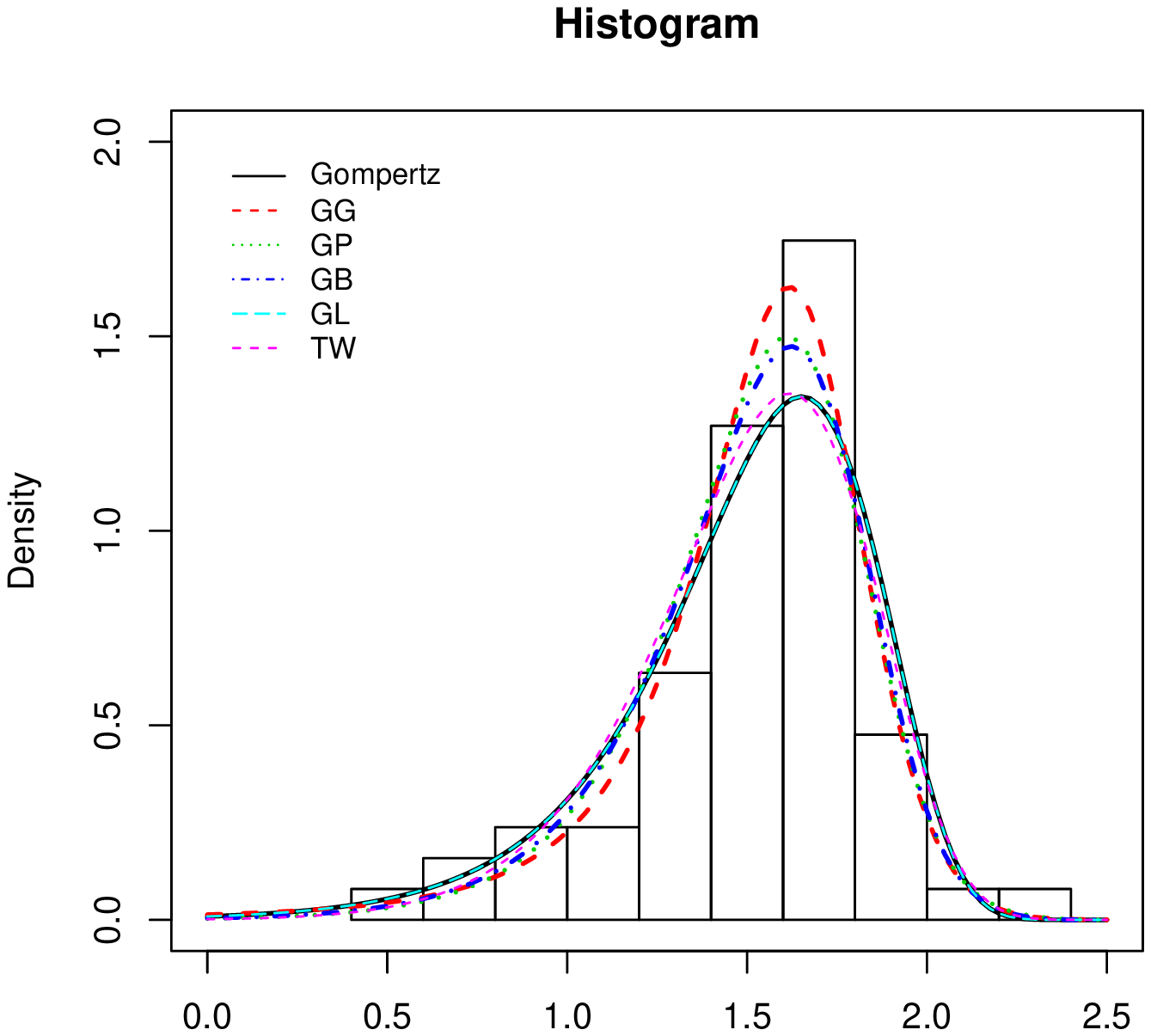}
\includegraphics[scale=0.53]{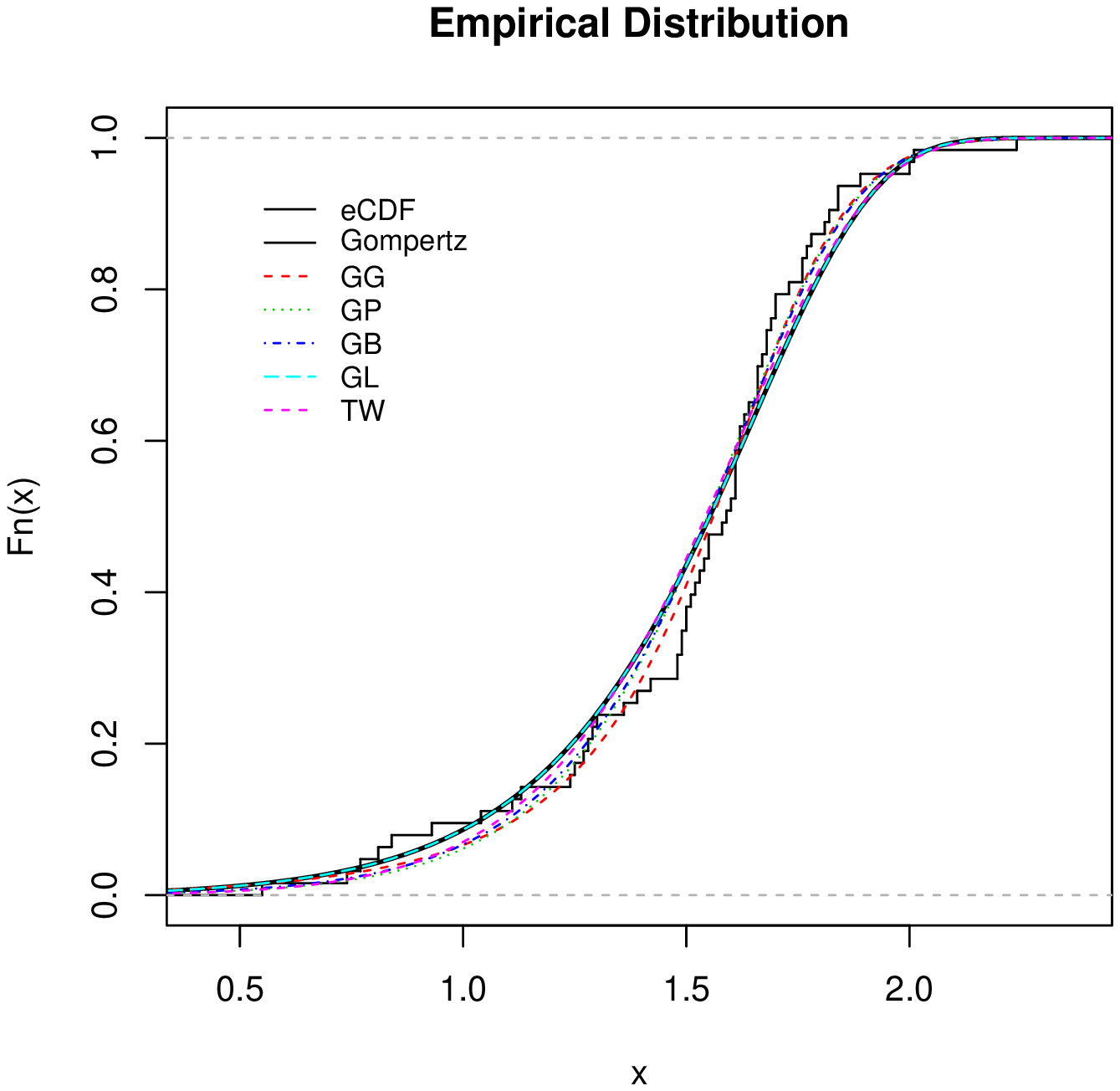}
\vspace{-0.8cm}
\caption[]{Plots (density and distribution) of  fitted  Gompertz, GG, GP, GB, GL and TW  distributions for the data set.}\label{plot.EX}

%\vspace{1cm}
\end{figure}

\begin{figure}
\centering
\includegraphics[scale=0.35]{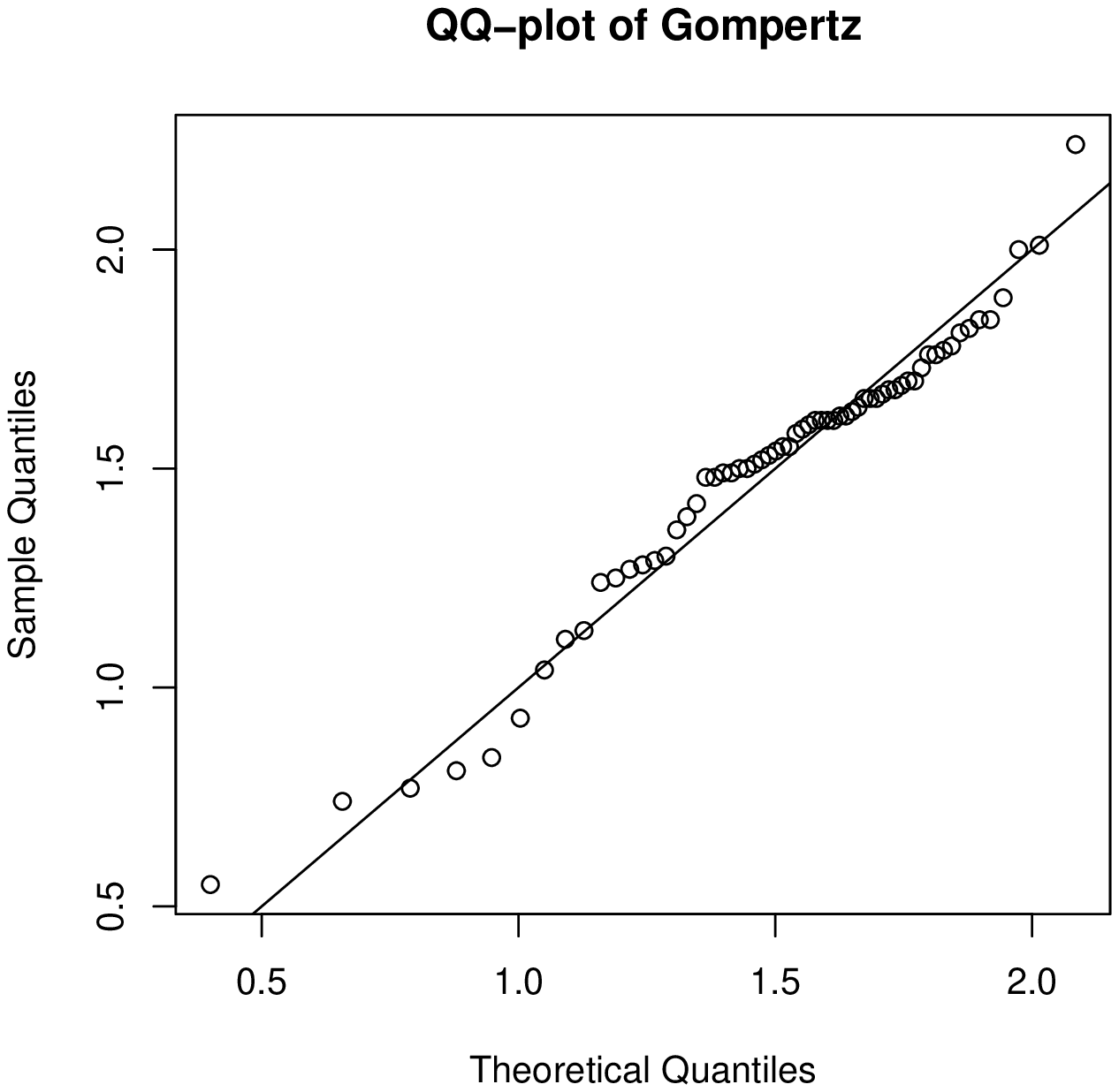}
\includegraphics[scale=0.35]{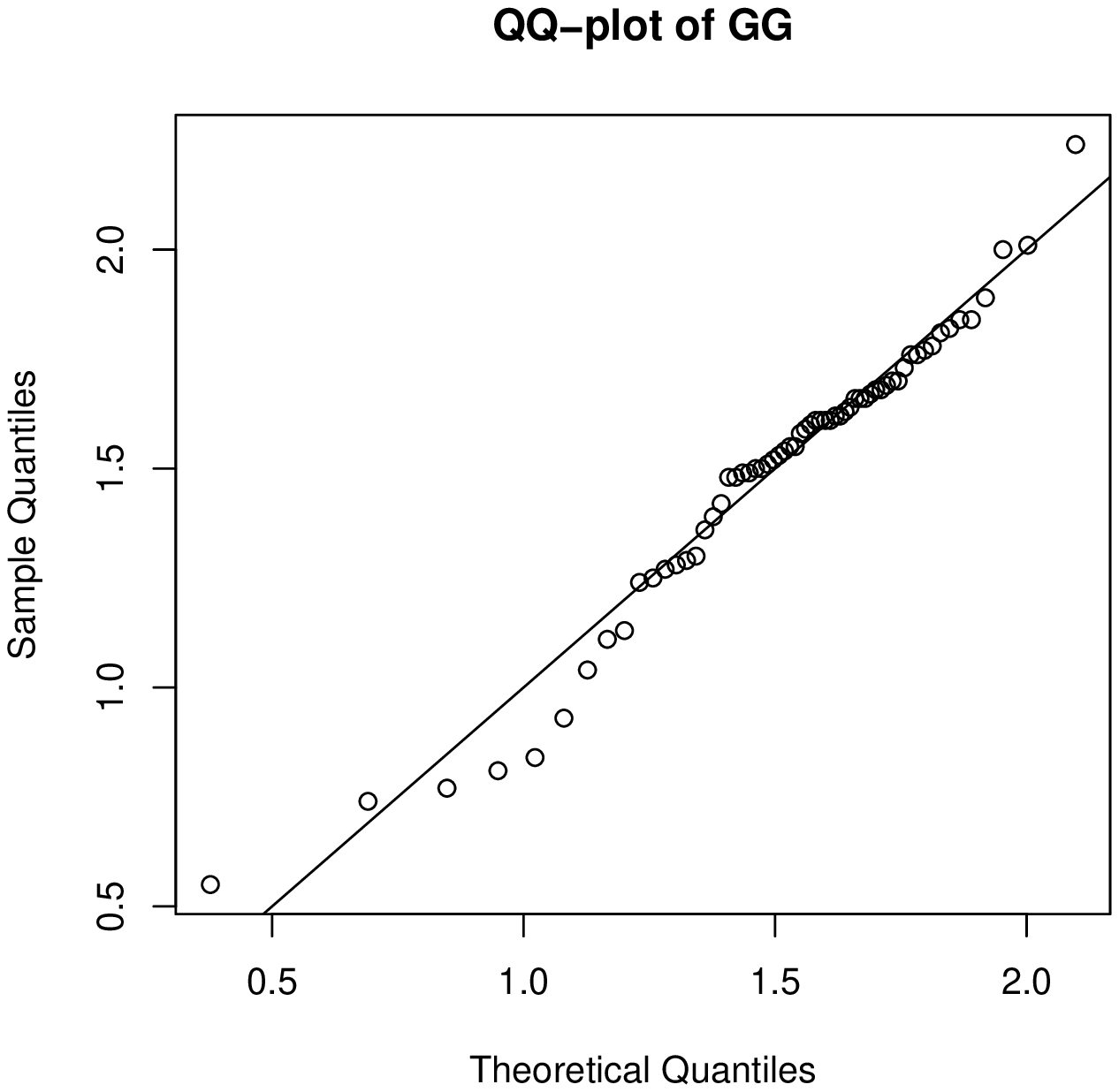}
\includegraphics[scale=0.35]{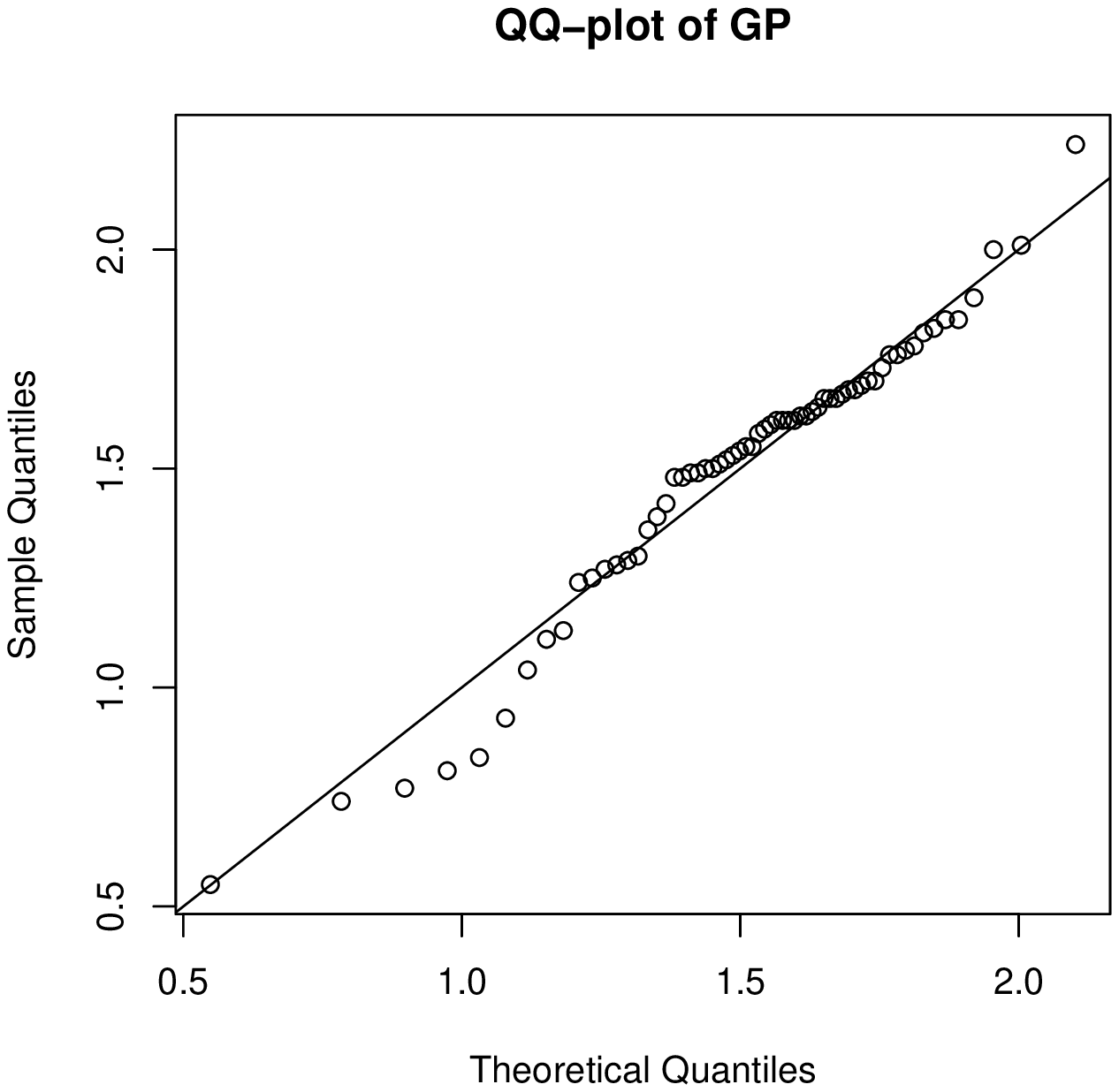}
\includegraphics[scale=0.35]{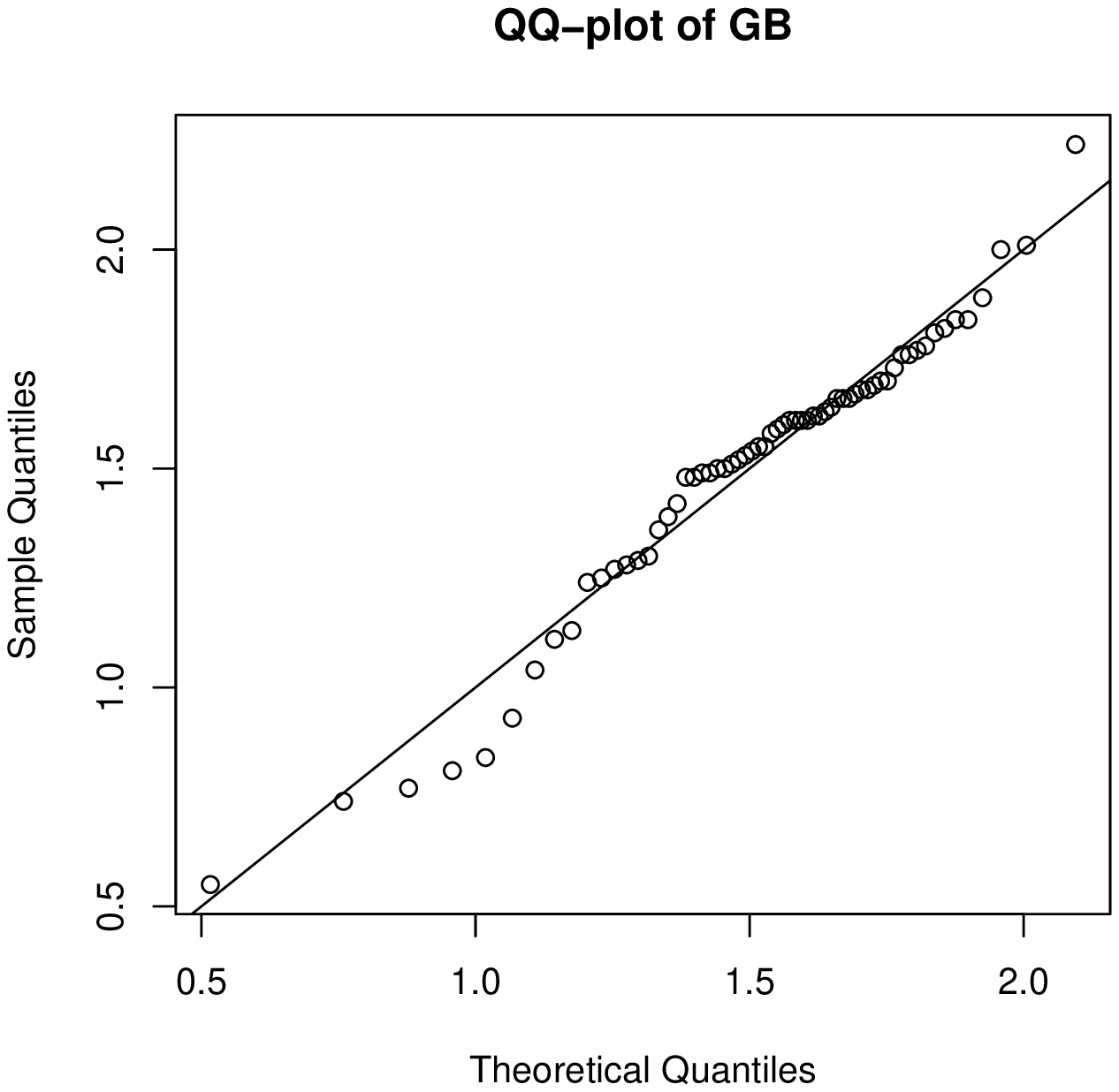}
\includegraphics[scale=0.35]{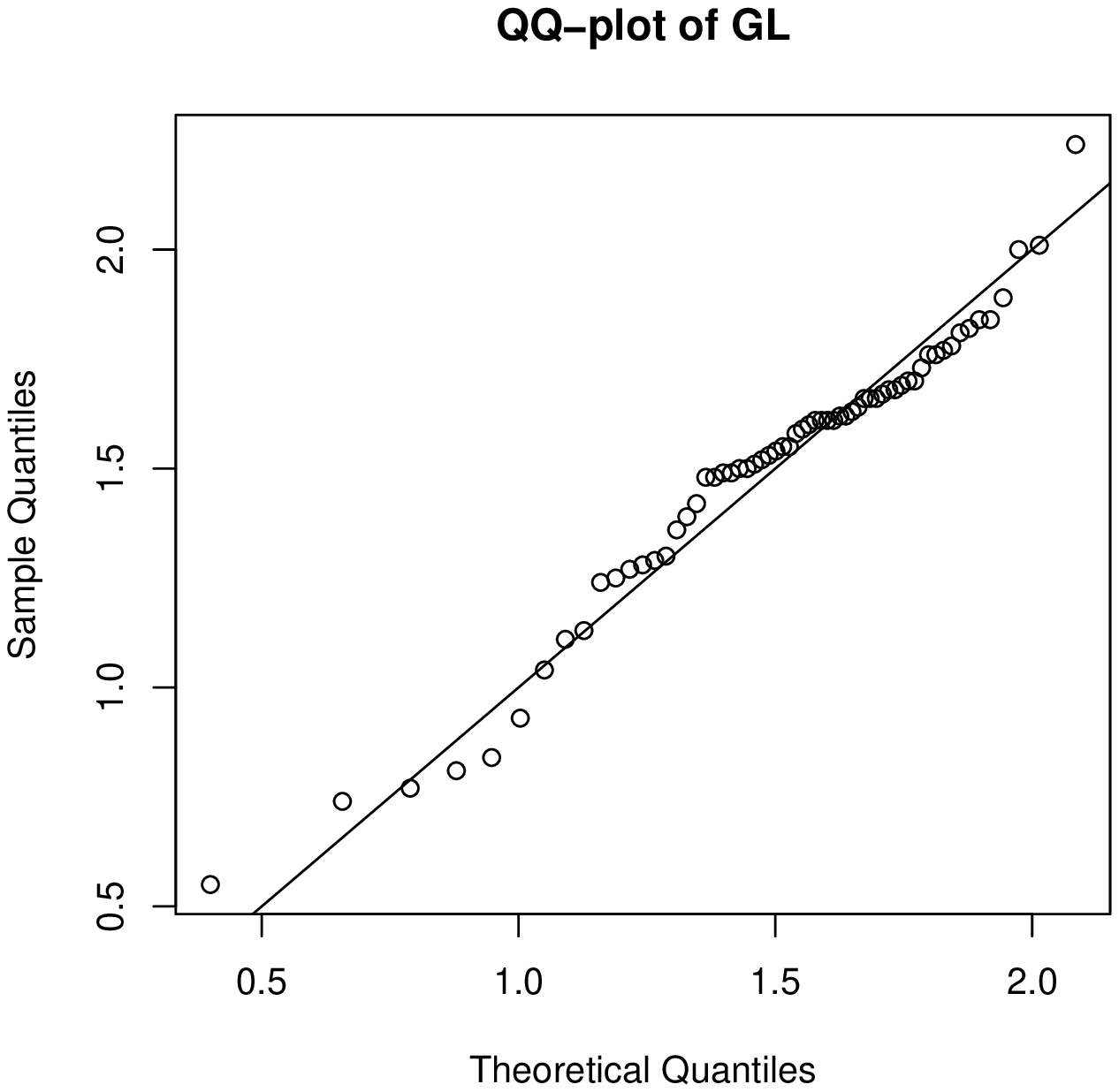}
\includegraphics[scale=0.35]{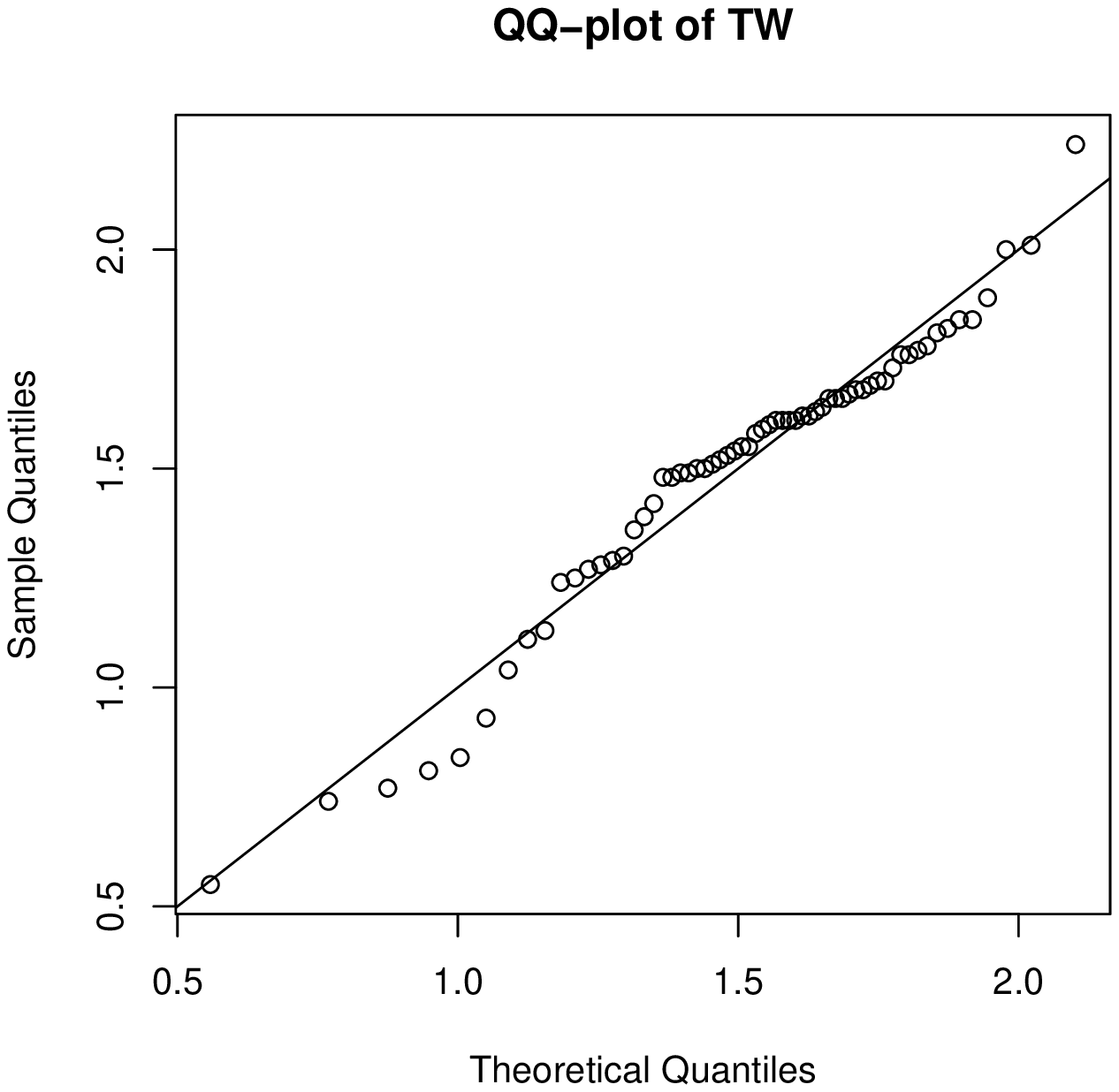}
\vspace{-0.4cm}
\caption[]{\label{fig.qex} QQ plots of the Gompertz, GG, GP, GB, GL and TW models.}
\end{figure}

\begin{table}[h]
\begin{center}
{\small
\caption{Parameter estimates (with std.), K-S statistic,
\textit{p}-value, AIC, AICC and BIC for the data set.}\label{table.EX}
\begin{tabular}{|l|cccccc|} \hline
Dis.           & Gompertz        & GG & GP & GB & GL & TW\\ \hline
$\hat{\beta }$ (std.)  & 0.0088 (0.001)  & 0.8023   (0.772) & 0.0006 (0.001) & 0.0013 (0.001) & 0.0088 (0.011)& -13.9192(-----) \\
$\hat{\gamma }$ (std.) & 3.6474 (0.069)  & 1.3082   (0.586) & 4.4611 (0.566) & 4.2406 (0.404) & 3.6474 (0.593)& 11.8558 (9.795)\\
$\hat{\theta }$ (std.) & ---             & -58.8912 (91.83) & 5.5965 (3.224) & 1.8740 (1.268) & 0.0001 (1.310)& -1.5934 (2.637) \\ \hline
$-\log(L)$             & 14.8081         & 12.2288 & 12.8702 & 13.0212 & 14.8067 & 14.2853 \\
K-S                    & 0.1268          & 0.0962 & 0.1207 & 0.1217 & 0.1267 & 0.0001\\
p-value                & 0.2636          & 0.6040 & 0.3177 & 0.3085 & 0.2636 & 0.9869\\
AIC                    & 33.6162         & 30.4576 & 31.7404 & 32.0424 & 35.6134 & 34.5712 \\
AICC                   & 33.8162         & 30.8644 & 32.1472 & 32.4491 & 36.0202 & 34.9774 \\
BIC                    & 37.9025         & 36.8870 & 38.1698 & 38.4718 & 42.0428 & 40.9999 \\ \hline
\end{tabular}
}
\end{center}
\end{table}

\newpage

\section*{ Appendix }

\subsection*{A.}
Here, we give the proof of Theorems \ref{th.lb}, \ref{th.lg}, and \ref{th.lt}. Consider $p_i=\exp(-\frac{1}{\gamma}(e^{-\beta x_i}-1))$.

\subsubsection*{A.1}
\label{Ap.g1}
Let $w_{1}\left(\beta;\gamma,\theta ,{\boldsymbol x}\right)=\sum^n_{i=1}\frac{\theta p^{\beta}_i \log(p_i) C''(\theta p^{\beta}_i)}{C'(\theta p^{\beta}_i)}=\frac{\partial}{\partial \beta}\sum^n_{i=1}\log (C'(\theta p^{\beta}_i))$. Then, $w_{1}\left(\beta;\gamma,\theta ,{\boldsymbol x}\right)$ is strictly increasing in $\beta $  and
\[{\mathop{\lim }_{\beta \rightarrow 0^+} w_{1}}\left(\beta;\gamma,\theta ,{\boldsymbol x}\right)=\frac{\theta C''(\theta)}{C'(\theta)}\sum^n_{i=1}\log\left(p_i\right),\ \ \ \ \ \ \ \ \ \ {\mathop{\lim}_{\beta \rightarrow \infty}} w_{1}\left(\beta;\gamma,\theta ,{\boldsymbol x}\right)=0.\]
Therefore,
\[{\mathop{\lim }_{\beta \rightarrow 0^+} {{\rm g}}_{1}}\left(\beta;\gamma,\theta ,{\boldsymbol x}\right)=\infty ,\ \ \ \ \ \ \ \ {\mathop{\lim }_{\beta \rightarrow \infty} {\rm g}_1}\left(\beta;\gamma,\theta ,{\boldsymbol x}\right)=\sum^n_{i{\rm =1}}\log\left(p_i\right)<0.\]
Also,
\[{\rm g}_1\left(\beta;\gamma,\theta ,{\boldsymbol x}\right)<\frac{n}{\beta}+\sum^n_{i=1}\log\left(p_i\right), \ \ \ \ \ \ {\rm g}_1\left(\beta;\gamma,\theta ,{\boldsymbol x}\right)>\frac{n}{\beta }+\left(\frac{\theta C''\left(\theta \right)}{C'\left(\theta\right)}+1\right)\sum^n_{i=1}\log\left(p_i\right).\]
Therefore,  ${\rm g}_1\left(\beta;\gamma,\theta ,{\boldsymbol x}\right)<0$ when $\frac{n}{\beta }+\sum^n_{i=1}\log\left(p_i\right)<0$, and ${\rm g}_1\left(\beta;\gamma,\theta ,{\boldsymbol x}\right)>0$ when
$\frac{n}{\beta }+\left(\frac{\theta C''(\theta )}{C'(\theta)}+1\right)\sum^n_{i=1}\log\left(p_i\right)>0$. Hence, the proof is completed.

\subsubsection*{A.2}
\label{Ap.g2}
It can be easily shown that
\[ {\mathop{\lim }_{\gamma\rightarrow 0^+} {\rm g_{2}}(\gamma;\beta,\theta,{\boldsymbol x} )}=n\bar{x}-\frac{\beta }{2}\sum^n_{i=1}{x^2_i\left(1+\frac{\theta e^{-\beta x_i}C''\left(\theta e^{-\beta x_i}\right)}{C'\left(\theta e^{-\beta x_i}\right)}\right)}, \; \ \ \ \
 {\mathop{\lim }_{\gamma\rightarrow +\infty} {\rm g_{2}}(\gamma;\beta,\theta,{\boldsymbol x} )}=-\infty. \]
Since the limits have different signs, the equation ${\rm g_{2}}(\gamma;\beta,\theta,{\boldsymbol x})=0$ has at least one root with respect to $\gamma$ for fixed values
$\beta$ and $\theta$. The proof is completed.

\subsubsection*{A.3}
\label{Ap.g3}
(i) For GP, it is clear that
 \[ {\mathop{\lim }_{\theta \rightarrow 0^+}} {\rm g}_3\left(\theta;\beta,\gamma ,{\boldsymbol x}\right)=\sum\limits_{i=1}^{n}t_{i}-\frac{n}{2},
 \qquad \qquad {\mathop{\lim }_{\theta \rightarrow \infty} }{\rm g}_3\left(\theta;\beta,\gamma ,{\boldsymbol x}\right)=-\infty.\]  Therefore, the equation ${\rm g}_3\left(\theta;\beta,\gamma ,{\boldsymbol x}\right)=0$ has at least one root for $\theta>0$, if $\sum\limits_{i=1}^{n}t_{i}-\frac{n}{2}>0$  or $ \sum\limits_{i=1}^{n}t_{i}>\frac{n}{2}$.\\
(ii) For GG, it is clear that
 \[ {\mathop{\lim }_{\theta \rightarrow \infty} }{\rm g}_3\left(\theta;\beta,\gamma ,{\boldsymbol x}\right)=-\infty, \qquad \qquad
{\mathop{\lim }_{\theta \rightarrow 0^+} }{\rm g}_3\left(\theta;\beta,\gamma ,{\boldsymbol x}\right)=-n+2\sum\limits_{i=1}^{n}t_{i}.\] Therefore, the
equation  ${\rm g}_3\left(\theta;\beta,\gamma ,{\boldsymbol x}\right)=0$ has at least one root for $0<\theta<1$, if  $-n+2\sum\limits_{i=1}^{n}t_{i} >0$ or $\sum\limits_{i=1}^{n}t_{i} >\frac{n}{2}$. \\
(iii) For GL, it is clear that
\[ {\mathop{\lim }_{\theta \rightarrow 0^+} }{\rm g}_3\left(\theta;\beta,\gamma ,{\boldsymbol x}\right)=\sum\limits_{i=1}^{n}t_{i}-\frac{n}{2}, \qquad\qquad
 {\mathop{\lim }_{\theta \rightarrow 1^-} }{\rm g}_3\left(\theta;\beta,\gamma ,{\boldsymbol x}\right)=-\infty.\]
 Therefore, the equation ${\rm g}_3\left(\theta;\beta,\gamma ,{\boldsymbol x}\right)=0$ has at least one root for $0<\theta<1$, if $\sum\limits_{i=1}^{n}t_{i}-\frac{n}{2}>0$  or $ \sum\limits_{i=1}^{n}t_{i}>\frac{n}{2}$.\\
(iv) It is clear that
\[ {\mathop{\lim }_{p \rightarrow 0^+} }{\rm g}_3\left(p;\beta,\gamma ,{\boldsymbol x}\right)=\sum\limits_{i=1}^{n}t_{i}(m-1)-\frac{n(m-1)}{2}, \quad \quad
 {\mathop{\lim }_{p \rightarrow 1^{-}} }{\rm g}_3\left(p;\beta,\gamma ,{\boldsymbol x}\right)=\sum\limits_{i=1}^{n}\frac{-m+1+m t_{i}}{t_{i}}. \]
   Therefore, the equation ${\rm g}_3\left(p;\beta,\gamma ,{\boldsymbol x}\right)=0$ has at least one root for $0<p<1$, if $\sum\limits_{i=1}^{n}t_{i}(m-1)-\frac{n(m-1)}{2}>0$  and  $\sum\limits_{i=1}^{n}\frac{-m+1+m t_{i}}{t_{i}}<0$ or $\sum\limits_{i=1}^{n}t_{i}>\frac{n}{2}$ and $\sum\limits_{i=1}^{n}t_{i}^{-1}>\frac{nm}{1-m}$.

\subsection*{B.}
\label{Ap.IF}
Consider
\begin{eqnarray*}
&&t_{i}=e^{-\frac{\beta}{\gamma}(e^{\gamma x_{i}}-1)}, \qquad b_i=\frac{\partial t_{i}}{\partial \gamma}=t_i d_i,
\qquad d_i=\frac{\partial \log(t_{i})}{\partial \gamma}=\frac{1 }{\gamma}(-\log(t_i)+\gamma x_i\log(t_i)-\beta x_i),\\
&&q_i=\frac{\partial d_{i}}{\partial \gamma}=d_i( x_i-\frac{2}{\gamma})+\frac{x_i}{\gamma}\log(t_i), \qquad A_{2i}=\frac{C''(\theta t_{i})}{C'(\theta t_{i})},
\qquad A_{3i}=\frac{C'''(\theta t_{i})}{C'(\theta t_{i})}.
\end{eqnarray*}
 Then, the elements of $3\times 3$  observed information matrix $I_{n}(\Theta)$  are given by
\begin{eqnarray*}
 I_{\beta\beta}&=& \frac{\partial^{2}l_{n}}{\partial \beta^{2}}=-\frac{n}{\beta^{2}}
 +\frac{\theta}{\beta^{2}}
 \sum\limits_{i=1}^{n}t_i(\log(t_i))^{2}A_{2i}+\frac{\theta^2}{\beta^{2}}
 \sum\limits_{i=1}^{n}t_{i}^{2}(\log(t_i))^{2}A_{3i}
 - \frac{\theta^2}{\beta^{2}}
 \sum\limits_{i=1}^{n}t_{i}^{2}(\log(t_i))^{2}A_{2i}^2,\\
  I_{\beta\gamma}&=& \frac{\partial^{2}l_{n}}{\partial \beta\partial\gamma}=\frac{1}{\beta}\sum\limits_{i=1}^{n}d_i
  +\frac{\theta}{\beta}\sum\limits_{i=1}^{n}b_i\log(t_i)A_{2i}
  +\frac{\theta}{\beta}\sum\limits_{i=1}^{n}b_iA_{2i}\\
  &&+\frac{\theta^2}{\beta}\sum\limits_{i=1}^{n}b_it_i\log(t_i)A_{3i} -\frac{\theta^2}{\beta}\sum\limits_{i=1}^{n} b_it_i\log(t_i)A_{2i}^2,\\
    I_{\beta\theta}&=& \frac{\partial^{2}l_{n}}{\partial \beta \partial \theta}=\frac{1}{\beta}\sum\limits_{i=1}^{n}t_{i}\log(t_i)A_{2i}
    +\frac{\theta}{\beta}\sum\limits_{i=1}^{n}t_{i}^{2}\log(t_i)
  A_{3i}
  -\frac{\theta}{\beta}\sum\limits_{i=1}^{n}t_{i}^{2}\log(t_i)A_{2i}^2,\\
    I_{\gamma\gamma}&=& \frac{\partial^{2}l_{n}}{\partial \gamma^{2}}=
  \sum\limits_{i=1}^{n}q_i+\theta\sum\limits_{i=1}^{n}(b_i d_i+t_i q_i)A_{2i}
  +\theta^2\sum\limits_{i=1}^{n}b_i^2A_{3i}
  -\theta^2\sum\limits_{i=1}^{n}b_i^2A_{2i}^2,\\
 I_{\gamma\theta}&=& \frac{\partial^{2}l_{n}}{\partial \theta \partial \gamma }=\sum\limits_{i=1}^{n}b_i A_{2i}+\theta\sum\limits_{i=1}^{n}t_ib_iA_{3i}
 -\theta\sum\limits_{i=1}^{n}t_ib_iA_{2i}^2
,\\
I_{\theta\theta}&=& \frac{\partial^{2}l_{n}}{\partial \theta^{2}}=-\frac{n}{\theta^{2}}+\sum\limits_{i=1}^{n}t_{i}^{2}
A_{3i}-\sum\limits_{i=1}^{n}t_{i}^{2}A_{2i}^2-\frac{nC''(\theta)}{C(\theta)}+\frac{n(C'(\theta))^{2}}{(C(\theta))^{2}}.
 \end{eqnarray*}

\section*{Acknowledgements}
The authors would like to thank the referees for their comments and suggestions which have contributed to
improving the manuscript.

%\bibliographystyle{apa}
%\bibliography{D:/mypapers/myBIB}

\end{document}